\pdfoutput=1
\documentclass[11pt]{scrartcl}
\usepackage[utf8]{inputenc}
\usepackage{amsmath,amssymb,mathtools,amsthm}
\usepackage{thmtools}
\usepackage[margin=1in]{geometry}
\usepackage{dsfont, lmodern}
\usepackage{booktabs}
\usepackage{enumerate}
\usepackage{tikz}
\usetikzlibrary{cd, calc}
\usetikzlibrary{decorations.pathmorphing}
\usetikzlibrary{decorations.markings}
\usetikzlibrary{fit}
\usepackage{marginnote}
\usepackage{MnSymbol}
\usepackage{stmaryrd}
\usepackage[most]{tcolorbox}
\usepackage{algorithm}
\usepackage{algorithmic}
\usepackage[normalem]{ulem}
\usepackage{cancel,xcolor}
\usepackage{orcidlink}
\usepackage{listings}
\usepackage{xcolor}
\usepackage{enumitem}
\usepackage{quantikz}

\newtheorem{theorem}{Theorem}

\newtheorem{proposition}[theorem]{Proposition}

\newtheorem{lemma}[theorem]{Lemma}
\declaretheorem[style=definition,sibling=theorem,name=Example]{example}
\newtheorem{corollary}[theorem]{Corollary}
\newtheorem*{theorem*}{Theorem}

\newcommand{\supp}[1]{\operatorname{supp}(#1)}
\newcommand{\R}{\mathds R}
\newcommand{\C}{\mathds C}
\newcommand{\kk}{\mathds K}
\newcommand{\conj}[1]{#1^\inv}
\newcommand{\adj}{*}
\newcommand{\xx}{X}

\newcommand{\ncmonoid}{\mathcal W}
\newcommand{\ncalgebra}{\mathcal P}

\newcommand{\inv}{*}
\newcommand{\idealgen}[1]{\langle #1 \rangle}

\newcommand{\lc}[1]{\operatorname{lc}(#1)}
\newcommand{\lm}[1]{\operatorname{lm}(#1)}
\newcommand{\lt}[1]{\operatorname{lt}(#1)}
\newcommand{\clique}{\mathcal C}
\colorlet{summarycolor}{yellow!70!red!90!blue!80!white}
\colorlet{algorithmcolor}{blue!50!white}
\colorlet{softwarecolor}{red!50!white}

\newtcbtheorem[auto counter,number within=section]{summary}%
  {Highlight}{fonttitle=\bfseries\upshape, fontupper=\slshape,
     arc=0mm, colback=blue!5!white,colframe=summarycolor}{theorem}

\newcommand\blueout{\bgroup\markoverwith
{\textcolor{blue}{\rule[.5ex]{2pt}{0.4pt}}}\ULon}
\newcommand{\mathsout}[1]% will draw line through middle of #1
{\bgroup\mathchoice
  {\sbox0{$\displaystyle{#1}$}%
    \usebox0\hspace{-\wd0}%
    \rule[0.5\ht0-0.5\dp0-.5pt]{\wd0}{1pt}}%
  {\sbox0{$\textstyle{#1}$}%
    \usebox0\hspace{-\wd0}%
    \rule[0.5\ht0-0.5\dp0-.5pt]{\wd0}{1pt}}%
  {\sbox0{$\scriptstyle{#1}$}%
    \usebox0\hspace{-\wd0}%
    \rule[0.5\ht0-0.5\dp0-.5pt]{\wd0}{1pt}}%
  {\sbox0{$\scriptscriptstyle{#1}$}%
    \usebox0\hspace{-\wd0}%
    \rule[0.5\ht0-0.5\dp0-.5pt]{\wd0}{1pt}}%
\egroup}

\newcommand{\id}{\mathbf{1}}

\let\epsilon\varepsilon

\newcommand{\graph}{\mathcal G}
\newcommand{\complem}[1]{\bar #1}

\newcommand{\Span}{\operatorname{Span}}
\newcommand{\gnf}{\operatorname{gnf}}
\newcommand{\lex}{\operatorname{lex}}
\newcommand{\foa}{\operatorname{foa}}

\newcommand{\nf}{\operatorname{nf}}

\newcommand{\layer}{{L}}

\title{Partially-Commutative Polynomial Optimization}

\usepackage{authblk}
\hypersetup{colorlinks=false, linkbordercolor=red, citebordercolor=green, urlbordercolor=blue}

\author[1]{Abhishek Mishra}
\author[2]{Mois\'es Bermejo Mor\'an}
\author[1]{Stefano Pironio}

\affil[1]{\normalsize Laboratoire d'Information Quantique, Universit\'e libre de Bruxelles, Belgium}
\affil[2]{\normalsize Department of Computer Science, School of Computing and Data Science, The University of Hong Kong, Hong Kong}
\date{\vspace{-5ex}}

\makeatletter
\def\thmt@trivialref#1#2{%
  \ifcsname r@#1\endcsname
    \expandafter\ifx\csname r@#1\endcsname\relax
      #2%
    \else
      \@xa\@xa\@xa\thmt@trivi@lr@f\csname r@#1\endcsname\relax\@nil
    \fi
  \else
    #2%
  \fi
}
\makeatother

\begin{document}

\maketitle

\begin{abstract}
Semidefinite programming hierarchies for commutative and non-commutative polynomial optimization represent a powerful computational tool with many applications in quantum information. 
In such applications, a given variable is typically not either commuting or non-commuting with all other variables, but instead commutes with some variables and does not commute with others, i.e., the variables satisfy some partial commutation relations. While such partial commutation relations can always be incorporated in a fully non-commutative setting through suitable linear constraints in the semidefinite programming relaxations, exploiting their algebraic properties from the onset can result in more compact relaxations. 
This leads us to introduce partially-commutative polynomial optimization, a framework that encompasses commutative and non-commutative polynomial optimization, allowing for arbitrary commutation relations among the variables. We point out that the underlying algebraic structure is that of a partially-commutative monoid. We present and review several key aspects of such monoids and show how they can be used to build SDP relaxations for partially-commutative polynomial optimization problems in which the partial commutations are natively implemented in the monomial structure, without the need of additional linear constraints.
\end{abstract}

% \tableofcontents  

\section{Introduction}

Non-commutative polynomial optimization (NCPO) refers to optimization problems over matrices or operators in a dimension-free setting. Formally, these operators are represented by non-commutative variables. A standard approach to NCPO is the hierarchy introduced in \cite{navascues2007bounding,navascues2008convergent,doherty2008quantum,pironio2010convergent} commonly known as the NPA hierarchy. It provides a sequence of semidefinite programming (SDP) relaxations, and can be viewed as the non-commutative counterpart of the commutative SDP hierarchies introduced by Lasserre \cite{lasserre2001global} and Parrilo \cite{parrilo2003semidefinite}.

The NPA hierarchy is a central tool in quantum information and quantum foundations. It is used, for example, to upper bound Bell inequality violations and characterize quantum correlations in device-independent settings \cite{brunner2014bell,tavakoli2024semidefinite}, to establish security in device-independent quantum key distribution \cite{brown2024device}, and to certify properties of quantum many-body systems \cite{wang2024certifying}. It is also connected to deep problems in operator algebras, such as Connes' embedding problem \cite{junge2011connes, ji2021mip}. The original framework has been extended in multiple directions, including dimension-bounded NCPO \cite{navascues2015bounding}, tracial NCPO \cite{burgdorf2013tracial}, and NCPO involving trace and state polynomials \cite{klep2022optimization, klep2023state}. Several works also exploit additional problem structure, such as sparsity \cite{wang2021exploiting} and symmetry \cite{bamps2015sum, rosset2018symdpoly}, to reduce SDP size.

The moment-matrix approach to NCPO starts from the observation that any feasible operator solution induces a linear functional on the free $\inv$-algebra of non-commutative polynomials. This functional is determined by its values on monomials -- the moments -- arranged into a moment matrix constrained to be positive semidefinite. Increasing the monomial set yields a hierarchy of SDP relaxations that converges, under suitable assumptions, to the NCPO optimum.

Equality polynomial constraints on the non-commuting variables are usually incorporated as linear constraints on moments. However, more compact SDP relaxations can be obtained by using these equalities to reduce the monomial basis from the outset, rather than imposing linear relations on the full basis. The idea is to identify polynomials that are equivalent modulo the equality constraints and represent moments in the smaller quotient algebra of these equivalence classes. This can be done, e.g., by computing canonical representatives of the equivalence classes.

This idea is not new and has already been suggested in commutative polynomial optimization \cite{parrilo2005exploiting,laurent2008sums} and non-commutative polynomial optimization \cite{pironio2010convergent}. It is used, implicitly or explicitly, in most numerical implementations of the NPA hierarchy. However, a detailed presentation of this idea is missing. Furthermore, careless reductions of the monomial basis can lead to suboptimal or even incorrect relaxations. One objective of this paper is to remedy this situation and to give a detailed presentation of this approach.

Handling arbitrary equality constraints through this quotient approach is computationally hard: one must compute canonical representatives modulo the ideal generated by these constraints. Although general algorithms exist (i.e., through Gr\"obner bases), they can become prohibitively expensive or even fail to terminate. This is related to the fact that the underlying mathematical problem is undecidable.

In this paper, we thus focus on a broad and practically relevant subclass of equality constraints: commutation relations between selected pairs of variables. Such partial commutations naturally arise in quantum information scenarios with multiple parties acting on shared systems. In Bell scenarios, for instance, operators belonging to different parties commute, while operators within one party need not. More generally, in multipartite systems, operators may act on overlapping or disjoint subsets of parties, resulting in a complex commutation structure. This defines a partially commutative regime, interpolating between the fully non-commutative and fully commutative cases.

To formalize this, we introduce \emph{partially commutative polynomial optimization} (PCPO). PCPO generalizes both commutative polynomial optimization (CPO) and NCPO by allowing selective commutation relations among the variables. The relevant underlying algebraic object is a \emph{partially commutative monoid}, also known as a \emph{trace monoid} \cite{cartier1969applications,diekert1995book}. These monoids have been extensively studied, notably in the context of concurrent computations, where non-commuting variables represent processes that are causally dependent, while commutative ones represent independent processes that can be executed in parallel. 
We point out that existing normal-form representations for these monoids can be adapted to build smaller and more efficient SDP relaxations in the PCPO setting that directly account for partial commutations without the need of imposing them as explicit linear constraints.

The paper is organized as follows. Section~\ref{sec:ncpo} reviews standard NCPO and its moment-matrix SDP relaxations. Section~\ref{sec:equality_constraints} explains in detail how equality constraints can be exploited to reduce the monomial basis and build more compact SDP relaxations. We point out in particular in Subsection~\ref{subsec:pitfalls} that computing canonical forms of polynomials through naive substitution rules may fail to identify equivalent polynomials, leading to suboptimal or incorrect SDP relaxations. The correct approach requires a representation of the quotient algebra modulo the equality constraints. We discuss in Subsection~\ref{sec:groebner} how this can in principle be done using Gr\"obner bases, but also the limitations of this approach.
Section~\ref{sec:pcpo} introduces PCPO and links it to partially commutative monoids. We present basic results about the associated Gr\"obner bases in Subsection~\ref{sec:word_representatives} and present more convenient structural representatives of partially-commutative monomials in Subsections~\ref{sec:circuit_representation} and~\ref{subsec:wire_projections}.  We introduce in particular a circuit-based representation that may be more natural for the quantum information community and explain how it can be used to build SDP hierarchies for PCPO. In Section~\ref{sec:tracial_pcpo}, we consider the tracial PCPO setting, where one optimizes traces rather than minimum eigenvalues and we discuss cyclic equivalence in the partially commutative context.

\section{Non-commutative polynomial optimization}
\label{sec:ncpo}
\subsection{Notation and definitions}\label{sec:notation}
Let $\xx$ be an alphabet, i.e., a finite set $\xx=\{x_1,\ldots,x_n\}$ of non-commutative variables or letters. Finite sequences $w=x_{i_1}x_{i_2}\ldots x_{i_k}$ of these variables correpond to monomials or words.  The collection $\ncmonoid=\langle \xx \rangle$ of these words forms a monoid, called the free monoid, where multiplication of two words correponds to their concatenation and the identity element is the empty word denoted $\id$. The length or degree of a monomial $w\in \ncmonoid$, written $|w|$, is the number of letters it contains. The set of all monomials of degree less than or equal to $d$ is denoted $\ncmonoid_d$.

We also define an involution on $\xx$, that is, a map $\inv : \xx \to \xx$, satisfying $(x^\inv)^\inv = x$ for each $x \in \xx$.  A letter $ x \in \xx$ is called hermitian if $ x^\inv = x $. The involution extends to $\ncmonoid$ via $(x_{i_1}x_{i_2} \dots x_{i_k})^\inv = x_{i_k}^\inv \dots x_{i_2}^\inv x_{i_1}^\inv$. The involution serves as an abstract analogue of the adjoint (or Hermitian conjugate) in operator or matrix algebras. 

Let $\kk$ denote the field of real numbers $\R$ or complex numbers $\C$. A polynomial on the alphabet $\xx$ over the field $\kk$ is a finite linear combination of monomials in $\ncmonoid$ with coefficients in $\kk$, 
\begin{equation}
    p = \sum_{w \in \ncmonoid} p_w w,\qquad p_w\in \kk,\quad p_w \neq 0 \text{ for finitely many } w\,.
\end{equation}
The set $\ncalgebra$ of all polynomials on $\xx$ forms an algebra, called the free $*$-algebra, with the natural addition, multiplication and scalar multiplication and involution defined by
\begin{equation}
    p^\inv = \sum_{w \in \ncmonoid} \conj{p_w} w^\inv\,,
\end{equation}
where  $ \conj{p_w} =p_w $ if $ \kk = \R $ and $\conj{p_w}$ is the complex conjugate of $p_w$ if $ \kk = \C $. A polynomial $ p $ is called Hermitian if $ p^\inv = p $.

Give a polynomial $p\in\ncalgebra$, the set $\operatorname{supp}(p) = \{ w : p_w \neq 0\} \subset \ncmonoid$ is called the support of $p$ and its degree $\deg(p)$ is the maximum degree of the monomials in its support: $\deg(p) = \max\{ |w| : w \in \operatorname{supp}(p)\}$. We denote $\ncalgebra_d$ the set of polynomials of degree at most $d$. We may  view $ \ncalgebra $ and $\ncalgebra_d$ as $ \kk $-vector spaces with bases given by $\ncmonoid $ and $\ncmonoid_d $, respectively.
% For a subset $P\subset \ncalgebra$, we denote by $\supp{P}=\bigcup_{p\in P}\supp{p}$ the set of all monomials appearing in the polynomials in $P$.

A monomial ordering $<$ on $ \ncmonoid $ is an order on $\ncmonoid$ with the following properties: \emph{(i)} it is a total order, i.e., for any $u, v \in \ncmonoid$, either $u \leq v$ or $v \leq u$; \emph{(ii)} it is a well order, i.e., every non-empty subset of $\ncmonoid$ has a minimal element; and \emph{(iii)} it is compatible with monomial multiplication, i.e., for any $u, v, l, r \in \ncmonoid$, if $u < v$ then $l u r < l v r$. For a fixed linear order on the letters $\xx$, an example of monomial ordering on $\ncmonoid$ is the degree-lexicographic ordering. This ordering is defined as $u < v$ if \emph{(i)} $ |u| < |v|$ or \emph{(ii)} $|u| = |v|$ and $u_i < v_i$, where $i$ is the index of the first letter in $u$ and $v$ that does not coincide. 

Given a monomial order on $\ncmonoid$, each polynomial $p \in \ncalgebra$ admits a decomposition of the form $p = \sum_{i=1}^n c_i w_i$ where the $w_i$ are ordered decreasingly. In this decomposition, $\lt{p} = c_1 w_1$ is called the leading term of $p$, $\lm{p} = w_1$ the leading monomial, and $\lc{p} = c_1$ the leading coefficient.

\subsection{Operator representations}
A $*$-representation of the non-commutative $*$-algebra $\mathcal{P}$ is a $*$-homomorphism
\begin{equation}
  \pi : \mathcal{P} \to \mathcal{B}(\mathcal{H})\,,
\end{equation}
for some Hilbert space $\mathcal{H}$, where $\mathcal{B}(\mathcal{H})$ is the algebra of bounded operators on $\mathcal{H}$. In other words, $\pi$ is a linear map that preserves multiplication and involution: $\pi(pq) = \pi(p)\pi(q)$ and $\pi(p^\inv) = \pi(p)^\adj$ for each $p, q \in \mathcal{P}$. 

Since $\mathcal{P}$ is generated by the variables in $\xx$, a $*$-representation $\pi$ is completely determined by the images $A_x = \pi(x)$ of the alphabet $\xx$. We can thus identify a $*$-representation with an operator representation $A = \{A_x : x \in \xx\}$ of the variables in $\xx$ that is compatible with the involution: $A_x^\adj = A_{x^\inv}$ for each $x \in \xx$. The operator representation $A$ of the alphabet $\xx$ extends to a $*$-representation of the whole algebra $\ncalgebra$ by associating to each polynomial $p \in \ncalgebra$ the operator $p(A) \in \mathcal{B}(\mathcal{H})$ obtained by substituting each variable $x$ with the corresponding operator $A_x$ and the empty word $\id$ with the identity operator $\mathds 1$. If $p$ is a Hermitian polynomial, then $p(A)$ is a self-adjoint operator and $\langle \psi, p(A) \psi \rangle$ is a real number for any state $\ket{\psi} \in \mathcal{H}$.

\subsection{NCPO formulation}
\label{subsec:ncpo}
A non-commutative polynomial optimization (NCPO) problem is specified by: \emph{(i)} the alphabet $\xx$, the involution on $\xx$, and the field $\kk$, defining the free $*$-algebra $\ncalgebra$ of polyomials; \emph{(ii)} a Hermitian polynomial $p \in \ncalgebra$ corresponding to the objective function; \emph{(iii)} a finite set of polynomials $R \subset \ncalgebra$ representing equality constraints; and \emph{(iv)} a finite set of Hermitian polynomials $Q \subset \ncalgebra$ representing inequality constraints.  The NCPO problem is then \cite{pironio2010convergent}
\begin{tcolorbox}[title=NCPO problem, colframe=gray!60, colback=white, coltitle=black, colbacktitle=gray!60]
  \begin{align}\label{eq:ncpo}
    p^{\mathrm{opt}} = \inf_{A} \inf_{||\psi||=1} \langle \psi,p(A)\psi\rangle \quad \text{s.t. } r(A)=0,\,\forall r\in R, \quad q(A)\geq 0,\, \forall q\in Q
  \end{align}
  where the optimization is over all operator representations $A$ of the variables in $\xx$ and all states (i.e., vectors) $\ket{\psi}$ on the Hilbert space $\mathcal H$ on which $A$ acts.
\end{tcolorbox}

Note that we will assume that the set of equality constraints $R$ is closed under involution, i.e., $r^\inv \in R$ for each $r \in R$. Indeed, since $r(A) = 0$ implies $r^\inv(A) =  \left(r(A)\right)^\inv = 0$ for any operator representation $A$, we can always extend $R$ to be closed under involution without changing the NCPO.

More generally, we could also add constraints of the form $\langle \psi, s(A) \psi \rangle \geq 0$ for some Hermitian polynomials $s$ and $t(A)|\psi\rangle = 0$ for some polynomias $t$. Such constraints are easily dealed with (see \cite{pironio2010convergent}) and the results below are straightforwardly extended to this more general setting. In the following, we do not consider these constraints explicitly for the sake of simplicity.

\subsection{Moment matrices and SDP relaxations}
\label{subsec:sdp_relaxations}

A standard way to relax the NCPO problem is to optimize over linear functionals on the free \(*\)-algebra. Indeed, any feasible operator representation \(A\) and unit vector \(\psi\) define a linear functional
\begin{equation}
L:\ncalgebra\to\kk,\qquad L(f)=\langle \psi,f(A)\psi\rangle.
\end{equation}
This functional satisfies the following properties. It is normalized: \(L(\id)=1\);
 it vanishes on any polynomial obtained by left and right multiplication of the equality constraints with any other polynomials:
    \begin{equation}
    L(s r t)=0 \qquad \forall r\in R,\ \forall s,t\in \ncalgebra;
    \end{equation}
it is nonnegative on all Hermitian squares localized by the inequality constraints:
    \begin{equation}\label{eq:local_positivity}
    L(s^\inv q s)\ge 0 \qquad \forall q\in \{\id\}\cup Q,\ \forall s\in \ncalgebra.
    \end{equation}
This leads to the following infinite-dimensional relaxation of the NCPO problem:
\begin{tcolorbox}[title=Linear functional relaxation of NCPO, colframe=gray!60, colback=white, coltitle=black, colbacktitle=gray!60]
\begin{align}
\label{eq:inf_lin_relax}
\hat p^{\mathrm{opt}}=\inf_L \quad & L(p)\\
\text{s.t.}\quad
& L(\id)=1, \nonumber\\
& L(s^\inv q s)\geq 0 \qquad \forall q\in \{\id\}\cup Q, \forall s\in\ncalgebra,\nonumber\\
& L(s r t)=0 \qquad \forall r\in R,\ \forall s,t\in \ncalgebra, \nonumber
\end{align}
where the optimization is over all linear functionals \(L:\ncalgebra\to\kk\).
\end{tcolorbox}
 Any feasible pair \((A,\psi)\) of the NCPO problem yields a feasible solution to \eqref{eq:inf_lin_relax}, and therefore $\hat p^{\mathrm{opt}}\le p^{\mathrm{opt}}$. Under suitable representation assumptions (e.g. archimedeanity), this relaxation is exact \cite{navascues2008convergent,pironio2010convergent}.

By linearity, \(L\) is completely determined by its values on the monomials. We may thus identify \(L\) with the infinite sequence of moments
\begin{equation}
y_w:=L(w),\qquad w\in \ncmonoid,
\end{equation}
so that, for any polynomial \(s=\sum_w s_w w\), one has $L(s)=\sum_w s_w\, y_w$. 
Hence, \eqref{eq:inf_lin_relax} can equivalently be viewed as an optimization problem over the moments \(y_w\).

The formulation \eqref{eq:inf_lin_relax} is infinite-dimensional. To obtain a tractable relaxation, we restrict \(L\) to a finite-dimensional subspace of \(\ncalgebra\). For this, choose a finite set $B\subset\ncmonoid$ of monomials and
consider linear functionals
\begin{equation}
L:\mathcal{P}_{B^\inv B}\to \kk,
\end{equation}
restricted to the subspace $\mathcal{P}_{B^\inv B}=\Span(B^\inv B)$ where $B^\inv B = \{u^\inv v:\; u,v\in B\}$. In this way, we can relax the condition \eqref{eq:local_positivity} for $q=\id$ to the condition
\begin{equation}
    L(s^\inv s)\ge 0 \qquad \forall s\in\Span(B),
\end{equation}
and $L(s^\inv s)$ is well-defined since then $s^\inv s\in \Span(B^*B)=\mathcal{P}_{B^\inv B}$ belongs to the domain of $L$.
Similarly, for an arbitrary polynomial $q$, define $B_q\subset \ncmonoid$ as the largest subset of monomials such that
\begin{equation}
\{u^\inv w v:\; u,v\in B_q,\ w\in \supp{q}\}\subseteq B^\inv B,
\end{equation}
where for $q=\id$, we have $B_{\id}=B$. Then, the condition \eqref{eq:local_positivity} for an arbitary $q$ can be relaxed to 
\begin{equation}\label{eq:local_positivity_truncated}
    L(s^\inv q s)\ge 0 \qquad \forall s\in\Span(B_q),
\end{equation}
and again $L(s^\inv q s)$ is well-defined since then $s^\inv q s\in \mathcal{P}_{B^\inv B}$. 
Finally, the linear constraints $L(s r t)=0$ are only retained for $s r t\in \mathcal{P}_{B^\inv B}$.

We can further notice that for each $q\in \{\id\}\cup Q$, the set of conditions \eqref{eq:local_positivity_truncated} is equivalent to the positive semidefiniteness of the matrix $M(q,L)$ whose rows and colums are indexed by $B_q$ and whose entries are defined by
\begin{equation}
M_B(q,L)(u,v)=L(u^\inv q v),\qquad u,v\in B_q\,.
\end{equation}
Indeed, for any polynomial $s=\sum_{w\in B_q} s_w w\in \Span(B_q)$, one has
\begin{equation}
L(s^\inv q s)=\sum_{u,v\in B_q}\overline{s_u}s_v\, M(q,L)(u,v)\geq 0.
\end{equation}
When \(q=\id\), the matrix \(M_B(\id,L):=M_B(L)\) is called the moment matrix; for \(q\in Q\), the matrices \(M_B(q,L)\) are called the localizing matrices.

The truncated version of \eqref{eq:inf_lin_relax} associated to the subset of monomials $B\subset \ncmonoid$ (where we assume that $p\in\mathcal{P}_{B^\inv B}$ so that the objective function is well-defined) is then
\begin{tcolorbox}[title=SDP relaxation of NCPO with truncated basis $B\subset \ncmonoid$, colframe=gray!60, colback=white, coltitle=black, colbacktitle=gray!60]
\begin{align}
\label{eq:truncated_lin_relax}
\hat p_B^{\mathrm{opt}}=\inf_L \quad & L(p)\\
\text{s.t.}\quad
& L(\id)=1, \nonumber\\
& M_B(q,L)\succeq 0 \qquad \forall q\in \{\id\}\cup Q, \nonumber\\
& L(urv)=0 \qquad \forall r\in R,\ \forall u,v\in B^\inv B
\text{ such that } \supp{urv}\subseteq B^\inv B, \nonumber
\end{align}
where the optimization is over all linear functionals \(L:\mathcal{P}_{B^\inv B}\to\kk\).
\end{tcolorbox}
By construction, \eqref{eq:truncated_lin_relax} is a finite-dimensional relaxation of \eqref{eq:inf_lin_relax}, and therefore also of the original NCPO problem.
This problem is an SDP, which can equivalently be written in terms of the truncated moment sequence $y=(y_w)_{w\in B^\inv B}$, fully characterizing the linear functional $L$ on $\mathcal{P}_{B^\inv B}$.

Given a nested sequence of finite monomial sets $B_1\subseteq B_2\subseteq \cdots \subseteq \ncmonoid$, $\bigcup_{k\ge 1} B_k=\ncmonoid$, one obtains a hierarchy of SDP relaxations satisfying $\hat p_{B_1}^{\mathrm{opt}}\le \hat p_{B_2}^{\mathrm{opt}}\le \cdots \le p^{\mathrm{opt}}$. 
Under suitable assumptions, this hierarchy converges to the optimum \(p^{\mathrm{opt}}\) \cite{navascues2008convergent,pironio2010convergent}.
The usual degree-truncated relaxation hierarchy is recovered by taking $B_d=\ncmonoid_{d}$ as the set of monomials of degree at most \(d\); in that case, $B_{d,q}=\ncmonoid_{d-\lceil \deg(q)/2\rceil}$.

\section{Equality constraints and quotient reduction}\label{sec:equality_constraints}
\subsection{Motivating example}
While encoding each equality constraint $r$ as linear conditions $L(urv)=0$ on the moments is straightforward, a more efficient approach is to use these constraints to reduce the monomial basis from the outset, before constructing the moment matrix.

As an illustration, consider the following standard example, corresponding to the Bell--CHSH scenario in quantum information.

\begin{example}[Bell--CHSH]
\label{example:chsh}
This NCPO instance is defined by the alphabet $\xx=\{a_0,a_1,b_0,b_1\}$, with involution $x^\inv=x$ for each $x\in \xx$ (that is, all variables are Hermitian), the field\footnote{If all polynomials defining the NCPO have real coefficients, the problem is invariant under complex conjugation, and one may restrict w.l.o.g. to $\kk=\R$ \cite{navascues2008convergent}.} $\kk=\R$, the objective polynomial $p = a_0b_0+a_0b_1+a_1b_0-a_1b_1$, and the set of equality constraint
\begin{align}
R &= \{a_i^2-1 : i=0,1\}
         \;\cup\;
         \{b_j^2-1 : j=0,1\}
         \;\cup\;
         \{b_j a_i-a_i b_j : i,j=0,1\}.
\end{align}
\end{example}
The corresponding NCPO is
\begin{equation}\label{eq:chsh_ncpo}
\begin{aligned}
p^{\mathrm{opt}}
  = \sup_{A,B}\ \sup_{\|\psi\|=1}\ 
  &\langle \psi,\,(A_0B_0+A_0B_1+A_1B_0-A_1B_1)\psi\rangle \\
\text{s.t.}\quad
  &A_i^2=\mathds{1}\qquad (i=0,1), \\
  &B_j^2=\mathds{1}\qquad (j=0,1), \\
  &B_jA_i=A_iB_j \qquad (i,j=0,1),
\end{aligned}
\end{equation}
where \(A_i\) and \(B_j\) are Hermitian operators corresponding to the variables \(a_i\) and \(b_j\), respectively.

Consider the degree-2 SDP relaxation based on
\begin{align*}
B =\ncmonoid_2
= \{&
\id,\,
a_0,\,
a_1,\,
b_0,\,
b_1,\,
a_0^2,\,
a_0a_1,\,
a_0b_0,\,
a_0b_1,\,
a_1a_0,\,
a_1^2,\,
a_1b_0,\,
a_1b_1,\\
&
b_0a_0,\,
b_0a_1,\,
b_0^2,\,
b_0b_1,\,
b_1a_0,\,
b_1a_1,\,
b_1b_0,\,
b_1^2
\},
\end{align*}
which consists of the $21$ monomials of degree at most 2. The corresponding retained moments are indexed by $B^\inv B=\ncmonoid_4$, which contains $341$ monomials of degree at most 4. Hence the relaxation involves $341$ moments $y=\left(y_w\right)_{w\in \ncmonoid_4}$, defining a $21\times 21$ moment matrix $M(y)$. 
Since all variables are real and $M(y)$ is, $y_{u^*v} = y_{v^*u}$, only $231$ of these moments are independent.

In addition, the equality constrains generate linear conditions $L(urv)=0$ for $r\in R$ and $u,v$ such that $\supp{urv}\subseteq \ncmonoid_4$, In this example, this yields $456$ linear constraints, of which 200 are linearly independent. Thus, the degree-2 SDP relaxation involves \(231\) scalar variables, one semidefinite constraint of size \(21\times 21\), and \(200\) independent linear constraints.

A more effective formulation is obtained by realizing that the equality constraints $R$ effectively induce equivalence relations among the polynomials, since  they lead us to identify in \eqref{eq:chsh_ncpo} $A_i^2$ and $B_j^2$ with $\mathds{1}$, and $B_jA_i$ with $A_iB_j$. We can thus view them as substitution rules:
\begin{equation}\label{eq:chsh_rules}
    S = \{a_i^2 \mapsto 1 : i=0,1\} \cup \{b_j^2 \mapsto 1 : j=0,1\} \cup \{b_j a_i \mapsto a_i b_j : i,j=0,1\}.
\end{equation}
These can be applied recursively to reduce any polynomial expression to a unique linear combination of canonical monomials -- those that can no longer be simplified using the rules in $S$. For instance,
\begin{equation}
a_0^2 b_0 a_1 \xrightarrow{a_0^2 \mapsto 1} b_0 a_1 \xrightarrow{b_0 a_1 \mapsto a_1 b_0} a_1 b_0.
\end{equation} 

Applying these rules to $\ncmonoid_2$ reduces the original set of $21$ monomials to 13 canonical monomials:
\begin{equation}
[\ncmonoid_2]=\{1, a_0, a_1, b_0, b_1, a_0 a_1, a_1 a_0, b_0 b_1, b_1 b_0, a_0 b_0, a_0 b_1, a_1 b_0, a_1 b_1\}.
\end{equation}
Similarly $\ncmonoid_4$ reduces to a set $[\ncmonoid_4]$ of 41 canonical monomials. The degree-2 SDP relaxation may therefore be written directly in terms of the moments $y=(y_w)_{w\in[\ncmonoid_4]}$, with a moment matrix indexed by $[\ncmonoid_2]$. This yields a $13 \times 13$ moment matrix involving $41$ retained moments, or $31$ independent scalar variables after the symmetric character of the moment matrix is taken into account. In this formulation, there are no explicit linear constraints left: they have been absorbed into the choice of canonical monomials.

Thus, the reduced formulation involves \(31\) scalar variables and one semidefinite constraint of size \(13\times 13\), instead of \(231\) scalar variables, one semidefinite constraint of size \(21\times 21\), and \(200\) additional linear constraints. This leads to a substantial computational gain and often also improves numerical stability.

The improvement becomes even more pronounced at higher levels of the hierarchy. In the CHSH example, the number of monomials in $\ncmonoid_d$ grows exponentially with $d$ as $(4^{d+1} - 1)/{3}$. By contrast, after reduction using the rules~\eqref{eq:chsh_rules}, the number of canonical monomials grows only quadratically with $d$ as $2d^2 + 2d + 1$. 
  
\subsection{Pitfall of naive substitutions}\label{subsec:pitfalls}
This approach described above is part of the folklore  of SDP implementations for NPCO. For instance, the Python package \texttt{ncpol2sdpa}~\cite{wittek2015algorithm} allows equality constraints to be handled in two ways: either by adding linear constraints among the moments or by turning them into substitution rules. In the later case, each polynomial
\begin{equation}
  r = c_1w_1 + r'
\end{equation}
where $w_1=\lm{r}$ is the leading monomial of $r$ and $r'$ the remaining terms, is turned into the substitution rule
\begin{equation}
  w_1 \mapsto - \frac{r'}{c_1},
\end{equation}
as in the CHSH example above. In that case, the resulting rewriting system fully reduces every monomial to a unique canonical representative. However, this is not true in general, as the following examples show.

\begin{example}\label{example:abc}  $\xx = \{a,b,c\}$ and $R = \{ba - ab,\; cb - bc\}$.

  These constraints imply that $b$ commutes with both $a$ and $c$. In particular, the  monomials $cab$ and $bca$ are equivalent, since $b$ can be moved freely across $a$ and $c$. However, if one naively interpret $R$ as the substitution rules $S=\{ba \mapsto ab, cb \mapsto bc\}$, this fails to identify this equivalence, as neither neither $cab$ nor $bca$ can be reduced, since neither contains $ba$ or $cb$ as a subword.
\end{example}

\begin{example}\label{example:ab} $\xx = \{a,b\}$ and $R = \{ab - a,\; ba - b\}$.

Interpreting $R$ as the substitution rules $ab \mapsto a$ and $ba \mapsto b$ leads to ambiguity. For instance, the monomial $aba$ can be simplified in two different ways: \begin{align}
    aba &\xrightarrow{ab \mapsto a} a^2\, , \nonumber\\
    aba &\xrightarrow{ba \mapsto b} ab \xrightarrow{ab\mapsto a }a\,.\label{eq:ex3}
\end{align}
Thus the same monomial reduces to two different outcomes, depending on the order of rule application. 
\end{example}

\begin{example}\label{example:ababa}
  $\xx = \{a,b,c\}$ and $R = \{aba - c\}$.
    
    Interpreting $R$ as the substitution rule $aba \mapsto c$ again leads to ambiguity: 
    \begin{align}
    ababa = (aba)ba &\xrightarrow{aba \mapsto c} cba \nonumber \\
    ababa = ab(aba) &\xrightarrow{aba \mapsto c} abc \label{eq:ex4}
\end{align}
\end{example}

These examples show that the folklore approach of turning each equality constraint into a substitution rule may fail to identify equivalent monomials, and may even lead to ambiguous reductions. For example, the package \texttt{ncpol2sdpa} would treat the monomials \(cab\) and \(bca\) in Example~\ref{example:abc} as distinct, leading to two different moment variables \(y_{cab}\) and \(y_{bca}\), even though they should be identified. The resulting SDP is still a valid relaxation, but \emph{(i)} it is generally less efficient than it could be, since it uses more variables and larger moment matrices than needed, and \emph{(ii)} it is a loser relaxation, which may lead to a suboptimal solution.

One of the objectives of the present paper is to draw attention to the above issues and to place on a rigorous algebraic footing the idea of using equality constraints to reduce the size of SDP relaxations for NCPO. We start in the next subsection by describing the algebraic structure underlying this idea. We will later on discuss algorithmic approaches to implement it in practice in a way that avoids the pitfalls illustrated above.

\subsection{Quotient reduction}
The equality constraints \(R\) in the NCPO problem generate the two-sided \(*\)-ideal \(I_R\) in \(\ncalgebra\) defined as the set
\begin{equation}
    I_R = \idealgen{R}=\{\sum_{i} g_i r_i h_i : r_i \in R, g_i, h_i \in \ncalgebra\}.
\end{equation}
This is the smallest subset of \(\ncalgebra\) that contains \(R\) and is closed under addition, multiplication on both sides by arbitrary polynomials, and involution\footnote{We remember that we assume w.l.o.g. that $R$ is closed under involution, i.e., both $r$ and $r^*$ are in $R$.}. The ideal \(I_R\) captures all the algebraic consequences of the equality constraints \(R\). If \((A,\psi)\) is feasible for the NCPO problem, then
\begin{equation}
r(A)=0 \qquad \forall r\in R,
\end{equation}
and therefore
\begin{equation}
f(A)=0 \qquad \forall f\in I_R.
\end{equation}
In particular, the associated linear functional $L: \ncalgebra \to \kk$ vanishes on $I_R$:
\begin{equation}
    L(f)=\langle \psi,f(A)\psi\rangle=0 \qquad \forall f\in I_R.
\end{equation}

The ideal $I_R$ defines an equivalence relation on $\ncalgebra$: two polynomials $f$ and $g$ are equivalent if their difference $f-g$ lies in $I_R$, i.e., if we can write
\begin{equation}
    f = g + s \qquad \text{for some } s \in I_R.
\end{equation}
The equivalence class of a polynomial $f$ is denoted by $[f]$, and the set of all equivalence classes forms the quotient algebra $\ncalgebra/I_R$. The quotient algebra inherits the structure of a $*$-algebra from $\ncalgebra$, with addition, multiplication, scalar multiplication, and involution defined by \begin{align}
    [f] + [g] &= [f+g], \\
    [f][g] &= [fg], \\
    \alpha [f] &= [\alpha f], \\
    [f]^\inv &= [f^\inv].
\end{align}

Given an operator representation $A$, the associated linear functional $L$, since it vanishes on $I_R$, only depends on the equivalence class of a polynomial modulo $I_R$: $L(f) = L(g)$ if $[f] = [g]$. Thus, $L$ factors through the quotient algebra $\ncalgebra/I_R$.
This suggests replacing the original optimization over \(\ncalgebra\) by an optimization over the quotient \(\ncalgebra/I_R\). We thus consider linear functionals
\begin{equation}
L:\ncalgebra/I_R \to \kk.
\end{equation}

The infinite-dimensional quotient version of \eqref{eq:inf_lin_relax} is then
\begin{tcolorbox}[title=Linear functional relaxation of NCPO modulo $I_R$, colframe=gray!60, colback=white, coltitle=black, colbacktitle=gray!60]
\begin{align}
\label{eq:quotient_inf_relax}
\hat p_{/I_R}^{\mathrm{opt}}=\inf_{L}\quad & L([p])\\
\text{s.t.}\quad
& L([\id])=1, \nonumber\\
& L([s]^\inv [q][s])\ge 0
\qquad \forall q\in \{\id\}\cup Q,\ \forall [s]\in \ncalgebra/I_R, \nonumber
\end{align}
where the optimization is over all linear functionals \(L:\ncalgebra/I_R\to\kk\).
\end{tcolorbox}
This problem does not contain any explicit linear constraints, except the normalization constraint $L([\id])=1$,  since the equality constraints in $R$ have been absorbed into the quotient structure.

The quotient formulation \eqref{eq:quotient_inf_relax} on \(\ncalgebra/I_R\) is equivalent to the formulation \eqref{eq:inf_lin_relax} on \(\ncalgebra\) with explicit equality constraints. Indeed, if a linear functional \(\mathsf{L}:\ncalgebra\to\kk\) on the original space satisfies
\begin{equation}
\mathsf{L}(srt)=0 \qquad \forall r\in R,\ \forall s,t\in \ncalgebra,
\end{equation}
then \(\mathsf{L}\) vanishes on \(I_R\) and therefore induces a unique linear functional
\begin{equation}
L:\ncalgebra/I_R\to\kk,\qquad L([f]):=\mathsf{L}(f).
\end{equation}
Conversely, any linear functional \( L\) on \(\ncalgebra/I_R\) defines a linear functional by
\begin{equation}
  \mathsf{L}: \ncalgebra\to\kk,\qquad \mathsf{L}(f): L([f]),
\end{equation}
which automatically vanishes on \(I_R\), hence satisfies all equality constraints in \eqref{eq:inf_lin_relax}.
Moreover, for every \(q\in \{\id\}\cup Q\) and every \(s\in \ncalgebra\), one has
\begin{equation}  
L([s]^\inv [q][s])
=
L([s^\inv q s])
=
\mathsf{L}(s^\inv q s),
\end{equation}
so that the positivity constraints in \eqref{eq:quotient_inf_relax} follow from those in \eqref{eq:inf_lin_relax} under the map $\mathsf{L}\rightarrow L$ and conversely the positivity constraints in \eqref{eq:inf_lin_relax} follow from those in \eqref{eq:quotient_inf_relax} under the map $L\rightarrow \mathsf{L}$.
Hence the two formulations have the same feasible set and the same optimal solution.

\subsection{Finite-dimensional truncation}
As in the free-algebra formulation, the quotient formulation \eqref{eq:quotient_inf_relax} is generally infinite-dimensional and can be truncated to obtain a tractable SDP relaxation. Specifically, one may restrict the linear functional \(L\) to a finite-dimensional subspace of \(\ncalgebra/I_R\), and only retain those positivity constraints that can be expressed within that truncation space, which again can be expressed as the positive semidefiniteness of a finite number of moment and localizing matrices. This leads to a hierarchy of finite-dimensional SDP relaxations, which can be viewed as quotient-based analogs of the usual free-algebra-based relaxations. The quotient formulation inherits the convergence properties of the free-algebra formulation, and therefore also converges to the optimum \(p^{\mathrm{opt}}\) under the same assumptions.

In particular, every finite SDP relaxation \eqref{eq:truncated_lin_relax} in the free-algebra formulation gives rise to a corresponding quotient-based SDP relaxation. Let \(B\subseteq \ncmonoid\) be the finite set of monomials defining the relaxation \eqref{eq:truncated_lin_relax}. Recall that the truncated functional \(L\) is then defined on the subspace
$\ncalgebra_{B^\inv B}:=\Span(B^\inv B)$
and is characterized by the finite sequence of moments $y=(y_w)_{w\in B^\inv B}$

Passing to the quotient, the corresponding truncated functional \(L\) is defined on the quotient space
$[\ncalgebra_{B^\inv B}]
\subseteq \ncalgebra/I_R$, generated by the set
$
[B^\inv B]:=\{[w]:\, w\in B^\inv B\}.
$
Note that the set \([B^\inv B]\) typically does not form a basis of \([\ncalgebra_{B^\inv B}]\), since the classes \([w]\), \(w\in B^\inv B\), may be linearly dependent in \(\ncalgebra/I_R\). We therefore choose a basis
$
W_{B^\inv B}
$
of \([\ncalgebra_{B^\inv B}]\). The quotient functional \(L\) is then characterized by the retained quotient moments
$
y_{[w]}:=L([w]),
[w]\in W_{B^\inv B}.
$
Typically,
$
|W_{B^\inv B}|<|B^\inv B|,
$
so that the quotient formulation involves fewer moment variables than the original one.

Similarly, for each \(q\in \{\id\}\cup Q\), consider the subspace
$
[\ncalgebra_{B_q}]
\subseteq [\ncalgebra_{B^\inv B}],
$ 
where \(B_q\) is the index set used in the free-algebra formulation. Since
$
[\ncalgebra_{B_q}] \subseteq [\ncalgebra_{B^\inv B}],
$
we may choose the basis \(W_{B^\inv B}\) so that it extends a basis
$
W_{B_q}
$
of \([\ncalgebra_{B_q}]\). In particular, we may assume
$
W_{B_q}\subseteq W_{B^\inv B}.
$

We then define the quotient moment and localizing matrices by
\begin{equation}
M_{B/I_R}(q,L)([u],[v])
=
L([u]^\inv [q][v]),
\qquad
[u],[v]\in W_{B_q}.
\end{equation}

The quotient-based version of \eqref{eq:truncated_lin_relax} is therefore
\begin{tcolorbox}[title={SDP relaxation of NCPO with truncated basis $B\subset \ncmonoid$, modulo $I_R$}, colframe=gray!60, colback=white, coltitle=black, colbacktitle=gray!60]
\begin{align}
\label{eq:quotient_truncated_relax}
\hat p_{B/I_R}^{\mathrm{opt}}=\inf_{L}\quad & L([p])\\
\text{s.t.}\quad
& L([\id])=1, \nonumber\\
& M_{B/I_R}(q,L)\succeq 0
\qquad \forall q\in \{\id\}\cup Q, \nonumber
\end{align}
where the optimization is over all linear functionals
$
L:[\ncalgebra_{B^\inv B}]\to \kk.
$
\end{tcolorbox}
This relaxation can equivalently be written as an SDP over the quotient moments
$
y=(y_{[w]})_{[w]\in W_{B^\inv B}}.
$
By construction, it contains no explicit linear constraints associated with the equalities \(R\), since these have been absorbed into the quotient structure.

It is easy to see that the quotient-based relaxation is at least as strong as the original truncated relaxation, i.e., that
\begin{equation}\label{eq:relaxation_comparison}
\hat p_B^{\mathrm{opt}} \le \hat p_{B/I_R}^{\mathrm{opt}}.
\end{equation}
Indeed, any feasible quotient functional
$
L:[\ncalgebra_{B^\inv B}] \to \kk
$
induces a pullback linear functional $
\mathsf{L}:\ncalgebra_{B^\inv B}\to \kk,
\mathsf{L}(f):= L([f]).
$
This pullback \(\mathsf{L}\) is feasible for the original truncated relaxation \eqref{eq:truncated_lin_relax}: the normalization is preserved, the positivity constraints follow from the positive semidefiniteness of the quotient moment and localizing matrices, and the equality constraints hold automatically since every polynomial in the ideal \(I_R\) has zero class in \(\ncalgebra/I_R\). Moreover,
$
\mathsf{L}(p)=L([p]).
$
Hence every feasible solution of the quotient-based relaxation yields a feasible solution of the original truncated relaxation with the same objective value, from which \eqref{eq:relaxation_comparison} follows.

The relation between the two formulations becomes particularly transparent at the matrix level. Let \(q\in \{\id\}\cup Q\), and remind that $W_{B_q}=\{[w]_1,\dots,[w]_r\}$ is a basis of \([\ncalgebra_{B_q}]\). For each \(u\in B_q\), its class \([u]\in [\ncalgebra_{B_q}]\) admits a unique expansion
\begin{equation}
[u]=\sum_{i=1}^r Z_q(i,u)\,[w]_i.
\end{equation}
This defines a matrix
\begin{equation}
Z_q\in \kk^{r\times |B_q|},
\end{equation}
whose \(u\)-th column is the coordinate vector of \([u]\) in the basis \(W_{B_q}\).

Let \(L\) be a quotient feasible solution, and let \(\mathsf{L}(f)=L([f])\) be its pullback to \(\ncalgebra_{B^\inv B}\). Then, for any \(u,v\in B_q\),
\begin{align*}
M_B(q,\mathsf{L})(u,v)
&= \mathsf{L}(u^\inv q v) \\
&= L([u]^\inv [q][v]) \\
&= L\!\left(
\Big(\sum_{i=1}^r \overline{Z_q(i,u)}\,[w]_i^\inv\Big)
[q]
\Big(\sum_{j=1}^r Z_q(j,v)\,[w]_j\Big)
\right) \\
&= \sum_{i,j=1}^r \overline{Z_q(i,u)}\, Z_q(j,v)\,
M_{B/I_R}(q,L)([w]_i,[w]_j).
\end{align*}
Hence,
\begin{equation}
M_B(q,\mathsf{L})= Z_q^\ast\, M_{B/I_R}(q,L)\, Z_q.
\end{equation}
Therefore, the positivity of the moment matrix \(M_B(q,\mathsf{L})\) follows from the positivity of the smaller matrix \(M_{B/I_R}(q,L)\) because of the above factorization.

If the basis \(W_{B_q}\) is chosen as the classes of a subset of monomials in \(B_q\), say
\begin{equation}
W_{B_q}=\{[u]:\, u\in R_q\}
\qquad
\text{for some } R_q\subseteq B_q,
\end{equation}
then, after reordering the indices so that the monomials in \(R_q\) come first, the matrix \(Z_q\) takes the block form
\begin{equation}
Z_q=
\begin{pmatrix}
I & Z'_q
\end{pmatrix},
\end{equation}
and the relation above becomes
\begin{equation}
M_B(q,\mathsf{L})=
\begin{pmatrix}
M' & M'Z'_q\\
{Z'_q}  ^\ast M' & {Z'_q}^\ast M' Z'_q
\end{pmatrix},
\end{equation}
where \(M'=M_{B/I_R}(q,L)\). In this case, the quotient matrix is literally a principal submatrix of the original one, and all remaining rows and columns are linear combinations of those indexed by the chosen representatives. Quotient reduction thus amounts to removing redundant rows and columns from the moment and localizing matrices.

At first sight, one might expect the quotient-based relaxation \eqref{eq:quotient_truncated_relax} and the free-algebra relaxation \eqref{eq:truncated_lin_relax} to be equivalent, since both encode the equalities \(R\). However, the two formulations differ in a crucial way.

In the free-algebra formulation, the equality constraints are imposed only through the truncated relations $\mathsf{L}(urv)=0$ for $\supp{urv}\subseteq B^\inv B$. Letting
\begin{equation}
J_B:=\Span\{urv:\ r\in R,\ \supp{urv}\subseteq B^\inv B\}\subseteq \ncalgebra_{B^\inv B},
\end{equation}
the feasible truncated functionals are exactly those with \(\mathsf{L}|_{J_B}=0\), i.e., they factor through \(\ncalgebra_{B^\inv B}/J_B\). As above, one may project onto this subspace, removing the explicit equality constraints and compressing the moment and localizing matrices via factorizations of the form \(M_B(q,\mathsf{L})=Z_q^\ast M'_q Z_q\) as above.

By contrast, the quotient formulation works modulo the full ideal \(I_R\) generated by \(R\). The quotient space $[\ncalgebra_{B^\inv B}]  \cong \ncalgebra_{B^\inv B}/I_R=\ncalgebra_{B^\inv B}/(I_R\cap \ncalgebra_{B^\inv B})$
identifies two truncated polynomials whenever they differ by an element of \(I_R\), even if this identification cannot be derived using only polynomials supported in \(B^\inv B\). Thus, we have
\begin{equation}
J_B \subseteq I_R\cap \ncalgebra_{B^\inv B}
\end{equation}
and the inclusion can possibly be strict. Since the quotient \(\ncalgebra_{B^\inv B}/(I_R\cap \ncalgebra_{B^\inv B})\) is in general smaller than \(\ncalgebra_{B^\inv B}/J_B\), the former formulation may identify as equal more moment variables and may yield a tighter relaxation. It follows that \(\hat p_B^{\mathrm{opt}} \le \hat p_{B/I_R}^{\mathrm{opt}}\) may hold with strict inequality. Equivalence is guaranteed to hold when \(J_B = I_R\cap \ncalgebra_{B^\inv B}\), i.e., when the truncated relations $\mathsf{L}(urv)=0$ already capture all consequences of the ideal on the truncated space.

In summary, the quotient approach that we presented is beneficial for several reasons. First, it reduces the number of optimization variables and the size of the moment and localizing matrices, leading to smaller SDPs. Second, it makes the relevant compression spaces intrinsic and purely algebraic, since they are determined by the quotient structure rather than by an explicit analysis of linear dependencies among moment variables. Third, compared with the naive truncated free-algebra formulation, it can yield a tighter relaxation by incorporating algebraic consequences of the equality constraints that are not directly visible within the truncation. Fourth, it can simplify the construction of the SDP relaxations in practice, since monomials and polynomials in \(\ncalgebra/I_R\) may admit significantly more compact representations than in the ambient free algebra \(\ncalgebra\).

\subsection{Representations of the quotient algebra}
Having established formally the quotient formulation, we can now go back to the folklore approach of turning each relation in \(R\) into a substitution rule, and understand it as an attempt to represent the quotient algebra \(\ncalgebra/I_R\) by choosing a representative for each class \([p]\in \ncalgebra/I_R\).
However, as the examples in subsection~\ref{subsec:pitfalls} show, it is not sufficient to treat each relation in \(R\) as an independent substitution rule. In order to rewrite correctly and consistently modulo the ideal, the rewriting system must capture not only the original relations in \(R\), but also all of their algebraic consequences inside \(\idealgen{R}\).

This phenomenon already appears in Example~\ref{example:abc}. There, the naive substitution rules $ba\mapsto ab$, $cb\mapsto bc$ derived from the two original relations $ba-ab$ and $cb-bc$ fail to identify the equivalent monomials \(cab\) and \(bca\) and are thus not sufficient to obtain normal forms. The missing relation arises from a combination of the original constraints:
\begin{equation}
\label{eq:cab_bca}
-c(ba-ab)+(cb-bc)a
= (-\cancel{cba}+cab)+(\cancel{cba}-bca)
= cab-bca.
\end{equation}
Thus the ideal $\idealgen{R}$ contains the additional relation \(cab-bca\), corresponding to the new substitution rule $cab \mapsto bca$. 
This is precisely the transformation needed to resolve the issue observed in~\eqref{example:abc}.

Similarly, in Example~\ref{example:ab}, the two relations $ab-a$ and $ba-b$
imply that the following combination also lies in the ideal:
\begin{equation}
-(ab-a)a+a(ba-b)+(ab-a)
=(-\cancel{aba}+a^2)+(\cancel{aba}-\cancel{ab})+(\cancel{ab}-a)
=a^2-a,
\end{equation}
which yields the additional substitution rule $a^2\mapsto a$, which resolves the ambiguity in~\eqref{eq:ex3}. 

Likewise, in Example~\ref{example:ababa}, the single relation \(aba-c\) implies
\begin{equation}
-(aba-c)ba+ab(aba-c)= (-\cancel{ababa}+cba)+(\cancel{ababa}-abc) = cba-abc,
\end{equation}
and produces the additional rule \(cba\mapsto abc\), thereby resolving the ambiguity in~\eqref{eq:ex4}.

\subsubsection{Rewriting rules and Gröbner bases}\label{sec:groebner}
A systematic way to generate a complete set of rewriting rules, yielding unique
normal forms for the equivalence classes in $\ncalgebra/I_R$, is through the
construction of a \emph{Gröbner basis}. 

We explain first how a set of polynomials $G \subset \ncalgebra$ formally induces a
rewriting system. We recall the notation $\lt{p}$, $\lm{p}$, and $\lc{p}$ for, respectively, the leading term, leading monomial, and leading coefficient of $p$ (according to some chosen monomial well-order). After rescaling the elements of $G$, we may assume that they are monic, i.e., that $\lc{g}=1$ for all $g\in G$.

Each polynomial $g \in G$ can be interpreted as defining a rewriting rule $\lm{g} \mapsto \lm{g} - g$. Thus, whenever a monomial $w \in \ncmonoid$ contains $\lm{g}$ as a subword,
say $w = u\,\lm{g}\,v$, we may rewrite $w \mapsto u\,(\lm{g}-g)\,v$.
Equivalently, for a polynomial $p$, a reduction step replaces an occurrence of
$u\,\lm{g}\,v$ in $p$ by $u(\lm{g}-g)v$. 
Formally, the rewriting system associated to $G$ is the following.
% \algorithmbox{df}{
\begin{algorithm}[H]
\caption{Rewriting or division algorithm relative to $G$}
\label{alg:nc-division}
\begin{algorithmic}[1]
\REQUIRE A polynomial $p \in \ncalgebra$ and a finite set
$G=\{g_1,\ldots,g_s\}\subset \ncalgebra$ of monic polynomials.
\STATE $q \gets p$ and $r \gets 0$
\WHILE{$q \neq 0$}
    \IF{there exist $g_i\in G$ and $u,v\in\ncmonoid$ such that
    $\lm{q}=u\,\lm{g_i}\,v$}
        \STATE $q \gets q - \lc{q}\,u\,g_i\,v$
    \ELSE
        \STATE $r \gets r + \lt{q}$
        \STATE $q \gets q - \lt{q}$
    \ENDIF
\ENDWHILE
\RETURN $r$, the remainder of $p$ reduced with respect to $G$
\end{algorithmic}
\end{algorithm}
The output is a polynomial $r$ that is reduced with respect to $G$: no monomial appearing in $r$
contains any of the leading monomials $\lm{g_i}$ as a subword. 

The above rewriting process has the following two properties. First, it is terminating. This follows from the fact that, since $\lm{g_i}-g_i$ only contains
monomials strictly smaller than $\lm{g_i}$, every reduction step strictly
decreases the leading monomial of the polynomial being reduced. Since the monomial order is a well-ordering (i.e., every strictly descending sequence of monomials is finite), the process must eventually terminate. Second, the output $r$ is equivalent to the input $p$ modulo the ideal generated by $G$: one can write
$
p = r + \sum_{i} u_i g_i v_i
$ for some $u_i,v_i\in\ncmonoid$. Thus, the rewriting process produces a representative of the class $[p]$ in the quotient $\ncalgebra/\idealgen{G}$. In general,
however, the result may depend on the choices made during reduction. If a monomial contains several reducible subwords, different choices may produce different remainders $r$.

A Gröbner basis $G$ for the ideal \(I_R=\idealgen{R}\) is a special set of polynomials with the following properties
\begin{enumerate}
\item \(G\) generates the ideal: \(\idealgen{G} = \idealgen{R}\).
\item The associated rewriting system is confluent, meaning that for every polynomial \(p\), the output \(r\) obtained by applying the rewriting algorithm relative to \(G\) is the same regardless of the choices made during reduction.
\end{enumerate}
The first condition ensures that the output $r$ is in the same class as $p$ modulo $I_R$. The second condition ensures that two polynomials $p$, $q$ equivalent modulo $I_R$ (i.e., such that $p-q\in I_R$) will be reduced to the same output $r$. Indeed, by confluence their reduction paths must eventually meet; since reduction terminates, the endpoint is unique.
It follows that the rewriting algorithm relative to a Gröbner basis \(G\) produces a unique normal form $\gnf(p)$ for each class $[p]$ in \(\ncalgebra/I_R\).

There exist general algorithms for constructing Gröbner bases from a given generating set \(R\). In the noncommutative setting, the standard tools are the noncommutative Buchberger algorithm~\cite{mora1994introduction} and the Knuth--Bendix completion procedure~\cite{knuth1970simple}. In principle, this provides a general method for quotient reduction in NCPO: once a Gröbner basis is known, one can reduce every polynomial to its normal form modulo \(I_R\), and use these normal forms to build the SDP relaxations in the quotient formulation.

Note that the rewriting algorithm \ref{alg:nc-division} depends on the chosen monomial order, since the leading monomial of each polynomial in \(G\) depends on this choice. Therefore, the Gröbner basis itself also depends on the monomial order.

In the Bell-CHSH example, the initial set of constraints \(R\) is already a Gröbner basis (for the deg-lex ordering). This explains why the substitution rules~\eqref{eq:chsh_rules} are sufficient to reduce every monomial to a unique canonical representative.

By contrast, in Example~\ref{example:abc}, the Gröbner basis of the ideal \(\idealgen{R}\) is 
\begin{equation}
G
=
R\cup \{ca^k b-bca^k : k=1,2,\dots\}.
\end{equation}
Interpreting each element of \(G\) as a rewriting rule,
\begin{equation}
ba\mapsto ab,
\qquad
cb\mapsto bc,
\qquad
ca^k b \mapsto bca^k \quad (k\ge 1),
\end{equation}
yields a complete rewriting system that correctly identifies equivalent monomials and produces unique normal forms. Likewise, in Example~\ref{example:ab}, the Gröbner basis is $G=\{ab-a,\; ba-b,\; a^2-a,\; b^2-b\}$, with the corresponding rewriting rules $ab\mapsto a$, $ba\mapsto b$, $a^2\mapsto a$, $b^2\mapsto b$. In Example~\ref{example:ababa}, the Gröbner basis is \(G=\{aba-c,\; cba-abc\}\), with the corresponding rewriting rules \(aba\mapsto c\) and \(cba\mapsto abc\).

Although Gröbner bases provide a conceptually general solution, they also have significant practical limitations. Their computation may be prohibitively expensive, or even fail to terminate, since computing a complete Gröbner basis is in general equivalent to solving the word problem for the ideal \(I_R\), which is undecidable. Even if a Gröbner basis can be determined, it may contain a large number of elements or even infinitely many, as in Example~\ref{example:abc}. In practice, one often truncates the computation to bounded degree, but this can still be costly.

Even when a Gröbner basis is available, it need not provide the most efficient concrete representation of the quotient algebra. In Example~\ref{example:abc}, the relations simply express that \(b\) commutes with both \(a\) and \(c\). Rather than applying recursively the infinitely many Gröbner rewriting rules, one may exploit this structure directly: the position of \(b\) in a word is irrelevant, so a canonical representative can be obtained by moving all occurrences of \(b\) to a fixed position and recording their multiplicity. For instance, the monomial
\begin{equation}
w=acbcabcabcbaba
\end{equation}
contains five occurrences of \(b\), and is equivalent modulo the relations to
\begin{equation}
ac^2acaca^2\, b^5.
\end{equation}
Thus the pair
\begin{equation}
(ac^2acaca^2,\; b^5)
\end{equation}
provides a more compact and efficient description of the corresponding quotient class than repeated application of the full Gröbner rewriting system.

In summary, Gröbner bases provide, when they can be computed, a general mechanism for constructing canonical representatives in quotient algebras, and hence for implementing quotient-based SDP relaxations. However, they are often too expensive to use directly, and in many structured situations there exist alternative quotient representations that are far more efficient. In Section~\ref{sec:pcpo}, we focus precisely on such a structured case, corresponding to partial commutation relations among the variables, where the quotient algebra admits a simple description.

\subsubsection{Binomial relations}\label{sec:binomial_relations}
In the Examples~\ref{example:chsh} to \ref{example:ababa}, the set $R$ consists of pure binomial relations of the form 
\begin{equation}
  R = \{u_1 - v_1, \ldots, u_m - v_m\},
\end{equation}
where each $u_i-v_i$ is a difference of two monomials $u_i, v_i \in \ncmonoid$. 
In this case, the algebraic analysis of the quotient $\ncalgebra/\idealgen{R}$ simplifies significantly by descending to the monoidal level, since the relations in $R$ act only on the basis monomials rather than arbitrary linear combinations of them.

More precisely, let 
\begin{equation}
  \ncmonoid_R = \langle X \mid R \rangle = \langle X \mid u_1=v_1,\ldots,u_m=v_m\rangle
\end{equation}
be the finitely presented monoid defined by the generators $X$ and the relations $R$. 
That is, $\ncmonoid_R$ is the set of equivalence classes of words in the alphabet $X$, where two words are identified if one can be transformed into the other by a finite sequence of replacements of some occurrence of $u_i$ by $v_i$, or of $v_i$ by $u_i$, inside a larger word.
Formally, this identification defines a ``congruence", i.e., an equivalence relation that respects the multiplication operation: if two words are $u$ and $u'$ are equivalent and two words $v$ and $v'$ are equivalent, then the concatenated words $uv$ and $u'v'$ are also equivalent. This ensures that $\ncmonoid_R$ is a monoid with well-defined product
\begin{equation}
  [u]\cdot [v] \coloneqq [uv],
\end{equation} 
where $[w]$ denotes the equivalence class of a word $w\in \ncmonoid$.

Since $R$ is closed under involution, the involution on $\ncmonoid$ also descends to $\ncmonoid_R$, with $[w]^\inv \coloneqq [w^\inv]$. This is well-defined: if two words represent the same element of $\ncmonoid_R$, then their involutions also represent the same element.

The polynomial quotient algebra $\ncalgebra/I_R$ can then be identified with the
\emph{monoid algebra} $\kk[\ncmonoid_R]$, which consists of all finite
$\kk$-linear combinations of equivalence classes of words:
\begin{equation}\label{eq:quotient_monoid_algebra}
  \ncalgebra/I_R \cong \kk[\ncmonoid_R]\,.
\end{equation}
In other words, first taking linear combinations of words and then quotienting
by the ideal $I_R$ gives the same result as first quotienting the free monoid
$\ncmonoid$ by the word relations $u_i=v_i$, and then taking linear combinations
of the resulting equivalence classes.

To see this, let $\pi:\ncmonoid\to\ncmonoid_R$ be the map sending a word $w$ to its equivalence class $[w]$. We extend this
linearly to a map
\begin{equation}
\Psi:\ncalgebra\to\kk[\ncmonoid_R],
\qquad
\Psi\left(\sum_w c_w w\right)=\sum_w c_w [w].
\end{equation}Since the product of words is sent to the product of their equivalence classes,
and since $\Psi$ is compatible with the involution, $\Psi$ is a surjective
$*$-homomorphism.

The first isomorphism theorem for algebras states that $\ncalgebra/\ker(\Psi)\cong \text{im}(\Psi)=\kk[\ncmonoid_R]$, which simply says that quotienting by the relation ``having the same image under $\Psi$'' gives precisely the image of $\Psi$.
Thus, establishing the desired isomorphism $\ncalgebra/I_R\cong \kk[\ncmonoid_R]$ reduces to showing that $\ker(\Psi)=I_R$.

For this, first notice that every generator
of $I_R$ is of the form $a(u_i-v_i)b$ for some words $a,b\in\ncmonoid$. Since the words $a u_i b$ and $a v_i b$ represent the same element in $\ncmonoid_R$, we have $\Psi(a u_i b-a v_i b) = [a u_i b]-[a v_i b] =0 $. Thus, $I_R\subseteq \ker(\Psi)$.

Conversely, let $p=\sum_w p_w w\in \ker(\Psi)$. The condition $\Psi(p)=0$ means that, after replacing each word by its equivalence class, the coefficients must cancel class by class. Equivalently, for
every equivalence class $C\subseteq\ncmonoid$, we must have $\sum_{w\in C} p_w=0$.
Fix one such class $C$ and choose a representative word $v_C\in C$. The part of
$p$ supported on $C$ can be written as
\begin{equation}
\sum_{w\in C} p_w w
=
\sum_{w\in C} p_w(w-v_C)
+
\left(\sum_{w\in C}p_w\right)v_C.
\end{equation}
The second term vanishes by our previous observation. 
Hence
\begin{equation}
\sum_{w\in C} p_w w
=
\sum_{w\in C} p_w(w-v_C).
\end{equation}
Since every word $w \in C$ is equivalent to $v_C$ in $\ncmonoid_R$, $w$ can be transformed into $v_C$ by a finite sequence of substitutions $u_i \leftrightarrow v_i$ inside larger words. Each such substitution step corresponds to adding an element of the form $a(u_i - v_i)b \in I_R$. It follows that each difference $w-v_C$ belongs to $I_R$. Summing over all classes $C$, we find $p \in I_R$, so $\ker(\Psi) \subseteq I_R$.

The isomorphism~\eqref{eq:quotient_monoid_algebra} identifies $\ncalgebra/I_R$ as the vector space spanned by the monoid elements of $\ncmonoid_R$. Crucially, quotienting by $I_R$ in $\ncalgebra$ does not create additional linear relations between distinct word-classes; it only merges words that represent the same element of the monoid $\ncmonoid_R$. Consequently, the classes $[w] \in \ncmonoid_R$ form a natural $\kk$-basis for the quotient algebra.

For the implementation of the quotient-based SDP relaxations, a polynomial in the quotient algebra can thus be represented as a finite coefficient vector indexed by elements of $\ncmonoid_R$. Consequently, the
quotient moment sequences and matrices are indexed by elements of
$\ncmonoid_R$ rather than by free words in $\ncmonoid$. Any concrete
representation of the elements of $\ncmonoid_R$ may be used, provided that it
supports equality testing, multiplication, and the involution.

The multiplication and involution of quotient polynomials are inherited directly
from the quotient monoid structure: if
\begin{equation}
  [p]=\sum_{[v]} p_{[v]} [v],
  \qquad
  [q]=\sum_{[w]} q_{[w]} [w],
  \qquad [v],[w]\in\ncmonoid_R,
\end{equation}
then
\begin{equation}
  [pq]=\sum_{[v],[w]}p_{[v]}q_{[w]}\,[v][w],
\end{equation}
and
\begin{equation}
  [p]^\inv=\sum_{[v]} \overline{p_{[v]}}\,[v]^\inv.
\end{equation}
Thus implementing the quotient algebra reduces to three operations on the
quotient monoid: representing its elements, multiplying them, and computing the
involution. The surrounding linear-algebraic layer is unchanged, except that its
basis is indexed by $\ncmonoid_R$ instead of $\ncmonoid$.

\subsection{Hybrid approaches}\label{sec:hybrid_approaches}
As we already mentioned, obtaining a concrete representation of the quotient algebra $\ncalgebra/I_R$ for general equality constraints $R$ can be difficult in general. However, in many cases of interest, the set of relations $R$ may admit a natural decomposition $R= R_0 \cup R_1$ into subsets that can be treated differently. For instance, $R_0$ may admit a simple concrete representation of the quotient algebra $\ncalgebra/I_{R_0}$, while $R_1$ may not. In such cases, it can be advantageous to use a hybrid approach where we treat the former by quotient reduction and the latter by explicit linear constraints in the SDP relaxation. The corresponding quotient-based SDP relaxations are then of the form 
\begin{tcolorbox}[title={SDP relaxation of NCPO with truncated basis \(B\subseteq \ncmonoid\), modulo $I_{R_0}$}, colframe=gray!60, colback=white, coltitle=black, colbacktitle=gray!60]
\begin{align}
\label{eq:quotient_truncated_relax_hybrid}
\hat p_{B/I_{R_0}}^{\mathrm{opt}}=\inf_{L}\quad & L([p])\\
\text{s.t.}\quad
& L([\id])=1, \nonumber\\& M_{B/I_{R_0}}(q,L)\succeq 0
\qquad \forall q\in \{\id\}\cup Q, \nonumber\\
& L([s] [r] [t])=0
\qquad \forall [r]\in [R_1],\ \forall [s],[t]\in W_{B^\inv B}
\text{ s.t.} \supp{[s][r][t]}\subseteq W_{B^\inv B}, \nonumber
\end{align}
where the optimization is over all linear functionals
$L:[\ncalgebra_{B^\inv B}]\to \kk$.
\end{tcolorbox} 
Here, $[R_1]=\{[r]:\, r\in R_1\}$ is the image of $R_1$ in the quotient algebra $\ncalgebra/I_{R_0}$. 

\subsection{CPO as a special case of NCPO with quotient reduction}
\label{subsec:cpo_special_case}

A basic example that brings together the ideas developed above---quotient reduction, the monoid-level treatment of binomial relations, and the possibility of combining quotient reduction with additional polynomial equalities---is commutative polynomial optimization (CPO). Indeed, CPO may be viewed as a special case of NCPO in which the equality constraints include all pairwise commutation relations among the variables. In this way, the standard Lasserre hierarchy for CPO appears as the fully commutative quotient of the quotient-based NPA hierarchy for NCPO.

More precisely, a CPO problem is specified by $n$ real commuting variables $x=(x_1,\ldots,x_n)$ defining the commutative polynomial algebra $\mathbb{R}[x_1,\ldots,x_n]$, together with a polynomial objective $p(x)$, a finite set $R$ of polynomial equality constraints, and a finite set $Q$ of polynomial inequality constraints. The optimization problem is
\begin{equation}
p^{\mathrm{opt}}
=
\inf_{x\in \mathbb{R}^n}
\{\, p(x):\ r(x)=0\ \forall r\in R,\ q(x)\ge 0\ \forall q\in Q \,\}.
\end{equation}

We may view this formally as an instance of the NCPO problem~\eqref{eq:ncpo} by taking
\begin{equation}
X=\{x_1,\ldots,x_n\},
\qquad
x_i^\inv=x_i \ \ \forall i,
\qquad
\kk=\mathbb{R},
\end{equation}
and regarding the same polynomials\footnote{Strictly speaking, one first replaces all polynomials $p$, $r$, $q$ by their Hermitian symmetrization if necessary, where the symmetrization of a polynomial $s$ is $(s+s^\inv)/2$. In the commutative quotient $\mathbb R\langle X\rangle/\idealgen{R_{\mathrm{com}}}$, the class of $s$ is unchanged by this symmetrization, but it makes the NCPO formulation well defined in the free algebra $\mathbb R\langle X\rangle$.}
$p$, $R$, and $Q$ as elements of the free $*$-algebra $\ncalgebra=\mathbb{R}\langle X\rangle$. To recover the commutative setting, one adds the full set of pairwise commutation relations
\begin{equation}
R_{\mathrm{com}}
=
\{x_i x_j-x_j x_i:\ 1\le i<j\le n\}
\end{equation}
to the equality constraints. Thus, the corresponding NCPO formulation has equality set $R_{\mathrm{com}}\cup R$.

We can quotient only by the commutation relations $R_{\mathrm{com}}$ and keep the remaining polynomial equalities $R$ explicitly in the SDP relaxation. The relations in $R_{\mathrm{com}}$ are binomial: each of them identifies two monomials $x_i x_j = x_j x_i$.
As explained earlier, in this situation the quotient can be described directly at the level of the underlying monoid. In the present case, the quotient monoid
\begin{equation}
\ncmonoid_{R_{\mathrm{com}}}
=
\langle x_1,\ldots,x_n \mid x_i x_j = x_j x_i,\ 1\le i < j \le n \rangle
\end{equation}
is the free commutative monoid on $n$ generators. Its elements are in bijection with multi-indices $\alpha=(\alpha_1,\ldots,\alpha_n)\in \mathbb{N}^n$ each corresponding to the commutative monomial
\begin{equation}
x^\alpha = x_1^{\alpha_1}\cdots x_n^{\alpha_n}.
\end{equation}
The product in $\ncmonoid_{R_{\mathrm{com}}}$ is simply $x^\alpha \cdot x^\beta = x^{\alpha+\beta}$ and the involution acts trivially since all variables are Hermitian: $(x^\alpha)^\inv = x^\alpha$.

Accordingly, the quotient algebra
\begin{equation}
\ncalgebra/\idealgen{R_{\mathrm{com}}} \cong \kk[\ncmonoid_{R_{\mathrm{com}}}] = \mathbb{R}[x_1,\ldots,x_n]
\end{equation}
is naturally identified with the ordinary commutative polynomial algebra $\mathbb{R}[x_1,\ldots,x_n]$. 

Let us now consider the truncated quotient relaxation at order $d$. In the free algebra, one would start from the set $B_d=\ncmonoid_d$ of words of degree at most $d$. After quotienting by $R_{\mathrm{com}}$, the corresponding quotient basis is indexed by the commutative monomials
\begin{equation}
\mathbb{N}^n_d:=\{\alpha\in \mathbb{N}^n:\ |\alpha|\le d\},
\qquad
|\alpha|:=\alpha_1+\cdots+\alpha_n.
\end{equation}
Similarly, the retained moments are indexed by $\mathbb{N}^n_{2d}:=\{\alpha\in \mathbb{N}^n:\ |\alpha|\le 2d\}$.
If we denote by
\begin{equation}
y_\alpha := L(x^\alpha),
\qquad
\alpha\in \mathbb{N}^n_{2d},
\end{equation}
the quotient moments, then the moment matrix is indexed by $\mathbb{N}^n_d$ and has entries
\begin{equation}
M_d(y)(\alpha,\beta)=y_{\alpha+\beta},
\qquad
\alpha,\beta\in \mathbb{N}^n_d.
\end{equation}
Likewise, if $q=\sum_\gamma q_\gamma x^\gamma$, the corresponding localizing matrix is indexed by $\mathbb{N}^n_{d-\lceil \deg(q)/2\rceil}$ and has entries
\begin{equation}
M_d(q,y)(\alpha,\beta)
=
\sum_\gamma q_\gamma\, y_{\alpha+\beta+\gamma}.
\end{equation}

Thus the order-$d$ quotient-based relaxation reads
\begin{align}
\label{eq:cpo_lasserre_relax}
p_d^{\mathrm{opt}}=\inf_y\quad & \sum_\alpha p_\alpha\, y_\alpha \\
\text{s.t.}\quad
& y_0=1, \nonumber\\
& M_d(y)\succeq 0, \nonumber\\
& M_d(q,y)\succeq 0 \qquad \forall q\in Q, \nonumber\\
& \sum_\gamma r_\gamma\, y_{\alpha+\beta+\gamma}=0
\qquad
\forall r=\sum_\gamma r_\gamma x^\gamma\in R,\ 
\forall \alpha,\beta\in \mathbb{N}^n_{d-\lceil \deg(r)/2\rceil}. \nonumber
\end{align}
This is exactly the standard SDP relaxation of the Lasserre hierarchy for commutative polynomial optimization.
In the next subsection, we will see that the same philosophy extends to the partially commutative setting.

\section{Partially commutative polynomial optimization}
\label{sec:pcpo}
A natural and important special case of the NCPO framework arises when some of the equality constraints encode partial commutation relations among the variables. We now specialize the general quotient approach to this setting.

\paragraph{Partial commutations and commutation graph.}
A partial commutation structure on the alphabet $\xx$ is encoded by a simple undirected graph $\graph = (\xx, E)$, called the \emph{commutation graph}, where an edge $(x_i, x_j) \in E$ indicates that $x_i$ and $x_j$ commute. The corresponding set of partial commutation relations is
\begin{equation}\label{eq:rg}
R_\graph = \{x_i x_j - x_j x_i : (x_i, x_j) \in E\}.
\end{equation}
We assume that the graph is compatible with the involution, in the sense that if $(x_i, x_j) \in E$, then $(x_i^\inv, x_j^\inv) \in E$. This ensures $R_\graph$ is closed under involution. 
In the literature on partially commutative algebras, the commutation graph $\graph$ is often called the \emph{independence} graph. By convention, one assumes that $\graph$ has no self-loops, so that no variable is specifically declared to commute with itself. This is a technical convention that does not affect the algebraic structure, since the relation $x_i x_i - x_i x_i=0$ is trivial and can be safely omitted from $R_\graph$. 

\paragraph{PCPO problem.}
A partially commutative polynomial optimization (PCPO) problem is an NCPO problem in which the equality constraints decompose as $R_\graph \cup R$, where $R_\graph$ are the commutation relations of a commutation graph $\graph$ and $R$ is an additional set of equality constraints. More explicitly, a PCPO problem is specified by \emph{(i)} the alphabet $\xx$, the involution on $\xx$, and the field $\kk$, defining the free $*$-algebra $\ncalgebra$ of polynomials; \emph{(ii)} a Hermitian polynomial $p \in \ncalgebra$ corresponding to the objective function; \emph{(iii)} a commutation graph $\graph = (\xx, E)$ encoding the partial commutation relations; \emph{(iv)} a finite set of polynomials $R \subset \ncalgebra$ representing additional equality constraints; and \emph{(v)} a finite set of Hermitian polynomials $Q \subset \ncalgebra$ representing inequality constraints.  The corresponding PCPO problem is then
\begin{tcolorbox}[title=Partially commutative polynomial optimization (PCPO), colframe=gray!60, colback=white, coltitle=black, colbacktitle=gray!60]
\begin{align}
p^{\mathrm{opt}}
=
\inf_A \inf_{\|\psi\|=1}\ 
&\langle \psi,p(A)\psi\rangle \\
\text{s.t.}\quad
&A_xA_y= A_yA_x \qquad \forall \{x,y\}\in E, \nonumber\\
&r(A)=0 \qquad \forall r\in R, \nonumber\\
&q(A)\ge 0 \qquad \forall q\in Q. \nonumber
\end{align}
\end{tcolorbox}
This represents a class of problems that interpolates naturally between fully noncommutative and fully commutative polynomial optimization. If $E=\emptyset$, one recovers a standard NCPO problem. If \(\graph\) is complete, all pairs of variables commute, and one recovers an ordinary CPO problem. The intermediate PCPO setting appears naturally in quantum information as it captures the commutation structure of operators acting on different local quantum systems (see Section~\ref{subsec:local_operators} below).

\subsection{The partially commutative monoid}
\label{subsec:pc_monoid}
To implement PCPO through the hybrid quotient approach, where the commutation relations $R_\graph$ are quotiented out and the remaining polynomial equalities $R$ are retained explicitly, we need a concrete way of working with the quotient algebra $\ncalgebra/\idealgen{R_\graph}$. Since the relations in \(R_\graph\) are binomial, the quotient can be described directly at the monoid level, as explained in subsection~\ref{sec:binomial_relations}

The monoid
\begin{equation}
\ncmonoid_\graph = \langle \xx \mid R_\graph \rangle = \langle \xx \mid x_i x_j = x_j x_i \text{ for all } (x_i, x_j) \in E \rangle
\end{equation}
is known in the literature as the \emph{partially commutative (PC) monoid} or the \emph{trace monoid}. Its mathematical foundations were established by Cartier and Foata in combinatorics \cite{cartier1969applications} and later by Mazurkiewicz for modelling concurrent computation systems \cite{mazurkiewicz1977concurrent}.  A thorough account can be found in~\cite{diekert1995book, rozenberg1997handbook, lothaire2002}.

The elements of $\ncmonoid_\graph$ are called \emph{partially commutative monomials (pc-monomials), partially commutative words (pc-words)}, or \emph{traces} in the litterature.
They are the equivalence classes of words in $\ncmonoid$ under the congruence generated by the commutation rules encoded by $\graph$. Two free words $u,v\in \ncmonoid$ belong to the same equivalence class $[u]=[v]$ if and only if there exists a finite sequence of swaps of adjacent commuting variables that transforms $u$ into $v$.

The product in $\ncmonoid_\graph$ is induced by concatenation of representatives, $[u] \cdot [v] = [uv]$. Since $R_\graph$ is invariant under the involution (the commutation graph satisfies $(x,y) \in E \Rightarrow (x^\inv, y^\inv) \in E$), the involution descends to the quotient: $[w]^\inv=[w^\inv]$.

\begin{example}\label{example:abcd} Let $X = \{a,b,c,d\}$, and
  \begin{equation}\label{ex:graph}\graph =
  \tikz[baseline=-0.5ex]{
    \node (a) at (0,0.5) {$a$};
    \node (b) at (1,0.5) {$b$};
    \node (c) at (1,-0.5) {$c$};
    \node (d) at (0,-0.5) {$d$};
    \draw (a) -- (b);
    \draw (b) -- (c);
  }\end{equation}
  so $b$ commutes with $a$ and $c$. The equivalence class of the word $w=cbabdb$ is
\begin{align}
    [cbabdb] = & \{cabbdb, cbabdb, bcabdb, cbbadb, bcbadb, bbcadb \}\,.
\end{align}
As further illustration, $[cbabdb]\cdot [cba] = [cbabdbcba] = [bcbadbbca] = [bcbadb]\cdot[bca]$ and assuming hermitian variables, $[cbabdb]^\inv = [bdbabc]$.
\end{example}

Since two words $u,v\in \ncmonoid$ belong to the same equivalence class only if their letters can be rearranged into one another by swapping adjacent commuting variables, they must have the same multiset of letters. In particular, all representatives of a pc-monomial $[w]$ have the same length or degree. This allows us to define the degree of a pc-monomial unambiguously as $|[w]| := |w|$, and we denote the set of pc-monomials of degree at most $d$ as
\begin{equation}
  \ncmonoid_{\graph,d} := \bigl\{ [w] : w \in \ncmonoid, |w| \le d \bigr\}.
\end{equation}

As explained in Subsection~\ref{sec:binomial_relations} for general binomial relations, the monoid algebra $\kk[\ncmonoid_\graph]$, whose elements are formal linear combinations of pc-monomials $\sum_{[w]\in \ncmonoid_\graph} p_{[w]} [w]$, is naturally isomorphic to the quotient algebra $\ncalgebra/\idealgen{R_\graph}$. The canonical projection $\pi: \ncalgebra \to \kk[\ncmonoid_\graph]$ acts by regrouping the coefficients of all free monomials that fall into the same equivalence class. For example, in the setting of Example~\ref{example:abcd}, a polynomial $p = 2abc +bac + 3bcbadb - 5bcabdb$ would be mapped to $[p] = (2+1)[abc]  + (3 - 5)[cabbdb] = 3[abc] - 2[cabbdb]$ in the quotient.

Thus, SDP implementations of PCPO problems simply rely on indexing moment variables and matrices with pc-monomials $[w]$ rather than free monomials $w$. All words that differ only by the commutation of variables merge into a single index, eliminating redundant constraints and variables from the relaxation, which significantly reduces the size of the SDP. 
To implement this quotient-based relaxation requires a concrete way of representing and manipulating pc-monomials. We now describe several possible representations.

\subsection{Word representatives of pc-monomials}
\label{sec:word_representatives}

A first way to work concretely with $\ncmonoid_\graph$ is to choose, for each
class $[w]\in \ncmonoid_\graph$, a distinguished word representative $v\in\ncmonoid$ with $[v]=[w]$.  Such a representative gives a normal form for pc-monomials and therefore a concrete indexing set for quotient-based SDP
relaxations.

As recalled in Subsection~\ref{sec:groebner}, one systematic way to obtain
normal forms is to reduce words with respect to a Gr\"obner basis of the
PC ideal $\idealgen{R_\graph}$. The Gröbner basis itself can be determined using the
noncommutative Buchberger algorithm, or one of its variants. This requires a monomial order
$<$ on $\ncmonoid$. Both the Gr\"obner basis and the resulting normal form
depend on this order; we denote the normal form of a word $w$ obtained in this
way by $\gnf_<(w)$.

There is also a more direct order-theoretic way to choose representatives.
Let $<$ be any partial order relation on $\ncmonoid$ such that, on every commutation class
$[w]\in\ncmonoid_\graph$, there is a unique minimal element.
We define the corresponding order-normal form by
\begin{equation}
  \nf_{<}(w)
  =
  \min_{<}\{\,v\in\ncmonoid : [v]=[w]\,\}.
\end{equation}
For this definition to make sense, the relation $<$ on $\ncmonoid$ needs not be a monomial order, as defined in Subsection~\ref{sec:notation}. When the chosen relation is a monomial order $<$, however, the order-normal form agrees with the Gr\"obner normal form.

\begin{lemma}\label{lemma:gnf_lex} Let $<$ be an admissible monomial order on $\ncmonoid$. Then
\begin{equation}
\gnf_<(w)=\nf_{<}(w)
\end{equation}
for every $w\in\ncmonoid$.
\end{lemma}

\begin{proof}
The Gröbner normal form
$\gnf_<(w)$ is the unique irreducible representative obtained by reducing $w$
with respect to $G$. Since every reduction rule associated with $G$ has the form
$\lm{g} \longmapsto \lm{g}-g$, it strictly decreases a word with respect to $<$. Assume that $v = \nf_{<}(w) \neq \gnf_<(w)$. Since $v$ is minimal in the class $[w]$, it cannot therefore be converted into $\gnf_<(w)$ by the reduction rules associated to $G$. But this contradicts the fact that $\gnf_<(w)$ is the unique irreducible representative of $[w]$.
\end{proof}

We now specialize this discussion to lexicographic representatives.
Given a
total order $<_{\xx}$ on the alphabet $\xx$, let $<_{\mathrm{lex}}$ be the
lexicographic order on $\ncmonoid$ induced by $<_{\xx}$. Thus
$v<_{\mathrm{lex}}w$ if, at the first position where $v$ and $w$ differ, the
letter of $v$ is smaller than the letter of $w$ according to $<_{\xx}$. Since
all representatives of a fixed pc-monomial have the same length, this
lexicographic comparison is unambiguous within each commutation class and yields a unique minimal element.
We can thus define the corresponding normal form $\nf_{<_\mathrm{lex}}(w)$.

Given a monomial order $<$ on $\ncmonoid$, let $<_{\xx}$ denote the order induced by $<$ on the indiviual letters and $<_{\mathrm{lex}}$ the corresponding induced lexicographic order. 
We say that
$<$ is \emph{classwise lexicographic} for $\ncmonoid_\graph$ if, on each
commutation class, it coincides with the lexicographic order: for all
$v,w\in\ncmonoid$ such that $[v]=[w]$, $v< w \Longleftrightarrow v<_{\mathrm{lex}} w$.
Many natural monomial orders have this property. For instance, the
degree-lexicographic order is classwise lexicographic, because all words in a
given commutation class have the same length, and hence the same total degree.
Thus, within a commutation class, degree-lexicographic comparison reduces to
lexicographic comparison. Lemma~\ref{lemma:gnf_lex}
therefore gives the following corollary.

\begin{corollary}
Let $<$ be a monomial order that is classwise lexicographic for
$\ncmonoid_\graph$. Then
\begin{equation}
  \gnf_<(w)
  =
  \nf_{<_\mathrm{lex}}(w)
\end{equation}
for every $w\in\ncmonoid$.
\end{corollary}
In the following, we assume that a classwise lexicographic monomial order $<$ is fixed and we simply write $\gnf(w) = \lex(w)$ for the common Gr\"obner and lexicographic normal form.

The Gr\"obner bases for monomial orders that are classwise lexicographic on $\ncmonoid_\graph$ have a particularly simple form. The following result already appears in  \cite[Theorem 7.1]{Bokut*31072001}. 

\begin{proposition}\label{prop:gb}
  Let $<$ be a monomial order on $\ncmonoid$ that is classwise lexicographic on $\ncmonoid_\graph$. Then $\idealgen{R_\graph}$ has a Gr\"obner basis (relative to $<$) consisting of all binomials 
  \begin{equation}\label{eq:gb-pc-general}
  yux-xyu,
  \end{equation}
  where $(x,y)\in E$ with $x< y$, and where $u\in\ncmonoid$ is a possibly empty
  word such that every letter $z$ appearing in $u$ satisfies
  \begin{equation}
  (z,x)\in E
  \qquad\text{and}\qquad
  z< x.
  \end{equation}
\end{proposition}
\begin{proof}
Let $G$ be the family of binomials in~\eqref{eq:gb-pc-general}. Since $x$
commutes with $y$ and with every letter of $u$, the words $yux$ and $xyu$
represent the same pc-monomial. Hence every element of $G$ belongs to
$\idealgen{R_\graph}$. Conversely, the oriented degree-two commutation
relations $yx-xy$, for $(x,y)\in E$ and $x<y$, are obtained by
taking $u=1$. These relations generate the same ideal as $R_\graph$, so
$\idealgen{G}=\idealgen{R_\graph}$.

To prove that $G$ forms a Gr\"obner basis, it remains to prove that the rewriting system associated to $G$ is confluent. 
We claim that a word is irreducible with respect to $G$ if and only if it is
the lexicographically minimal representative of its commutation class. One
direction is immediate: if a word contains a factor $yux$ of the form
\eqref{eq:gb-pc-general}, replacing it by $xyu$ gives an equivalent and
lexicographically smaller word.

Conversely, suppose that a word $w$ is not lexicographically minimal in its
class. Let $v$ be a lexicographically smaller representative. Passing from
$w$ to $v$ is achieved by swaps of adjacent commuting letters. At the first
position where $v$ differs from $w$, a smaller letter $x$ from a later
position of $w$ has moved left across a larger letter $y$. Thus $w$ contains
a factor $yux$ such that $x< y$ and $x$ commutes with $y$ and with
every letter of $u$. Choose such a move with $u$ as short as possible. If a
letter $z$ in $u$ satisfied $x< z$, then the same occurrence of
$x$ could be moved left across $z$ with a shorter intervening block, again
producing a lexicographically smaller word. This contradicts the choice of
$u$. Since $x$ cannot commute with another occurrence of itself, every letter
$z$ in $u$ satisfies $z< x$. Hence $w$ contains a leading monomial
of an element of $G$ and is reducible.

Therefore the irreducible words are exactly the lexicographically minimal
representatives. Each commutation class has a unique such representative, so
the terminating rewriting system has unique normal forms. Hence $G$ is a
Gr\"obner basis.
\end{proof}

\begin{example}[name=Running example,continues=example:abcd]
Consider a classwise lexicographic monomial order on $\ncmonoid$ which induces the alphabet order $a<b<c<d$. Then, according to Proposition~\ref{prop:gb}, a Gr\"obner basis of
$\idealgen{R_{\graph}}$ is
\begin{equation}\label{eq:ex:gb}
  G
  =
  \{ba-ab\}
  \cup
  \{cb-bc\}
  \cup
  \{ca^k b-bca^k:\ k\geq 1\}.
\end{equation}
For the word $cbabdb$, reduction gives
\begin{align} \label{eq:cbabd_lex}
cbabdb
=
(cb)abdb
&\mapsto
(bc)abdb
=
b(cab)db
\mapsto
b(bca)db
=
bbcadb.
\end{align}
Thus $\gnf(cbabdb)=bbcadb=\lex(cbabdb)$.

If instead we choose the alphabet order $a<c<b<d$, the Gr\"obner basis
reduces to the oriented degree-two relations
\begin{equation}\label{eq:ex:gb2}
  G'
  =
  \{ba-ab\}
  \cup
  \{bc-cb\}.
\end{equation}
The same word now reduces as
\begin{align}
cbabdb
=
c(ba)bdb
&\mapsto
c(ab)bdb
=
cabbdb.
\end{align}
Hence $\gnf(cbabdb)=cabbdb = \lex(cbabdb)$, which is the lexicographic normal form for this
alphabet order.
\end{example}

The Gr\"obner basis in Proposition~\ref{prop:gb} always contains the oriented degree-two relations $yx-xy$ for $(x,y) \in E$ with $x<y$, which correspond to the case $u=1$. As the above example illustrates, the higher-degree relations in~\eqref{eq:gb-pc-general} corresponding to $|u|>1$ may or may not be needed, depending on the chosen alphabet order.

There is a simple graph-theoretic condition that determines when the
higher-degree rules in~\eqref{eq:gb-pc-general} are necessary. An order
$<$ on the alphabet is \emph{transitive relative to} $\graph$ if
whenever $(x,y)\in E$ with $x<y$, and $(y,z)\in E$ with
$y<z$, then $(x,z)\in E$ and $x<z$
\cite[Definition~2.3]{Diekert1997}. Equivalently, orienting each edge from
the smaller endpoint to the larger endpoint gives a transitive orientation of
$\graph$. Conversely, any transitive orientation can be extended to a total
order with this property.

\begin{lemma}\label{lem:transitive_groebner}
If the alphabet order induced by $<$ is transitive relative to $\graph$,
then the original commutation relations $R_\graph$ constitute a Gröbner basis. If the
alphabet order is not transitive relative to $\graph$, then the Gr\"obner
basis contains infinitely many binomials of the form~\eqref{eq:gb-pc-general}.
\end{lemma}

\begin{proof}
Suppose first that the alphabet order is transitive relative to $\graph$. Let $G$ be the family of binomials in~\eqref{eq:gb-pc-general} and assume that $G$ contains a binomial of the form $yux-xyu$ with $|u|>1$. We will show that this binomial is redundant and can be removed from the Gr\"obner basis.

By assumption $(x,y)\in E$ with $x<y$, and every letter $z$ in $u$ satisfies
$(z,x)\in E$ and $z<x$. By transitivity, $(z,y)\in E$ and
$z<y$. Hence the oriented degree-two relations allow us to move $y$ to
the right across all letters of $u$, and then to swap $yx$ to $xy$:
\begin{equation}
  yux \longrightarrow uyx \longrightarrow uxy.
\end{equation}
They also allow us to move the initial $x$ in $xyu$ to the right across all
letters of $u$, after first moving $y$ to the right across $u$:
\begin{equation}
  xyu \longrightarrow xuy \longrightarrow uxy.
\end{equation}
Thus $yux$ and $xyu$ reduce to the same word using only degree-two
commutation relations, so the higher-degree rule is redundant in the rewriting system and can be removed from the Gr\"obner basis. Repeating this argument for all higher-degree rules shows that the original commutation relations $R_\graph$ suffice to form a Gröbner basis.

Conversely, suppose that the alphabet order is not transitive relative to
$\graph$. Then there exist letters $z<x<y$ such that
$(z,x)\in E$, $(x,y)\in E$, but $(z,y)\notin E$. For every
$k\geq 1$, the binomial
\begin{equation}
  yz^k x - xyz^k
\end{equation}
belongs to the family~\eqref{eq:gb-pc-general}. The leading monomial
$yz^k x$ has no proper reducible subword of the form required by
\eqref{eq:gb-pc-general}: $y$ cannot cross any occurrence of $z$, since
$(y,z)\notin E$, and the final $x$ has no smaller commuting letter to its
right. Hence these leading monomials are not consequences of shorter leading
monomials in the family. Since the words $yz^k x$ are distinct for different
values of $k$, infinitely many such binomials are needed.
\end{proof}

The fact that the Gröbner basis reduces to the original commutation relations
$R_\graph$ when the order is transitive relative to $\graph$ was already
established in~\cite[Theorem~2.2]{Diekert1997}.

As shown in the above Example, it is sometimes possible to simplify the Gröbner basis by choosing a suitable alphabet order on $\xx$ that is transitive relative to $\graph$. However,not every graph admits a transitive orientation, so this simplification is
not always available.

\begin{lemma} Let $X=\{a,b,c,d,e\}$ and let $\graph$ be the pentagon
  \begin{equation}
    \graph =
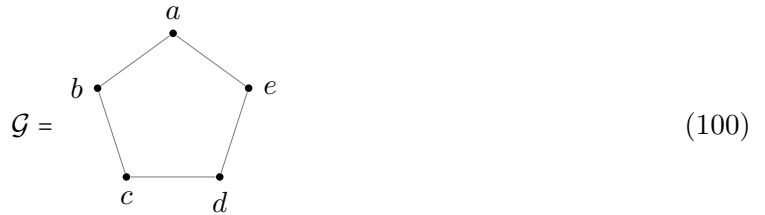
  \begin{tikzpicture}[
    baseline={([yshift=-2ex]current bounding box.center)},
    dot/.style={circle, fill, inner sep=1pt},
    lbl/.style={font=\small},
    scale=0.7
  ]
    \node[dot,label=above:{$a$}] (a) at (90:1.5)  {};
    \node[dot,label=left:{$b$}]  (b) at (162:1.5) {};
    \node[dot,label=below:{$c$}] (c) at (234:1.5) {};
    \node[dot,label=below:{$d$}] (d) at (306:1.5) {};
    \node[dot,label=right:{$e$}] (e) at (18:1.5)  {};
    \draw[gray] (a) -- (b) -- (c) -- (d) -- (e) -- (a);
  \end{tikzpicture}
\end{equation}
For every alphabet order, the Gr\"obner basis of $\idealgen{R_\graph}$
contains infinitely many binomials of the form~\eqref{eq:gb-pc-general}.
This is the smallest example, by number of letters, with this property.
\end{lemma}

\begin{proof}
By Lemma~\ref{lem:transitive_groebner}, the Gr\"obner basis is infinite for
every alphabet order if and only if $\graph$ admits no transitive
orientation. The pentagon is the smallest graph that does not admit a
transitive orientation.
\end{proof}

While Gröbner bases provide a systematic way to compute normal forms,
rewriting a word $w$ into its normal form $\gnf(w)$ through repeated
applications of the Gröbner basis substitution rules can be computationally
expensive. For
classwise lexicographic monomial orders, the normal form is just the
lexicographically minimal representative, and this representative can be
computed directly by a greedy procedure that does not require the Gröbner basis to be known. The procedure iteratively builds the normal form from left to right, at each step choosing the smallest letter that can be commuted to the current position.

Let $w=y_1\cdots y_r$. We say that position $j$ is \emph{reachable from
position $i$} if $y_j$ commutes with every intervening letter:
\begin{equation}
  (y_j,y_l)\in E
  \qquad
  \text{for all } i\leq l<j.
\end{equation}
Equivalently, the occurrence $y_j$ can be moved to position $i$ by swapping
it leftward across commuting letters. Note that reachability of position $j$ from $i$ depends
only on the letter $y_j$ and the letters $y_i, \ldots, y_{j-1}$; in
particular, a position $j' > j$ may be reachable from $i$ even if $j$ is not.

\begin{algorithm}[H]
\label{algo:lex_nf}
\caption{Lexicographic normal form}
\begin{algorithmic}[1]
  \REQUIRE A word $w = y_1 \cdots y_r$.
  \STATE $i \gets 1$.
  \WHILE{$i \le r$}
    \STATE Scan all positions $j = i, i+1, \ldots, r$.
    \STATE For each $j$, check reachability from $i$, i.e., $(y_j, y_l) \in E$ for all $i \le l < j$.
    \STATE Among all reachable positions, let $j^*$ be the leftmost one
    achieving the smallest letter $y_{j^*}$.
    \STATE Move in $w$ letter $y_{j^*}$ to position $i$:    remove it from position $j^*$ and reinsert it at position $i$, shifting $y_i, \ldots, y_{j^*-1}$ one step to the right.
    \STATE $i \gets i + 1$.
  \ENDWHILE
  \RETURN $w=\lex(w)$.
\end{algorithmic}
\end{algorithm}

The correctness of Algorithm~\ref{algo:lex_nf} follows from the same
prefix argument as the proof of Proposition~\ref{prop:gb}. Once the prefix
$y_1\cdots y_{i-1}$ has been fixed, the next letter of any equivalent word
must be an occurrence that is reachable from position $i$; otherwise it would
have to cross a non-commuting occurrence, which is impossible. The algorithm
therefore chooses the smallest feasible next letter at each step, producing
the lexicographically minimal representative.

A naive implementation of the reachability test would take $O(r^3)$ time,
because checking whether a position $j$ is reachable from $i$ requires testing
commutation with all letters $y_i,\ldots,y_{j-1}$. However, the reachability
tests for a fixed $i$ can be performed incrementally. While scanning
$j=i,i+1,\ldots,r$, maintain the set $S$ of letters that commute with all
previously scanned letters $y_i,\ldots,y_{j-1}$. Then $j$ is reachable from
$i$ precisely when $y_j\in S$, and after testing $j$ one updates
$S\leftarrow S\cap N(y_j)$, where $N(y_j)$ is the set of letters commuting
with $y_j$. With constant-time membership and update operations, for instance
for fixed alphabet size or using bitsets, this scan takes $O(r-i+1)$ time.
Summing over $i$ gives $O(r^2)$ time. The insertion step shifts at most $r$
letters at each iteration, hence also contributes $O(r^2)$ time.

Algorithm~\ref{algo:lex_nf} has several practical advantages over Gröbner
basis reduction. It does not require constructing or storing the Gröbner basis,
which may be infinite. It avoids matching higher-degree rules of the form $yux \to xyu$, placing each letter directly at its target position in a single move. And it is uniform: the same $O(r^2)$ procedure applies to every PC monoid, regardless of the structure of the commutation relations.

The above normal forms can be used to implement quotient-based SDP relaxations of PCPO problems. 
To build the moment variables and moment matrices one forms products such as $u^\inv v$ then replaces each product by its normal form before using it as a moment index. Thus the construction works in the free monoid and repeatedly 
projects back to the quotient monoid by normalization.

This is conceptually the same mechanism used in implementations such as
QuantumNPA.jl \cite{QuantumNPA.jl}: monomials are stored in a canonical word
form in which commuting blocks are moved into lexicographic order as far as
the commutation relations allow. In practice, normalization can often be
performed incrementally during multiplication, by inserting new letters or
blocks into an already normalized word and re-normalizing only the affected
part.

A different approach, developed in the next subsection, is to represent
pc-monomials by objects that encode the quotient element directly rather than
by a distinguished word representative. In such representations, operations
such as multiplication, involution, and equality testing can be performed on
the representing objects themselves, avoiding repeated word-normalization
steps when constructing the moment and localizing matrices.

\subsection{Circuit representation of pc-monomials}\label{sec:circuit_representation}

\subsubsection{Structural representations of pc-monomials}\label{sec:structural_representations}
The idea of representing elements of a partially commutative monoid by objects other than a chosen word representative has a long history in the theory of PC monoids. Already in the foundational works of Cartier and Foata \cite{cartier1969applications}, and later in the concurrency-oriented developments of Mazurkiewicz \cite{mazurkiewicz1977concurrent} pc-monomials were understood not merely as equivalence classes of words but as combinatorial objects encoding  events that may be permuted (when the letters commute) and relative orders that are fixed (when the letters do not commute). This includes Cartier--Foata normal forms~\cite{cartier1969applications}, dependency or occurrence graphs~\cite{duboc1986some}, and Viennot's heaps of pieces~\cite{viennot1986heaps}.

In this subsection, we reformulate these classical representations in a
circuit language closer to quantum information science. In that language,
commutation corresponds to independent local operations, while
non-commutation records pairs of operations whose order may matter because
their supports overlap.

We start by describing a general connection between the commutation graph
$\graph$ and the local operator structure of a family of operators
$\{A_x:\,x\in\xx\}$.

\subsubsection{Commutation graphs and local operator structures}\label{subsec:local_operators}

Consider a quantum system composed of $m$ subsystems $1,\ldots,m$, and let
$\{A_x:\,x\in\xx\}$ be a family of observables or operators acting on this
system. Each operator $A_x$ may act on a single subsystem, or more generally
on a subset $S_x\subseteq [m]$ of subsystems, called its support. If two
operators have disjoint supports, then locality forces them to commute. If
their supports overlap, locality no longer forces commutation; in an abstract
commutation model, such pairs are precisely the pairs that may be
non-commuting. Thus a support assignment induces a commutation graph
$\graph=(\xx,E)$ by
\begin{equation}\label{eq:commutation_graph_subsets}
  (x,y)\in E
  \qquad\Longleftrightarrow\qquad
  S_x\cap S_y=\emptyset .
\end{equation}

Conversely, every commutation graph can be realized in this way by a suitable
choice of supports. To see this, let $\complem\graph$ be the complement of the commutation graph: its vertices are the letters in $\xx$, and its edges
are the pairs $(x,y)$ with $(x,y)\notin E$. Thus, the edges of $\complem\graph$ connect the pairs of letters that may fail to commute. We refer to $\complem\graph$ as the \emph{non-commutation} graph.
In the PC monoid literature, it is often called the dependence graph. Remember that the commutation graph $\graph$ is assumed not to contain the diagonal pairs $(x,x)$, so $\complem\graph$ contains them. It will become clear later why considering these self-loops is useful. When drawing $\complem\graph$ we will omit these self-loops for visual clarity.

A clique of $\complem\graph$ is a subset of letters that are all connected by an edge, i.e., a subset of variables that are pairwise non-commuting. An edge clique cover of $\complem\graph$ is a family of cliques that covers every edge (including self-loops). 
Equivalently, it is a
family of cliques $\clique_1,\ldots,\clique_m$ such that every non-commuting pair of distinct letters lies together in at least one $\clique_k$, and every letter belongs to at least one $\clique_k$.
The minimum possible value of $m$ is the edge clique cover number of $\complem\graph$.  The set of all maximal cliques of $\complem\graph$ is one such cover, although not necessarily a minimum one.

Given such a cover $\clique_1,\ldots,\clique_m$, associate one subsystem with
each clique. Then assign to every letter $x\in\xx$ the support
\begin{equation}
  S_x := \{\,k\in[m] : x\in\clique_k\,\}.
\end{equation}
If $(x,y)\in E$, then $x$ and $y$ are not adjacent in
$\complem\graph$, so no clique of $\complem\graph$ contains them both. Hence
$S_x\cap S_y=\emptyset$. Conversely, if $x\neq y$ and $(x,y)\notin E$,
then $x$ and $y$ are adjacent in $\complem\graph$, and the cover places them
together in some clique $\clique_k$. Hence $S_x\cap S_y\neq\emptyset$. Thus
the support assignment reproduces exactly the commutation graph $\graph$.

Remember that we assumed that the commutation graph $\graph$ is compatible with the involution, namely $(x,y)\in E
  \Longleftrightarrow
(x^\inv,y^\inv)\in E$. Equivalently, the involution is an automorphism of the commutation graph, and hence also of the non-commutation graph $\complem\graph$. This compatibility should also be reflected by any edge clique cover that is interpreted as a decomposition into local systems.
% Indeed, if a clique $\clique_i$ represents one local system, then the adjoints of the operators associated with $\clique_i$ should again belong to one of the chosen local systems.
Given a clique $\clique_i$ of $\complem\graph$, define
\begin{equation}
  \clique_i^\inv
  :=
  \{x^\inv : x\in\clique_i\}
\end{equation}
as the set of involuted letters in $\clique_i$.
Since the involution preserves non-commutation, $\clique_i^\inv$ is again a clique of $\complem\graph$. We say that the cover $\clique_1,\ldots,\clique_m$ is \emph{closed under involution} if, for every $i$, there exists an index $i^\ast$ such that
\begin{equation}
\clique_i^\inv=\clique_{i^\ast}.
\end{equation}
The map $i\mapsto i^\ast$ is then an involution on the set of clique indices. Indeed, applying the letter involution twice gives $(\clique_i^\inv)^\inv=\clique_i$, and therefore $(i^\ast)^\ast=i$.

Choosing an edge clique cover that is closed under involution is not restrictive. Starting from any edge clique cover, one may always add the involuted cliques $\clique_i^\inv$ and then delete duplicates. Moreover, the cover consisting of all maximal cliques of $\complem\graph$ is automatically closed under involution, because graph automorphisms map maximal cliques to maximal cliques. Finally, if all variables are Hermitian, so that $x^\inv=x$ for every $x\in\xx$, then every edge clique cover is automatically closed under involution and one simply has $i^\ast=i$ for every clique index. In the sequel, whenever an edge clique cover is used to define local systems, we assume that it is closed under involution.

\begin{example}[Tripartite Bell scenario]
\label{ex:three_systems}
Consider six variables $\xx=\{a_0,a_1,b_0,b_1,c_0,c_1\}$ and the commutation graph $\graph=K_{2,2,2}$, the complete tripartite graph with vertex classes $\{a_0,a_1\}$, $\{b_0,b_1\}$, and $\{c_0,c_1\}$:
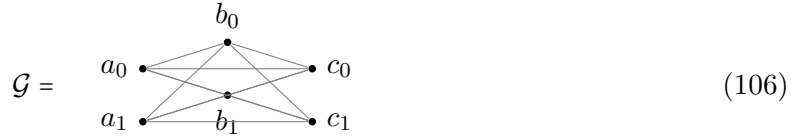
\begin{equation}
\mathcal{G} = \quad
\begin{tikzpicture}[
  baseline={([yshift=-2ex]current bounding box.center)}, % Centers the graph vertically relative to the = sign
  dot/.style={circle, fill, inner sep=1pt},
  lbl/.style={font=\small},
  scale=0.7 % Slightly reduced scale to fit better in an equation line
]
  \node[dot,label=left:{$a_0$}] (a0) at (0,1.5) {};
  \node[dot,label=left:{$a_1$}] (a1) at (0,0.5) {};
  \node[dot,label=above:{$b_0$}] (b0) at (1.6,2) {};
  \node[dot,label=below:{$b_1$}] (b1) at (1.6,1) {};
  \node[dot,label=right:{$c_0$}] (c0) at (3.2,1.5) {};
  \node[dot,label=right:{$c_1$}] (c1) at (3.2,0.5) {};
  \foreach \u in {a0,a1} \foreach \v in {b0,b1,c0,c1} \draw[gray] (\u) -- (\v);
  \foreach \u in {b0,b1} \foreach \v in {c0,c1} \draw[gray] (\u) -- (\v);
\end{tikzpicture}
\end{equation}
Apart from the diagonal dependence pairs, the non-commutation graph
$\complem\graph$ consists of three disjoint edges
$\{a_0,a_1\}$, $\{b_0,b_1\}$, and $\{c_0,c_1\}$:
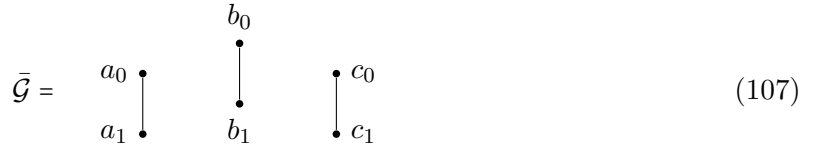
\begin{equation}
\bar{\graph} = \quad
\begin{tikzpicture}[  baseline={([yshift=-2ex]current bounding box.center)}, % Centers the graph vertically
  dot/.style={circle, fill, inner sep=1pt},
  lbl/.style={font=\small},
  scale=0.8
]
  \node[dot,label=left:{$a_0$}]  (a0c) at (0,1.5) {};
  \node[dot,label=left:{$a_1$}]  (a1c) at (0,0.5) {};
  \node[dot,label=above:{$b_0$}] (b0c) at (1.6,2) {};
  \node[dot,label=below:{$b_1$}] (b1c) at (1.6,1) {};
  \node[dot,label=right:{$c_0$}] (c0c) at (3.2,1.5) {};
  \node[dot,label=right:{$c_1$}] (c1c) at (3.2,0.5) {};
  \draw (a0c) -- (a1c);
  \draw (b0c) -- (b1c);
  \draw (c0c) -- (c1c);
\end{tikzpicture}
\end{equation}
An edge clique cover, consisting of the maximal cliques, is $\clique_A=\{a_0,a_1\}$, $\clique_B=\{b_0,b_1\}$, and $\clique_C=\{c_0,c_1\}$.
The associated supports are $S_{a_i}=\{A\}$, $S_{b_j}=\{B\}$, and $S_{c_k}=\{C\}$. Thus the variables
$a_i$ act on a subsystem $A$, the variables $b_j$ on a subsystem $B$, and the
variables $c_k$ on a subsystem $C$.

This is the standard tripartite Bell commutation structure. It generalizes to
arbitrary $m$ parties with $k_i$ local operators on party $i$, corresponding
to the complete $m$-partite graph $K_{k_1,\dots,k_m}$. The Bell--CHSH
commutation structure in Example~\ref{example:chsh} is the special case of
two parties with two local operators each.
\end{example}

\begin{example}[Routed Bell scenario]
\label{ex:routed_bell}
Consider six variables $\xx = \{a_0, a_1, t_0, t_1, b_0, b_1\}$ and the commutation graph $\graph$ where $a_i$ commutes with $t_j$ and $b_k$ for all $i,j,k$, and $b_0$ commutes with $b_1$:
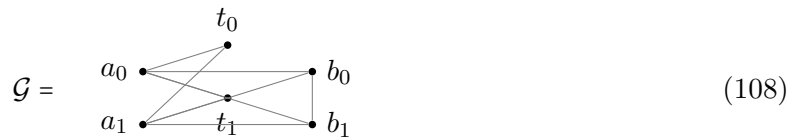
\begin{equation}
\mathcal{G} = \quad
\begin{tikzpicture}[
  baseline={([yshift=-2ex]current bounding box.center)},
  dot/.style={circle, fill, inner sep=1pt},
  lbl/.style={font=\small},
  scale=0.7
]
  \node[dot,label=left:{$a_0$}]  (a0) at (0,1.5) {};
  \node[dot,label=left:{$a_1$}]  (a1) at (0,0.5) {};
  \node[dot,label=above:{$t_0$}] (t0) at (1.6,2) {};
  \node[dot,label=below:{$t_1$}] (t1) at (1.6,1) {};
  \node[dot,label=right:{$b_0$}] (b0) at (3.2,1.5) {};
  \node[dot,label=right:{$b_1$}] (b1) at (3.2,0.5) {};
  \foreach \u in {a0,a1} \foreach \v in {t0,t1,b0,b1} \draw[gray] (\u) -- (\v);
  \draw[gray] (b0) -- (b1);
\end{tikzpicture}
\end{equation}
The complement graph $\bar{\graph}$ has edges $\{a_0,a_1\}$, $\{t_0,t_1\}$, and $\{t_j,b_k\}$ for all $j,k$:
\begin{equation}
\bar{\graph} = \quad
\begin{tikzpicture}[
  baseline={([yshift=-2ex]current bounding box.center)},
  dot/.style={circle, fill, inner sep=1pt},
  lbl/.style={font=\small},
  scale=0.7
]
  \node[dot,label=left:{$a_0$}]  (a0c) at (0,1.5) {};
  \node[dot,label=left:{$a_1$}]  (a1c) at (0,0.5) {};
  \node[dot,label=above:{$t_0$}] (t0c) at (1.6,2) {};
  \node[dot,label=below:{$t_1$}] (t1c) at (1.6,1) {};
  \node[dot,label=right:{$b_0$}] (b0c) at (3.2,1.5) {};
  \node[dot,label=right:{$b_1$}] (b1c) at (3.2,0.5) {};
  \draw (a0c) -- (a1c);
  \draw (t0c) -- (t1c);
  \foreach \u in {t0c,t1c} \foreach \v in {b0c,b1c} \draw (\u) -- (\v);
\end{tikzpicture}
\end{equation}
An edge clique cover is provided by the three maximal cliques $\clique_A=\{a_0,a_1\}$, $\clique_{B_0}=\{t_0,t_1,b_0\}$, and $\clique_{B_1}=\{t_0,t_1,b_1\}$.
The associated supports are $S_{a_i}=\{A\}$, $S_{t_j}=\{B_0,B_1\}$, and $S_{b_k}=\{B_k\}$. Thus the
variables $a_i$ act on subsystem $A$, the variables $t_j$ act jointly on
subsystems $B_0$ and $B_1$, and the variables $b_k$ act on subsystem $B_k$.

Note that \cite{Lobo2024certifyinglongrange} describes this scenario in coarser terms as a bipartite system, with Alice holding subsystem $A$ and Bob holding the composite subsystem $B=(B_0,B_1)$. In that coarser model,
$t_0,t_1,b_0,b_1$ all act on Bob's side, and the commutativity
$[b_0,b_1]=0$ appears as an additional algebraic constraint on Bob's
measurement operators. In the refined support model above, the same
commutativity is encoded directly as the edge $\{b_0,b_1\}\in E$, because
$b_0$ and $b_1$ act on disjoint subsystems $B_0$ and $B_1$.
\end{example}

\subsubsection{Circuits}
Fix an edge clique cover $\clique_1,\ldots,\clique_m$ of $\complem\graph$, as above. We interpret each letter $x\in\xx$ as an operator acting on its support $S_x$, i.e., on the subsystems corresponding to the cliques that contain $x$. We represent each of these cliques as a wire and each letter as a gate acting on the wires corresponding to its support.
A word $w=y_1y_2\cdots y_r$
can then be represented as a circuit $C(w)$ by placing the gates
$y_1,y_2,\ldots,y_r$ from left to right.

\begin{example}[name=Running example,continues=example:abcd]\label{example:abcd_circuit}
The non-commutation graph associated to the commutation graph \eqref{ex:graph} is
\begin{equation}
  \complem\graph = 
  \tikz[baseline=-0.5ex]{
    \node (a) at (0,0.5) {$a$};
    \node (b) at (1,0.5) {$b$};
    \node (c) at (1,-0.5) {$c$};
    \node (d) at (0,-0.5) {$d$};
    \draw (a) -- (c);
    \draw (a) -- (d);
    \draw (b) -- (d);
    \draw (c) -- (d);
  }
\end{equation}
An edge clique cover is provided by the following maximal cliques $\clique_1=\{a,c,d\}$, $\clique_{2}=\{b,d\}$.
We can thus represent each word as a circuit consisting of two wires, one for $\clique_1$ and one for $\clique_2$, with $a$ and $c$ acting on the first wire, $b$ acting on the second wire, and $d$ acting on both wires. The word $cbabdb$ is represented as the circuit
\begin{equation}\label{eq:cbabd_circuit}
  \begin{tikzpicture}[baseline=(current bounding box.center)]
    % Wires/Subsystems
    \draw[thick] (0,2) node[left] {$\clique_1$} -- (7.75,2);
  \draw[thick] (0,1) node[left] {$\clique_2$} -- (7.75,1);
  \filldraw[fill=green!20, draw=green!60, thick] (0.25, 1.7) rectangle (1.25, 2.3) node[midway] {$c$};
   \filldraw[fill=red!20, draw=red!60, thick] (1.5, 0.7) rectangle (2.5, 1.3) node[midway] {$b$};
  \filldraw[fill=violet!20, draw=violet!60, thick] (2.75, 1.7) rectangle (3.75, 2.3) node[midway] {$a$};
  \filldraw[fill=red!20, draw=red!60, thick] (4, 0.7) rectangle (5, 1.3) node[midway] {$b$};
  \filldraw[fill=blue!20, draw=blue!60, thick] (5.25, 0.7) rectangle (6.25, 2.3) node[midway] {$d$};
  \filldraw[fill=red!20, draw=red!60, thick] (6.5, 0.7) rectangle (7.5, 1.3) node[midway] {$b$};
  \end{tikzpicture}
\end{equation}
\end{example}

Since the support $S_b$ and $S_c$ of two commuting letters $(b,c)\in E$ are disjoint, the corresponding gates act on disjoint wires. If they are adjacent in the circuit, they can be slid past each other, representing the fact that we can implement them in the order $bc$ or $cb$. But they can also be executed in parallel, i.e., simultaneously. This has no analog in the word representation, and can be represented in the circuit representation by stacking the two gates vertically instead of placing them side by side\footnote{This explains why the non-commuting graph $\complem\graph$ has a self-loop at each individual letter. Two identical gates, acting on the same wires, cannot be executed simultaneously.}. The following three circuit fragments are thus equivalent and represent the same pc-subword:
\begin{equation}
  \begin{tikzpicture}[baseline=(current bounding box.center)]
    % Wires/Subsystems
  \draw[thick] (0,2) node[left] {} -- (2.75,2);
  \draw[thick] (0,1) node[left] {} -- (2.75,1);
  \filldraw[fill=green!20, draw=green!60, thick] (0.25, 1.7) rectangle (1.25, 2.3) node[midway] {$c$};
   \filldraw[fill=red!20, draw=red!60, thick] (1.5, 0.7) rectangle (2.5, 1.3) node[midway] {$b$};
  \end{tikzpicture}
  \qquad \simeq \qquad
   \begin{tikzpicture}[baseline=(current bounding box.center)]
    % Wires/Subsystems
  \draw[thick] (0,2) node[left] {} -- (1.5,2);
  \draw[thick] (0,1) node[left] {} -- (1.5,1);
  \filldraw[fill=green!20, draw=green!60, thick] (0.25, 1.7) rectangle (1.25, 2.3) node[midway] {$c$};
    \filldraw[fill=red!20, draw=red!60, thick] (0.25, 0.7) rectangle (1.25, 1.3) node[midway] {$b$};
  \end{tikzpicture}  
    \qquad \simeq \qquad
    \begin{tikzpicture}[baseline=(current bounding box.center)]
    % Wires/Subsystems
  \draw[thick] (0,2) node[left] {} -- (2.75,2);
  \draw[thick] (0,1) node[left] {} -- (2.75,1);
  \filldraw[fill=red!20, draw=red!60, thick] (0.25, 0.7) rectangle (1.25, 1.3) node[midway] {$b$};
   \filldraw[fill=green!20, draw=green!60, thick] (1.5, 1.7) rectangle (2.5, 2.3) node[midway] {$c$};
  \end{tikzpicture} 
  \end{equation}

  In contrast, two gates corresponding to non-commuting letters share at least one wire, so they cannot be slid freely past each other. For instance, 
  \begin{equation}
  \begin{tikzpicture}[baseline=(current bounding box.center)]
    % Wires/Subsystems
  \draw[thick] (0,2) node[left] {} -- (2.75,2);
  \draw[thick] (0,1) node[left] {} -- (2.75,1);
  \filldraw[fill=red!20, draw=red!60, thick] (0.25, 0.7) rectangle (1.25, 1.3) node[midway] {$b$};
   \filldraw[fill=blue!20, draw=blue!60, thick] (1.5, 0.7) rectangle (2.5, 2.3) node[midway] {$d$};
  \end{tikzpicture}
  \qquad \not\simeq \qquad
    \begin{tikzpicture}[baseline=(current bounding box.center)]
    % Wires/Subsystems
  \draw[thick] (0,2) node[left] {} -- (2.75,2);
  \draw[thick] (0,1) node[left] {} -- (2.75,1);
  \filldraw[fill=blue!20, draw=blue!60, thick] (0.25, 0.7) rectangle (1.25, 2.3) node[midway] {$d$};
   \filldraw[fill=red!20, draw=red!60, thick] (1.5, 0.7) rectangle (2.5, 1.3) node[midway] {$b$};
  \end{tikzpicture}
\end{equation}

Mathematically, a circuit thus corresponds to a sequence $C=(K_1,K_2,\ldots,K_k)$ of vertical layers. Each layer $K_i$ represents a time-step and consists of a subset of gates in $\xx$ that commute with each other, i.e., $K_i$ is a clique in $\graph$. Gates in the same layer are executed in parallel, while gates in different layers are executed sequentially. When representing a circuit $C$ graphically, each layer $K_i$ consists of gates that act on disjoint sets of wires (otherwise they would not all commute).

The elementary commutation relations in the PC monoid correspond to allowed gate-slide moves in the circuit representation. Two words $v$ and $w$ represent the same pc-monomial $[v]=[w]$ if and only if their circuits $C(v)$ and $C(w)$ are related by a finite
sequence of such slides of adjacent disjoint gates. In this sense, a pc-monomial can be regarded as a circuit modulo local reorderings of
independent gates.

\subsubsection{Left-justified circuits}
A canonical circuit representative is obtained by pushing every gate as far left as the wire-order constraints allow. In the example above, we can rearrange the circuit \eqref{eq:cbabd_circuit} in the following four layers:
\begin{equation}\label{eq:abcd_ljc}
  \begin{tikzpicture}[baseline=(current bounding box.center)]
    % Wires/Subsystems
  \draw[thick] (0,2) node[left] {$\clique_1$} -- (6,2);
  \draw[thick] (0,1) node[left] {$\clique_2$} -- (6,1);
  \filldraw[fill=green!20, draw=green!60, thick] (0.25, 1.7) rectangle (1.25, 2.3) node[midway] {$c$};
   \filldraw[fill=red!20, draw=red!60, thick] (0.25, 0.7) rectangle (1.25, 1.3) node[midway] {$b$};
  \filldraw[fill=violet!20, draw=violet!60, thick] (1.75, 1.7) rectangle (2.75, 2.3) node[midway] {$a$};
  \filldraw[fill=red!20, draw=red!60, thick] (1.75, 0.7) rectangle (2.75, 1.3) node[midway] {$b$};
  \filldraw[fill=blue!20, draw=blue!60, thick] (3.25, 0.7) rectangle (4.25, 2.3) node[midway] {$d$};
  \filldraw[fill=red!20, draw=red!60, thick] (4.75, 0.7) rectangle (5.75, 1.3) node[midway] {$b$};
    \draw[dashed, gray] (1.5, 0.5) -- (1.5, 2.5);
    \draw[dashed, gray] (3, 0.5) -- (3, 2.5);
    \draw[dashed, gray] (4.5, 0.5) -- (4.5, 2.5);
    \node at (0.75, 0) {$\layer_1$};
    \node at (2.25, 0) {$\layer_2$};
    \node at (3.75, 0) {$\layer_3$};
    \node at (5.25, 0) {$\layer_4$};
  \end{tikzpicture}
\end{equation}
We say that such a circuit is left-justified.

Formally, a \emph{left-justified circuit} is a sequence of layers ${L} = (\layer_1, \layer_2, \ldots, \layer_k)$, such that 

\begin{equation}
  \forall i\in\{2,\ldots,k\},\ \forall x\in\layer_i,\ 
  \exists y\in\layer_{i-1}
  \text{ such that } x \text{ and } y \text{ do not commute}.
\end{equation}

This says that the circuit has no avoidable delay: every letter is pushed as far to the left as possible. If a letter $x \in \layer_i$ commuted with all letters in $\layer_{i-1}$, it could have been moved into $\layer_{i-1}$, contradicting the left-justified property.
When representing a left-justified circuit $\layer$ graphically, every gate in layer $\layer_i$ must share a wire with at least one gate in the previous layer $\layer_{i-1}$, ensuring that it cannot be moved to the left of $\layer_{i-1}$.

Given a word $w$, its associated left-justified circuit $\layer(w)$ can be computed in one left-to-right pass through
the word. For a letter $x\in\xx$, let
\begin{equation}
  \operatorname{NC}(x)
  :=
  \{\,z\in\xx:\ z=x\text{ or } (z,x) \notin E\,\}
\end{equation}
be the set of letters that do not commute with $x$ and therefore cannot be moved past it. These sets can be precomputed and then accessed in constant time during the pass through the word.

\begin{algorithm}[H]
\caption{Left-justified circuit of a word}
\label{algo:left_justified_circuit}
\begin{algorithmic}[1]
  \REQUIRE A word $w=y_1\cdots y_r$.
  \STATE $T_x\gets 0$ for all
  $x\in\xx$.
  \FOR{$j=1,\ldots,r$}
    \STATE $t\gets 1+\max\{\,T_z:\ z\in\operatorname{NC}(y_j)\,\}$.
    \STATE Insert $y_j$ into the layer $L_t$ if it exists, or create a new layer $L_t=\{y_j\}$ if it does not exist.
    \STATE $T_{y_j}\gets t$.
  \ENDFOR
  \RETURN the sequence of all layers $L_t$, in increasing order of
  $t$.
\end{algorithmic}
\end{algorithm}

The value $T_x$ stores the latest time (largest layer position) already assigned to an occurrence of
the letter $x$. Therefore the layer position $t$ chosen in Step~3 is exactly one more than the latest time of a previous occurrence that does not commute with
$y_j$. This ensures that $y_j$ is placed in the earliest possible layer, i.e., that the resulting circuit is left-justified. The algorithm runs in time $O(r)$, i.e., linear in the word length $r$, for a fixed alphabet $\xx$ and commutation graph $\graph$.

The map $w\mapsto \layer(w)$ is constant on each class in $\ncmonoid_\graph$, i.e., $L(v) = L(w)$ if $[v]=[w]$. To show this, it suffices to show it for two words $w=u(ab)v$ and $w'=u(ba)v$ that differ by a single application of the commutation relation $(a,b)\in E$ (since any two equivalent words can be connected by a sequence of such relations). Up to the prefix $u$, the two words are identical, so they produce the same layer values and the same $T_x$ values for all letters in $u$. Then in Step~3, when we process the letter $a$ in $w$, we compute $t(a)=1+\max\{T_z : z\in\operatorname{NC}(a)\}$, and when we process the letter $b$, we compute $t(b)=1+\max\{T_z : z\in\operatorname{NC}(b)\}$. Now since $a$ and $b$ commute, we have $a\notin\operatorname{NC}(b)$ and $b\notin\operatorname{NC}(a)$, so the two layer values $t(a)$ and $t(b)$ are independent of each other. Hence the two orders $ab$ and $ba$ produce the same two layer values and the same updated state $(T_x)_{x\in\xx}$. The remaining suffix $v$ is then processed identically, so the final $T_x$ values, and therefore all layer indices, coincide. Thus equivalent words yield the same left-justified circuit.

Conversely, any left-justified circuit $L=(L_1,\ldots,L_k)$ defines a valid class in $\ncmonoid_\graph$. Indeed, reading the layers from left to right and the letters inside each layer in any fixed order, yields a word $w$ for which $\layer(w)=L$. And, since two different left-justified circuits yield two such different words, $w$ is a representative of a unique class in $\ncmonoid_\graph$.

Therefore, left-justified circuits are in bijection with pc-monomials: each class $[w]$ has a unique circuit representative $L(w)$, and two distinct left-justified circuits correspond to two distinct classes.

\subsubsection{Basic operations on left-justified circuits}
We now describe how to perform the basic operations of the PC monoid directly
on left-justified circuits, i.e, how to implement a product between two left-justified circuits satisfying the monoid relation
\begin{eqnarray}
  L(v)\cdot L(w) & = & L(vw), \\
\end{eqnarray}
and an involution of a left-justified circuit satisfying
\begin{equation}
  L(v)^\inv = L(v^\inv).
\end{equation}
Since left-justified circuits are in bijection with pc-monomials, these
procedures implement multiplication and involution directly in the quotient
monoid.

The guiding principle is that both multiplication
and involution can be processed layer by layer and can be implemented by repeated application of a single primitive operation that multiply a left-justified circuit by a single layer of commuting gates.

For efficient layer updates, it is useful to store a left-justified circuit
together with its \emph{tail vector}. Thus, an augmented circuit is a pair
$(L,e)$, where $L=(L_1,\ldots,L_k)$ is the sequence of layers and
$e=(e_1,\ldots,e_r)$ is the tail vector. This vector records, for each wire $i$, the layer index of the last gate
touching that wire:
\begin{equation}
  e_i 
  =
  \max\{\,t:\ \exists x\in L_t \text{ with } i\in S_x\,\},
\end{equation}
with maximum $0$ if the wire has not yet been touched. The vector $e$ is fully
determined by $L$, so it is only cached metadata; keeping it avoids scanning
the whole circuit during repeated products. If only $L$ is given, the tail
vector is obtained by scanning the layers of $L$ from right to left, stopping as
soon as all entries of $e$ have been determined.

We first define the basic primitive operation of multiplying a left-justified circuit by a single layer $M$ of commuting gates. Since the gates in $M$ commute, their
supports are disjoint. Hence their new layer positions can all be computed from
the same tail vector $e$.

\begin{algorithm}[H]
\caption{Product of a left-justified circuit with a single layer}
\label{algo:left_circuit_product}
\begin{algorithmic}[1]
  \REQUIRE Augmented left-justified circuit $(L,e)$ with
  $L=(L_1,\ldots,L_k)$ and $e=(e_1,\ldots,e_r)$, and a layer $M$ of commuting
  gates.
  \FOR {every gate $x\in M$}
  \STATE Compute
  $t_x\gets 1+\max\{\,e_i:\ i\in S_x\,\}$.
  \STATE Insert $x$ into layer $L_{t_x}$, creating a new
  empty layer $L_{k+1}$ if needed.
  \STATE For every $i\in S_x$, set $e_i\gets t_x$.
  \ENDFOR 
  \RETURN The updated pair $(L,e)$.
\end{algorithmic}
\end{algorithm}

The updated circuit is again left-justified, because each
gate $x\in M$ is placed in the first layer strictly after the last gate touching
one of the wires in $S_x$. Note that the for loop can be executed in parallel, since the gates in $M$ commute and therefore have disjoint supports. 
Given an input augment circuit $(L,e)$, let $\Phi_M(L,e)$ denote the output of Algorithm~\ref{algo:left_circuit_product} when multiplying by the layer $M$. 

Let now $K=(K_1,\ldots,K_\ell)$ be another left-justified circuit.
Multiplication by $K$ is obtained by applying the single-layer update to the
layers of $K$ from left to right:
\begin{equation}
  (L,e)\cdot K
  =
  \Phi_{K_\ell}\circ\cdots\circ\Phi_{K_1}(L,e).
\end{equation}
For a fixed alphabet and a fixed edge clique cover, its
cost is linear in the number of gates of $K$ (and even linear in the number $\ell$ of layers of $K$ if the for loop in Algorithm \ref{algo:left_circuit_product} is executed in parallel).

Note that the same primitive also provides a way to construct, alternatively to Algorithm \ref{algo:left_circuit_product}, the augmented left-justified circuit of a word
$w=y_1\cdots y_r$: start from the empty circuit with zero tail vector $(L,e)=(\emptyset,0)$ and
multiply successively by the single-letter layers
$\{y_1\},\ldots,\{y_r\}$.

For the involution operation, let us first define the involution of a single layer $M$:
\begin{equation}
  M^\inv:=\{x^\inv:x\in M\}.
\end{equation}
This is again a layer, because any two letters in $M$ commute if and only if
their involutions commute. Graphically, the operation replaces each gate label
$x$ by $x^\inv$ and moves the gate from the support $S_x$ to the support
$S_{x^\inv}$. Since the edge clique cover is closed under involution, this is
equivalently the support obtained by sending each wire $i$ to its involuted
wire $i^\ast$:
\begin{equation}
  S_{x^\inv}=\{\,i^\ast:\ i\in S_x\,\}.
\end{equation}

Let $L=(L_1,\ldots,L_k)$ be a non-empty left-justified circuit. Since $L$
represents the product $L_1\cdots L_k$, its involution is the reverse product of the involuted layers:
\begin{equation}
  (L_1\cdots L_k)^\inv
  =
  L_k^\inv L_{k-1}^\inv\cdots L_1^\inv .
\end{equation}
Thus the augmented left-justified circuit representing the involution is
obtained by first forming the augmented one-layer circuit
$(L_k^\inv,e^{L_k^\inv})$, where $e^{L_k^\inv}$ denote the tail vector of the one-layer circuit, that is
\begin{equation}
  e_i^{L_k^\inv}
  =
  \begin{cases}
    1, & \text{if there exists } x\in L_k^\inv\text{ such that } i\in S_x,\\
    0, & \text{otherwise.}
  \end{cases}
\end{equation}
One then successively applies the single-layer multiplications
associated with the remaining involuted layers:
\begin{equation}
  (L^\inv,e^\inv)
  =
  \Phi_{L_1^\inv}\circ\cdots\circ\Phi_{L_{k-1}^\inv}
  \bigl(L_k^\inv,e^{L_k^\inv}\bigr).
\end{equation}
The cost is linear in the total size of the supports scanned
during the single-layer updates. In particular, for a fixed alphabet and a fixed
edge clique cover, it is linear in the number of gates (or linear in the number of layers if the for loop in Algorithm \ref{algo:left_circuit_product} is executed in parallel).

\begin{example}[name=Running example,continues=example:abcd_circuit]
The multiplication $cbabdb\cdot ab$ gives in circuit form
\begin{equation}
  \begin{tikzpicture}[baseline=(current bounding box.center),scale=.9,transform shape]
    \draw[thick] (0,2) node[left] {$\clique_1$} -- (4.3,2);
    \draw[thick] (0,1) node[left] {$\clique_2$} -- (4.3,1);

    \filldraw[fill=green!20, draw=green!60, thick]
      (0.25,1.72) rectangle (0.85,2.28) node[midway] {$c$};
    \filldraw[fill=red!20, draw=red!60, thick]
      (0.25,0.72) rectangle (0.85,1.28) node[midway] {$b$};

    \filldraw[fill=violet!20, draw=violet!60, thick]
      (1.30,1.72) rectangle (1.90,2.28) node[midway] {$a$};
    \filldraw[fill=red!20, draw=red!60, thick]
      (1.30,0.72) rectangle (1.90,1.28) node[midway] {$b$};

    \filldraw[fill=blue!20, draw=blue!60, thick]
      (2.35,0.72) rectangle (2.95,2.28) node[midway] {$d$};

    \filldraw[fill=red!20, draw=red!60, thick]
      (3.40,0.72) rectangle (4.00,1.28) node[midway] {$b$};

    \foreach \x in {1.075,2.125,3.175}
      \draw[dashed, gray] (\x,0.5) -- (\x,2.5);

    \node at (0.55,0.15) {$L_1$};
    \node at (1.60,0.15) {$L_2$};
    \node at (2.65,0.15) {$L_3$};
    \node at (3.70,0.15) {$L_4$};
  \end{tikzpicture}
  \quad\cdot\quad
  \begin{tikzpicture}[baseline=(current bounding box.center),scale=.9,transform shape]
    \draw[thick] (0,2) -- (1.15,2);
    \draw[thick] (0,1) -- (1.15,1);

    \filldraw[fill=violet!20, draw=violet!60, thick]
      (0.25,1.72) rectangle (0.85,2.28) node[midway] {$a$};
    \filldraw[fill=red!20, draw=red!60, thick]
      (0.25,0.72) rectangle (0.85,1.28) node[midway] {$b$};

    \node at (0.55,0.15) {$K_1$};
  \end{tikzpicture}
  \quad=\quad
  \begin{tikzpicture}[baseline=(current bounding box.center),scale=.9,transform shape]
    \draw[thick] (0,2) -- (5.35,2);
    \draw[thick] (0,1) -- (5.35,1);

    \filldraw[fill=green!20, draw=green!60, thick]
      (0.25,1.72) rectangle (0.85,2.28) node[midway] {$c$};
    \filldraw[fill=red!20, draw=red!60, thick]
      (0.25,0.72) rectangle (0.85,1.28) node[midway] {$b$};

    \filldraw[fill=violet!20, draw=violet!60, thick]
      (1.30,1.72) rectangle (1.90,2.28) node[midway] {$a$};
    \filldraw[fill=red!20, draw=red!60, thick]
      (1.30,0.72) rectangle (1.90,1.28) node[midway] {$b$};

    \filldraw[fill=blue!20, draw=blue!60, thick]
      (2.35,0.72) rectangle (2.95,2.28) node[midway] {$d$};

    \filldraw[fill=violet!20, draw=violet!60, thick]
      (3.40,1.72) rectangle (4.00,2.28) node[midway] {$a$};
    \filldraw[fill=red!20, draw=red!60, thick]
      (3.40,0.72) rectangle (4.00,1.28) node[midway] {$b$};

    \filldraw[fill=red!20, draw=red!60, thick]
      (4.45,0.72) rectangle (5.05,1.28) node[midway] {$b$};

    \foreach \x in {1.075,2.125,3.175,4.225}
      \draw[dashed, gray] (\x,0.5) -- (\x,2.5);

    \node at (0.55,0.15) {$L_1$};
    \node at (1.60,0.15) {$L_2$};
    \node at (2.65,0.15) {$L_3$};
    \node at (3.70,0.15) {$L_4$};
    \node at (4.75,0.15) {$L_5$};
  \end{tikzpicture}.
\end{equation}
Assuming all letters are Hermitian, here is an example of involution:
\begin{equation}
  \left(\begin{tikzpicture}[baseline=(current bounding box.center),scale=.9,transform shape]
    \draw[thick] (0,2) node[left] {$\clique_1$} -- (2.2,2);
    \draw[thick] (0,1) node[left] {$\clique_2$} -- (2.2,1);

    \filldraw[fill=violet!20, draw=violet!60, thick]
      (0.25,1.72) rectangle (0.85,2.28) node[midway] {$a$};
    \filldraw[fill=red!20, draw=red!60, thick]
      (0.25,0.72) rectangle (0.85,1.28) node[midway] {$b$};

    \filldraw[fill=green!20, draw=green!60, thick]
      (1.30,1.72) rectangle (1.90,2.28) node[midway] {$c$};

    \draw[dashed, gray] (1.075,0.5) -- (1.075,2.5);

    \node at (0.55,0.15) {$L_1$};
    \node at (1.60,0.15) {$L_2$};
  \end{tikzpicture}\right)^\inv
 =
  \begin{tikzpicture}[baseline=(current bounding box.center),scale=.9,transform shape]
    \draw[thick] (0,2) -- (2.2,2);
    \draw[thick] (0,1) -- (2.2,1);

    \filldraw[fill=green!20, draw=green!60, thick]
      (0.25,1.72) rectangle (0.85,2.28) node[midway] {$c$};
    \filldraw[fill=red!20, draw=red!60, thick]
      (0.25,0.72) rectangle (0.85,1.28) node[midway] {$b$};

    \filldraw[fill=violet!20, draw=violet!60, thick]
      (1.30,1.72) rectangle (1.90,2.28) node[midway] {$a$};

    \draw[dashed, gray] (1.075,0.5) -- (1.075,2.5);

    \node at (0.55,0.15) {$L_1$};
    \node at (1.60,0.15) {$L_2$};
  \end{tikzpicture}.
\end{equation}
\end{example}

\begin{example} We provide a simple example where the index involution
$i\mapsto i^\ast$ is non-trivial. Let
\begin{equation}
  \xx=\{x,\bar{x},y\},
  \qquad
  x^\inv=\bar{x},
  \qquad
  \bar{x}^\inv=x,
  \qquad
  y^\inv=y .
\end{equation}
Assume that $x$ and $\bar{x}$ commute, but $y$ does not commute with either of them. The non-commutation graph has two maximal cliques
\begin{equation}
  \clique_1=\{x,y\},
  \qquad
  \clique_2=\{\bar{x},y\}.
\end{equation}
Then
\begin{equation}
  \clique_1^\inv=\clique_2,
  \qquad
  \clique_2^\inv=\clique_1,
\end{equation}
so $1^\ast=2$ and $2^\ast=1$. The supports are
\begin{equation}
  S_x=\{1\},
  \qquad
  S_{\bar{x}}=\{2\},
  \qquad
  S_y=\{1,2\}.
\end{equation}
Thus the involution does not merely relabel gates; it also swaps the supports.
For instance,
\begin{equation}
  \left(\begin{tikzpicture}[baseline=(current bounding box.center),scale=.95,transform shape]
    \draw[thick] (0,2) node[left] {$\clique_1$} -- (2.2,2);
    \draw[thick] (0,1) node[left] {$\clique_2$} -- (2.2,1);

    \filldraw[fill=orange!20, draw=orange!60, thick]
      (0.25,1.72) rectangle (0.85,2.28) node[midway] {$x$};

    \filldraw[fill=blue!20, draw=blue!60, thick]
      (1.30,0.72) rectangle (1.90,2.28) node[midway] {$y$};

    \draw[dashed, gray] (1.075,0.5) -- (1.075,2.5);
  \end{tikzpicture}\right)^\inv = 
  \begin{tikzpicture}[baseline=(current bounding box.center),scale=.95,transform shape]
    \draw[thick] (0,2) node[left] {$\clique_1$} -- (2.2,2);
    \draw[thick] (0,1) node[left] {$\clique_2$} -- (2.2,1);

    \filldraw[fill=blue!20, draw=blue!60, thick]
      (0.25,0.72) rectangle (0.85,2.28) node[midway] {$y$};

    \filldraw[fill=orange!20, draw=orange!60, thick]
      (1.30,0.72) rectangle (1.90,1.28) node[midway] {$\bar{x}$};

    \draw[dashed, gray] (1.075,0.5) -- (1.075,2.5);
  \end{tikzpicture}.
\end{equation}
\end{example}

\subsubsection{Relation to other representations of pc-monomials}\label{sec:other_representations}
The above representation of the elements of $\ncmonoid_\graph$ as left-justified circuits is closely related to other representations of pc-monomials that have appeared in the literature.

\paragraph{Foata normal form.}
Given a word $w$, its left-justified circuit representation $L(w) = (L_1,\ldots,L_k)$ can be represented back as a word $v = v_{L_1} v_{L_2} \cdots v_{L_k}$, where $v_{L_i}$ is the subword obtained by concatenating the letters in $L_i$ in lexicographic order. This word $v$ is a representative of the same class as $w$, and is known as the Foata, or Cartier-Foata, normal form $\foa(w):=v$ \cite{cartier1969applications}. The left-justified circuit representation and Foata normal forms are also closely related to the \emph{heap} representation of pc-monomials by Viennot \cite{viennot1986heaps}.

\begin{example}[name=Running example,continues=example:abcd_circuit]
  The Foata normal form obtained from the left-justified circuit representation \eqref{eq:abcd_ljc} of the word $cbabdb$ is $\foa(cbabdb) = bcabdb$, or $bc|ab|d|b$ where we use $|$ to separate the layers. Note that the Foata normal form is different from the lexicographic normal form $\lex(cbabdb)=bbcadb$ obtained in \eqref{eq:cbabd_lex}.
\end{example}

\paragraph{Occurrence graph and dependency graph.}
A circuit is a particular time-ordered representation of a pc-monomial, where different gate arrangements compatible with the gate-slide moves lead to different time ordering. It is also possible to represent a pc-monomial by a more abstract graph representation that captures the fixed relative orders between gates without explicitly encoding a time ordering. Such a representation can be obtained by taking the ``immediate wire-successor graph" of any circuit representation, where one interprets each gate occurrence as a vertex and draws a directed edge from one occurrence to another if the first occurrence is an immediate predecessor of the second along some wire in the circuit. The resulting graph is the \emph{occurrence graph} introduced by Duboc \cite{duboc1986some}.

Explicitely, let $C$ be an arbitrary circuit representation of some word $w$. Let $|x|_w$ denote the number of occurrences of the gate $x$ in $C$ (equivalently of the letter $x$ in $w$). For each gate $x\in\xx$, we can label the occurrences of $x$ in $C$ from left to right as $(x,1), (x,2), \ldots, (x,|x|_w)$. This labelling is well-defined and independent of the choice of circuit representation $C$ of $w$ (i.e., left-jusfitifed, right-justified or other), since two gates of the same type $x$ must necessarily be consecutive and their relative order cannot be changed by any of the allowed gate-slide moves.

Then the occurrence graph $\Gamma$ is the directed acyclic graph (DAG) constructed as follows.
\begin{enumerate}
  \item Vertices: the vertices are the gate occurrences, i.e., the set 
  \begin{equation}\Omega = \{(x,k): x\in\xx, k=1,\ldots,|x|_w\}
  \end{equation} 
    \item Edges: there is a directed edge
    \begin{equation}
    (x,k)\to(y,\ell)
    \end{equation}
    if there exists a wire $i$ such that the two occurrences $(x,k)$ and $(y,\ell)$ are immediately consecutive along wire $i$. Equivalently, $i\in S_x\cap S_y$, the gate $(x,k)$ stands on the left of gate $(y,\ell)$, and there is no other gate between the gate $(x,k)$ and $(y,\ell)$ on wire $i$.
\end{enumerate}
The DAG $\Gamma$ captures the immediate dependencies between gate occurrences in the circuit. Note that any two circuit representations $C$ and $C'$ of two words $w$ and $w'$ in the same class $[w]=[w']$ yield the same occurrence graph $\Gamma=\Gamma'$ since the allowed gate-slide moves do not change the relative order of gates and the immediate successor relations on any given wire. 
Thus, the occurrence graph provides a canonical representation of pc-monomials $[w]$.

From the occurrence graph $\Gamma$, on can read off the left-justified circuit representation $L$ of $[w]$ by performing a topological sort of the vertices of $\Gamma$ and grouping together in the same layer all vertices that have the same longest path length from the source vertices (i.e., vertices with no incoming edges).  Hence, the occurrence graph and left-justified circuit representations are equivalent representations of pc-monomials.

One can also define related graphs that carry the same information. In particular, the dependency graph $\bar \Gamma$  captures all dependencies between gate occurrences, not just the immediate ones \cite{duboc1986some}. It is obtained by adding a directed edge from $(x,k)$ to $(y,\ell)$ whenever the two gates share a wire $i$ and $(x,k)$ occurs before $(y,\ell)$, regardless of whether they are immediate successors along that wire. 

\begin{example}[name=Running example,continues=example:abcd_circuit]
  The occurrence graph associated to the circuit \eqref{eq:abcd_ljc} is
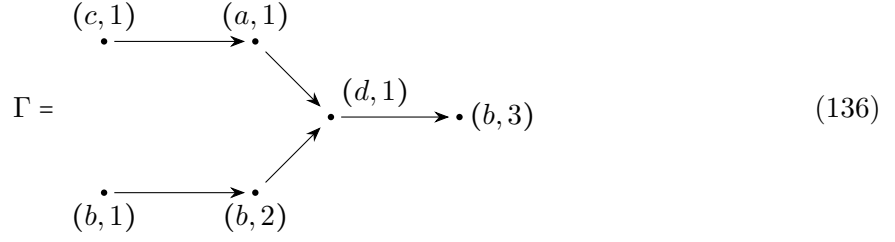
  \begin{equation}
    \Gamma = 
    \begin{tikzpicture}[baseline=(current bounding box.center)]
      \node (c) at (0,2) {};
      \node (b1) at (0,0) {};
      \node (a) at (2,2) {};
      \node (b2) at (2,0) {};
      \node (d) at (3,1) {};
      \node (b3) at (4.7,1) {};
      \fill (c) circle (1.2pt);
      \fill (b1) circle (1.2pt);
      \fill (a) circle (1.2pt);
      \fill (b2) circle (1.2pt);
      \fill (d) circle (1.2pt);
      \fill (b3) circle (1.2pt);
      \node[above] at (c) {$(c,1)$};
      \node[below] at (b1) {$(b,1)$};
      \node[above] at (a) {$(a,1)$};
      \node[below] at (b2) {$(b,2)$};
      \node[above right] at (d) {$(d,1)$};
      \node[right] at (b3) {$(b,3)$};
      \draw[-{Stealth[length=2mm]}] (c) -- (a);
      \draw[-{Stealth[length=2mm]}] (b1) -- (b2);
      \draw[-{Stealth[length=2mm]}] (a) -- (d);
      \draw[-{Stealth[length=2mm]}] (b2) -- (d);
      \draw[-{Stealth[length=2mm]}] (d) -- (b3);
    \end{tikzpicture}
  \end{equation}

The dependency graph is obtained by adding the edges $(c,1)\to(d,1)$, $(b,1)\to(d,1)$, $(b,1)\to(b,3)$, and $(b,2)\to(b,3)$ to the above occurrence graph:
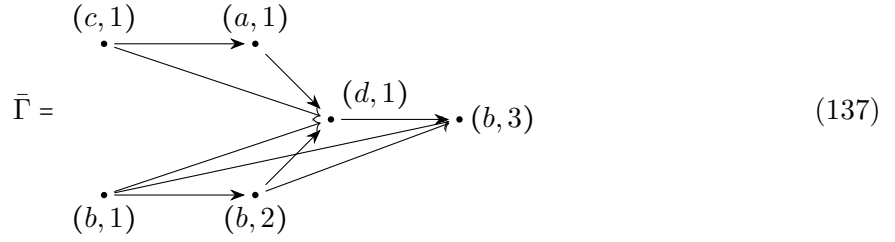
\begin{equation}
\bar \Gamma = 
  \begin{tikzpicture}[baseline=(current bounding box.center)]
    \node (c) at (0,2) {};
    \node (b1) at (0,0) {};
    \node (a) at (2,2) {};
    \node (b2) at (2,0) {};
    \node (d) at (3,1) {};
    \node (b3) at (4.7  ,1) {};
    \fill (c) circle (1.2pt);
    \fill (b1) circle (1.2pt);
    \fill (a) circle (1.2pt);
    \fill (b2) circle (1.2pt);
    \fill (d) circle (1.2pt);
    \fill (b3) circle (1.2pt);
    \node[above] at (c) {$(c,1)$};
    \node[below] at (b1) {$(b,1)$};
    \node[above] at (a) {$(a,1)$};
    \node[below] at (b2) {$(b,2)$};
    \node[above right] at (d) {$(d,1)$};
    \node[right] at (b3) {$(b,3)$};
    \draw[-{Stealth[length=2mm]}] (c) -- (a);
    \draw[-{Stealth[length=2mm]}] (b1) -- (b2);
    \draw[-{Stealth[length=2mm]}] (a) -- (d);
    \draw[-{Stealth[length=2mm]}] (b2) -- (d);
    \draw[-{Stealth[length=2mm]}] (d) -- (b3);
    % Additional edges for the dependency graph
    \draw[->,] (c) -- (d);
    \draw[->,] (b1) -- (d);
    \draw[->,] (b1) -- (b3);
    \draw[->,] (b2) -- (b3);
  \end{tikzpicture}
\end{equation}
\end{example}

\subsection{Wire representation and clique projections}
\label{subsec:wire_projections}

The circuit representation described above reads a pc-monomial vertically: a
left-justified circuit is a sequence of layers $L=(L_1,\ldots,L_k)$,
where each layer consists of pairwise commuting gates. There is a dual way of
reading the same circuit. Instead of reading it layer by layer, one may read it
wire by wire and represent it by a tuple $W=(w_1,\ldots,w_m)$ where each $w_i=(w_{i_1},\ldots,w_{i_{r_i}})$ is the sequence of gates encountered along the wire $\clique_i$ from left to right. This is the \emph{wire representation} of a pc-monomial, and the words $w_i$ are its \emph{wire words}.

\begin{example}[name=Running example,continues=example:abcd_circuit]
  The word $cbabdb$, corresponding to the circuit \eqref{eq:abcd_ljc}, has the layer representation $L = cb|ab|d|b$ (where we use $|$ to separate layers), and the wire representation $W=(cad,bbdb)$.  
\end{example}
Note that the wire representation is independent of the circuit representation (i.e., left-justified, right-justified, or any other), because the sequences of gates along a wire are not affected by swaps of commuting gates, since they have disjoint supports.

Given a word $v=y_1\cdots y_r$, its wire representation can be immediately computed, without building first the circuit representation. For this, define the projection of $v$ on the wire $\clique_i$ by deleting all letters that do not act on that wire:
\begin{equation}\label{eq:wire_projection}
  \pi_i(v)
  :=
  y_{j_1}\cdots y_{j_s},
  \qquad
  j_1<\cdots<j_s,
  \qquad
  y_{j_q}\in\clique_i .
\end{equation}
Notice that $\pi_i(v)$ is a fully non-commutative word in $\ncmonoid_i$, the free monoid generated by the subalphabet corresponding to the clique $\clique_i \subset \xx$. The wire words $\pi_i(v)$ are usually referred in the PC monoid litterature as the \emph{clique projections} of $v$. The wire (or clique) representation $W$ of $v$ is then the tuple\begin{equation}
 \Pi(v)
  :=
  \bigl(\pi_1(v),\ldots,\pi_m(v)\bigr)\in \ncmonoid_1\times\cdots\times\ncmonoid_m.
\end{equation}
It is immediate to implement the projections \eqref{eq:wire_projection} to compute the wire representation $\Pi(v)$ from a word $v$. For a
fixed alphabet and a fixed edge clique cover, the cost is linear in the number of
gate occurrences.

The clique projection construction is standard in the theory of partially commutative monoids and provides a convenient way to represent pc-monomials \cite[Proposition 1.2]{duboc1986some}.

\begin{proposition}[Projection theorem]
\label{prop:wire_projection_theorem}
For words $u,v\in\ncmonoid$,
\begin{equation}
  [u]=[v]
  \qquad\Longleftrightarrow\qquad
  \Pi(u)=\Pi(v).
\end{equation}
\end{proposition}

\begin{proof}
The direct implication is immediate. If $[u]=[v]$, then $u$ and $v$ are related by swaps of adjacent commuting letters. As we noticed above, such swaps do not change any wire word, hence $\Pi(u)=\Pi(v)$. The reverse implication can be easily obtained by induction on the length of $u$. 
\end{proof}

\subsubsection{Reconstructible wire words}
By the above propostion, the map $\Pi\,:\, \ncmonoid\to\ncmonoid_1\times\cdots\times\ncmonoid_m$ is injective. However, it is not a bijection. 
Indeed, an arbitrary tuple of wire words $W=(w_1,\ldots,w_m) \in \ncmonoid_1 \times \cdots \times \ncmonoid_m$ need not arise from a global word $v$ such that $\Pi(v)=W$. In other words, not every tuple of wire words is \emph{reconstructible}.

A first necessary condition is that occurrences of a letter that appears on
several wires must be consistently matched across those wires.  For a tuple $W=(w_1,\ldots,w_m)$, let $|x|_{w_i} $ denote the number of occurrences of a letter $x$ in the wire word $w_i$.
The tuple $W$
is called \emph{quasi-reconstructible} if \emph{(i)} the letter $x$ appears only on the wire in its support $S_x$, i.e., $|x|_{w_i}=0$ for all $i\notin S_x$, and \emph{(ii)} the number of occurrences of $x$ is the same in all wire words $w_i$ in its support $S_x$, i.e.,
\begin{equation}
  |x|_{w_{i_1}} = \ldots = |x|_{w_{i_k}}:=|x|_{W}\qquad \text{for all } i_j \in S_x.
\end{equation}
In that case the occurrences of $x$ can be matched across the wires in $S_x$ and we can label the occurences of $x$ in $W$ as $(x,1),(x,2),\ldots,(x,|x|_W)$.

A quasi-reconstructible tuple of wire words defines a \emph{pre-circuit}: the
global gate occurrences are known, and each wire carries a prescribed horizontal
order of the gates acting on it. But a circuit comes with an extra global time ordering that assigns a layer coordinate to each gate occurrence. The existence of such a global time ordering is not guaranteed for an arbitrary pre-circuit, and is the extra condition that distinguishes reconstructible tuples of wire words from quasi-reconstructible ones.

As in subsection~\ref{sec:other_representations}, we can assign to a pre-circuit, or equivalently to a quasi-reconstructible tuple of wire words, an occurrence graph $\Gamma$ that captures the immediate successor relations along the wires. As before, its vertices are the gate occurrences $\Omega = \{(x,k)\,:\, x\in \xx, k\in \{1,\ldots,|x|_W\}\}$ and there is a directed edge from $(x,k)$ to $(y,\ell)$ if $(y,\ell)$ is the immediate successor of $(x,k)$ along some wire word $w_i$.

\begin{proposition} \cite[Proposition 1.4]{duboc1986some} \label{prop:reconstructibility}
Let $W=(w_1,\ldots,w_m)$ be a tuple of wire words. Then it is reconstructible if and only if it is quasi-reconstructible and the associated occurrence graph $\Gamma$ is acyclic.
\end{proposition}

Indeed, it is clear that if $W$ is reconstructible, then the occurrence graph $\Gamma$ is acyclic. Conversely, if $\Gamma$ is acyclic, then it defines a partial order on the gate occurrences $\Omega$. Any total order compatible with this partial order defines a valid time ordering of the gate occurrences, yielding a global word $v$ such that $\Pi(v)=W$.

Acyclicity implies, in particular, consistency on intersections. Namely, if
$W=(w_1,\ldots,w_m)$ is reconstructible, then for all $i,j$,
\begin{equation}
  \pi_{ij}(w_i)=\pi_{ij}(w_j),
\end{equation}
where $\pi_{ij}$ denotes the projection onto
$\clique_i\cap\clique_j$. Indeed, if $W$ is reconstructible, then there exists a global word $w$ such that $\pi_{k}(w)=w_k$. Then $\pi_{ij}(w_i) = \pi_{ij}(w) = \pi_{ij}(w_j)$. This condition is
necessary but not sufficient: the tuple may agree on all intersections and
still impose a cyclic order obstruction, as in the example below.

\begin{example} Let $\xx = \{a,b,c,d\}$ and $\complem\graph$ have edges $\{\{a,b\},\{b,c\},\{c,d\},\{d,a\}\}$, so an edge clique cover is given by $\clique_1=\{a,b\}$, $\clique_2=\{b,c\}$, $\clique_3=\{c,d\}$, $\clique_4=\{d,a\}$.
  
Consider the tuple of clique words  $v = (ab, bc, cd, da)$. It is quasi-reconstructible, because each letter $a,b,c,d$ occurs exactly once on each wire on which it appears.

The associated occurrence graph is
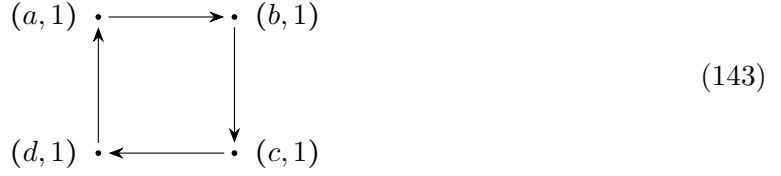
\begin{equation}
  \begin{tikzpicture}[baseline=(current bounding box.center),scale=0.9]
    \node[label=left:{$(a,1)$}] (a) at (0,2) {};
    \node[label=right:{$(b,1)$}] (b) at (2,2) {};
    \node[label=right:{$(c,1)$}] (c) at (2,0) {};
    \node[label=left:{$(d,1)$}] (d) at (0,0) {};
    \fill (a) circle (1.2pt);
    \fill (b) circle (1.2pt);
    \fill (c) circle (1.2pt);
    \fill (d) circle (1.2pt);
    \draw[-{Stealth[length=2mm]}] (a) -- (b);
    \draw[-{Stealth[length=2mm]}] (b) -- (c);
    \draw[-{Stealth[length=2mm]}] (c) -- (d);
    \draw[-{Stealth[length=2mm]}] (d) -- (a);
  \end{tikzpicture}
\end{equation}
Since this graph is cyclic, there exist no global word $w$ such that $\Pi(w)=v$. 

Consider, instead the tuple of clique words  $X = (ab, bc, cd, ad)$. It is also quasi-reconstructible and its occurrence graph is
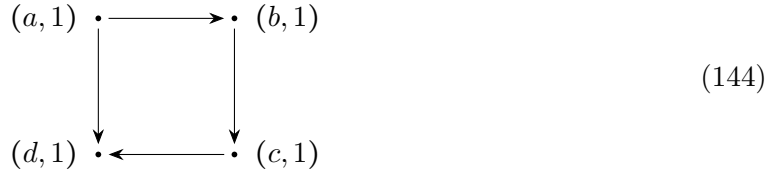
\begin{equation}
  \begin{tikzpicture}[baseline=(current bounding box.center), scale=0.9]
    \node[label=left:{$(a,1)$}] (a) at (0,2) {};
    \node[label=right:{$(b,1)$}] (b) at (2,2) {};
    \node[label=right:{$(c,1)$}] (c) at (2,0) {};
    \node[label=left:{$(d,1)$}] (d) at (0,0) {};
    \fill (a) circle (1.2pt);
    \fill (b) circle (1.2pt);
    \fill (c) circle (1.2pt);
    \fill (d) circle (1.2pt);  
    \draw[-{Stealth[length=2mm]}] (a) -- (b);
    \draw[-{Stealth[length=2mm]}] (b) -- (c);
    \draw[-{Stealth[length=2mm]}] (c) -- (d);
    \draw[-{Stealth[length=2mm]}] (a) -- (d);
  \end{tikzpicture}
\end{equation}
Since this graph is acyclic, there exists a global word $w$ such that $\Pi(w)=X$. For instance, $w=abcd$ works. 
\end{example}

While Proposition~\ref{prop:reconstructibility} provides a criterion for reconstructibility, it is not directly algorithmic. We present below a simple algorithm from \cite{liu1990efficient} that tests whether a quasi-reconstructible tuple of wire words is reconstructible and, if yes, produces a left-justified circuit representing it.

The procedure scans each wire word from left to right and repeatedly extracts the next feasible layer. For a tuple $W=(w_1,\ldots,w_m)$, let $h_i$ be the current head position in $w_i$, i.e., the first not-yet-consumed letter. Equivalently, $(w_i)_1\cdots (w_i)_{h_i-1}$ has already been assigned to previous layers, while $(w_i)_{h_i}\cdots (w_i)_{|w_i|}$ remains to be placed. We then define
\begin{equation}
  F(W;h)
  :=
  \{\,x\in\xx:\ \forall i\in S_x, (w_i)_{h_i}=x\,\}.
\end{equation}
Thus $F(W;h)$ is the set of letters that are simultaneously visible at the current head position of every wire on which they act; these are exactly the letters that can be placed next.

\begin{algorithm}[H]
\caption{Reconstruction from wire projections}
\label{algo:wire_projection_reconstruction}
\begin{algorithmic}[1]
  \REQUIRE A quasi-reconstructible tuple $W=(w_1,\ldots,w_m)$.
  \ENSURE Either \textsc{False}, or \textsc{True} together with an augmented
  left-justified circuit $(L,e)$ representing $W$.
  \STATE Initialize $h_i\gets 1$ and $e_i\gets 0$ for $i=1,\ldots,m$.
  \STATE Initialize $L\gets()$.
  \WHILE{there exists $i$ such that $h_i\le |w_i|$}
    \STATE $A\gets F(W;h)$.
    \IF{$A=\emptyset$}
      \RETURN \textsc{False}.
    \ELSE
      \STATE Append the layer $A$ to $L$.
      \STATE Let $t$ be the index of the newly appended layer.
      \STATE For each $x\in A$ and each $i\in S_x$, increment $h_i$ by $1$
      and set $e_i\gets t$.
    \ENDIF
  \ENDWHILE
  \RETURN \textsc{True}, together with $(L,e)$.
\end{algorithmic}
\end{algorithm}

At any stage, the head positions $h_i$ describe the part of the pre-circuit that has
not yet been placed. If the remaining occurrence graph is acyclic, it has at
least one source, i.e., one remaining gate occurrence with no predecessor. Such
a source must be visible at the front of every wire in its support, and hence
belongs to $F(W;h)$. Conversely, every letter in $F(W;h)$ corresponds to a
remaining gate occurrence with no predecessor. Thus $F(W;h)$ is precisely the
set of gates that can be placed in the next layer.

After appending this layer, the algorithm advances one step along every wire
touched by the gates in $F(W;h)$. This removes the corresponding source
vertices from the remaining occurrence graph. Repeating this procedure is
therefore the same as peeling off the acyclic graph layer by layer. If the
algorithm exhausts all wire words, it has constructed the left-justified circuit
associated with the earliest possible time assignment. If it gets stuck with
$F(W;h)=\emptyset$ while some wire word is not exhausted, then the remaining
graph has no source, hence contains a directed cycle. Therefore the algorithm
succeeds exactly when the occurrence graph is acyclic, and in that case the
returned pair $(L,e)$ is the augmented left-justified circuit representing $W$.

Let  $N:=\sum_{i=1}^m |w_i|$ 
be the total size of the tuple of wire words. With a queue/counter
implementation, Algorithm~\ref{algo:wire_projection_reconstruction} runs in
linear time in $N$, up to constants depending only on the alphabet and the
chosen clique cover.
Indeed, for each letter $x$, one can maintain how many wires in $S_x$ currently
display $x$ at their front. A letter $x$ is ready precisely when this counter is
equal to $|S_x|$. The ready letters are stored in a queue. When a layer $A$ is
removed, the algorithm advances the heads only on the wires touched by the
letters in $A$, and updates the counters of the newly exposed front letters.
Each occurrence in each wire word is exposed once and consumed once. Hence the
total number of counter updates is $O(N)$ for a fixed alphabet and fixed clique
cover.

\subsubsection{Basic operations}

Multiplication and involution of pc-monomials can be implemented directly on
wire projections, without reconstructing either a global word or a circuit.
They are especially simple in this representation.

Let
\begin{equation}
  \Pi(u)=(u'_1,\ldots,u'_r),
  \qquad
  \Pi(v)=(v'_1,\ldots,v'_r)
\end{equation}
be two wire representations. Then multiplication is componentwise concatenation:
\begin{equation}
  \Pi(uv)
  =
  \Pi(u)\Pi(v)
  =
  (u'_1v'_1,\ldots,u'_rv'_r).
\end{equation}

For the involution, recall that since the edge clique cover is closed under involution, each wire $i$ has an involuted wire $i^\ast$ satisfying $\clique_i^\inv=\clique_{i^\inv}$.
If $v=y_1\cdots y_s$ is a word, we write
\begin{equation}
  v^\inv:=y_s^\inv\cdots y_1^\inv .
\end{equation}
Then, for every wire $i$,
\begin{equation}
  \pi_i(v^\inv)
  =
  \bigl(\pi_{i^\inv}(v)\bigr)^\inv .
\end{equation}
Indeed, the wire $i$ in $v^\inv$ corresponds to the involuted wire $i^\inv$ in
$v$, and the order of the local word is reversed.

Therefore, if $\Pi(v)=(v'_1,\ldots,v'_r)$, then
\begin{equation}
  \Pi(v^\inv)
  =
  \bigl((v'_{1^\inv})^\inv,\ldots,(v'_{r^\inv})^\inv\bigr).
\end{equation}
If all variables are Hermitian, this simply
reverses each wire word.

\begin{example}[name=Running example,continues=example:abcd_circuit]
In the wire representation, the product $cbabdb\cdot ab$ is $(cad,bbdb)(a,b) = (cada,bbdbb)$. The involution of $cbabdb$ is $(cad,bbdb)^\inv = (dac,bdbb)$.
\end{example}

Note that when wire words are computed from an actual word, or equivalently from an
actual circuit, quasi-reconstructibility and reconstructibility are automatic.
Thus, if wire projections are used as the data structure and are initialized
from a valid pc-monomial, one never needs to test reconstructibility after
each multiplication and involution operation since they preserve validity.

\subsection{Implementation of PCPO}
\label{subsec:pcpo_implementation}
The preceding representations can be used in practice to build
SDP relaxations for PCPO following the hybrid approach discussed in subsection~\ref{sec:hybrid_approaches}, where the partial commutation relations are handled by working in the quotient algebra $\kk[\ncmonoid_\graph]$. The main
implementation choice is how to represent pc-monomials. The
most convenient representation is the wire representation
\begin{equation}
  \Pi(w)=(\pi_1(w),\ldots,\pi_m(w)).  
\end{equation}
This representation is canonical in the quotient monoid, supports direct
multiplication and involution, and avoids repeated normalization of word
representatives.

At a given relaxation level, the first task is to build a monomial basis. For
the degree-truncated relaxation,  the set of retained moments are indexed by the pc-monomials of degree $2d$: 
\begin{equation}
  B_{2d}=\ncmonoid_{\graph,{2d}}
  =
  \{\, [w] : |w|\leq 2d\,\}.
\end{equation}
Using the wire representation, we identify this set with its image under
$\Pi$:
\begin{equation}
  B_{2d}
  \simeq
  \Pi(B_{2d})
  =
  \{\, W:\ W=\Pi(w)\text{ for some }w,\ \deg(w)\leq 2d\,\}.
\end{equation}
In an implementation, each element of $B_{2d}$ is thus stored as a tuple $W=(w_1,\ldots,w_m)$ of wire words.
Since the wire representation is injective, this tuple can be used as a
canonical key. The basis can be stored as a hash table, mapping each tuple of wire words to its index in the basis. This allows for efficient lookup and insertion of new basis elements.

The basis $B_{2d}$ can be generated incrementally. Starting from the empty tuple
$\Pi(\id)$, one multiplies on the right by the single-letter wire tuples
$\Pi(x)$, for $x\in\xx$. This multiplication simply amounts to appending the letter $x$ to the wires in $S_x$. 
Different insertion orders may produce the same pc-monomial, i.e., the same wire tuple, but then they are automatically merged by the hash table.

Some bookeeping is useful for the degree. Since a letter may appear on several wires, the degree of a basis element is not simply the sum of the lengths of its wire words. One has to count the number of matched occurrences of each letter in the wire words. In
practice, it is simpler to cache the degree during basis generation: the degree
is initialized to $0$ for $\Pi(\id)$ and is incremented by $1$ whenever one
multiplies by a letter.

Every input polynomial can then be represented as a sparse map from wire tuples to
coefficients. If the problem is already given in the quotient algebra
$\kk[\ncmonoid_\graph]$, no further normalization is needed. If instead the
input is given in the full free algebra, each word monomial is first projected
to the quotient by computing its wire representation, and coefficients of words
with the same wire tuple are collected. Thus, for instance, an objective
polynomial
\begin{equation}
  p=\sum_w p_w w
\end{equation}
is converted into the quotient polynomial
\begin{equation}
  [p]
  =
  \sum_W
  \left(
    \sum_{w:\,\Pi(w)=W} p_w
  \right) W = \sum_W p_W W,
\end{equation}
where $W$ ranges over valid wire tuples in the basis. The same projection is applied to the
constraint polynomials $q$ and $r$. 

The moment and localizing matrices are then built directly in the wire representation, using the multiplication and involution rules described above.
\begin{equation}
  M_d(q,y)_{V,W}
  =
 \sum_U q_U y_{V^\inv U W},
  \qquad
  V,W\in B_{d-\lceil \deg(q)/2\rceil}.
\end{equation}
Here $V^\inv U W$ is computed directly in the wire representation, by applying
the wire-level involution to $V$ and then componentwise concatenating with
$U$ and $W$. 

The additional equality constraints $R$ are imposed using the same principle. One adds the linear constraint
\begin{equation}
  \sum_U r_U\, y_{V^\inv U W}=0,\qquad V,W\in B_{d-\lceil \deg(r)/2\rceil},
\end{equation}
where $V^\inv U W$ is computed by wire-level involution and multiplication.

In this representation, equality testing of pc-monomials is just equality of
keys. There is no need to compare word representatives or to run a normal-form
algorithm. Moreover, if all monomials are initialized from valid wire tuples,
one does not need to test reconstructibility during SDP construction:
multiplication and involution preserve valid wire representations.

Note that the moment sequence $y$ satisfies the involutive symmetry
\begin{equation}
  y_{W^\inv}=\overline{y_W},
\end{equation}
where $\overline{y_W}$ denotes the complex conjugate of $y_W$. This immediately follows from the definition of the moment variables as $y_W=L(W)$ for a positive linear functional $L$ on $\kk[\ncmonoid_\graph]$, or equivalently from the fact that the moment matrix $M_d(y)$ is Hermitian.

In practice, one can thus store each moment variables modulo the involution
$W\mapsto W^\inv$. For each wire tuple $W$, choose a canonical representative $\tilde W$ of the orbit $\{W,W^\inv\}$ under the involution, for instance by taking the minimum of $W$ and $W^\inv$ with respect to any fixed total order on wire tuples:
\begin{equation}
 \tilde W
  :=
  \min\{W,W^\inv\}\,.
\end{equation}
Then store one complex moment variable
\begin{equation}
  y_{\tilde W}=a_{\tilde W}+i b_{\tilde W},
\end{equation}
represented by the two real scalar variables $a_{\tilde W}$ and $b_{\tilde W}$. A lookup of
$y_W$ returns $a_{\tilde W}+i b_{\tilde W}$ if $W=\tilde W$, and $a_{\tilde W}-i b_{\tilde W}$ if
$W=(\tilde W)^\inv$. If $W=W^\inv$, or if the underlying field $\kk=\R$ then only the real
variable $a_{\tilde W}$ is stored. In this way, the conjugation symmetry is built into the
parametrization, resulting in a more compact representation.

 An implementation of PCPO based on the wire representation has been developed by two of the authors and is available at \url{https://github.com/abh1mishra/PCPOP.jl}.

\subsubsection{Special cases}
If the alphabet is fully non-commutative, then the non-commutation graph has a single maximal clique.
There is only one wire, and the wire representation is just the original
non-commutative word.

At the opposite extreme, if the alphabet is fully
commutative, then one may take one wire per letter. The wire representation then
records only the number of occurrences of each letter, and is exactly the usual
exponent-vector representation of commutative monomials presented in subsection~\ref{subsec:cpo_special_case}.

More generally, if the non-commutation graph $\complem\graph$ decomposes as a disjoint union of
connected components $\complem\graph_1,\ldots,\complem\graph_t$  then the PC monoid decomposes as a direct product of the
PC monoids associated with these components. Letters belonging to
different components commute with one another, and the wire representation
splits into independent blocks. In this situation, basis generation, multiplication, and involution can be performed component by component.

The $m$-partite standard Bell scenario is an important example of this decomposition. The
non-commutation graph is a disjoint union of $m$ complete graphs, one for each
party. Hence the corresponding PC monoid is a direct product of local free
monoids: within each party the measurement operators are fully
non-commutative, while operators belonging to different parties commute. The
wire representation therefore stores one local non-commutative word per party,
which is precisely the natural indexing structure for Bell-type moment
matrices. 

% The expected reduction in SDP size can be read directly from the number of
% pc-monomials at each degree. Let $a_k(\graph)$ denote the number of
% pc-monomials of degree exactly $k$, or equivalently the number of commutation
% classes of words of length $k$. The trace-monoid growth series is
% \begin{equation}
%   A_\graph(t)
%   :=
%   \sum_{k\ge 0} a_k(\graph)t^k
%   =
%   \left(
%   \sum_{\substack{C\subseteq \xx\\ C\text{ clique in }\graph}}
%   (-1)^{|C|}t^{|C|}
%   \right)^{-1},
% \end{equation}
% where the sum in the denominator runs over all complete subgraphs of the
% commutation graph, including the empty clique. Thus the row-and-column size of
% the degree-$d$ moment matrix is
% \begin{equation}
%   N_d(\graph)
%   =
%   \sum_{k=0}^d a_k(\graph),
% \end{equation}
% while the retained moment variables are indexed, before exploiting involution
% symmetry and the additional constraints $R$, by at most
% $N_{2d}(\graph)$ pc-monomials. Adding commutation edges to $\graph$ can only
% merge word classes, hence can only decrease these numbers. The two extreme
% cases are
% \begin{equation}
%   N_d(\emptyset)
%   =
%   \sum_{k=0}^d |\xx|^k,
%   \qquad
%   N_d(K_{\xx})
%   =
%   \binom{|\xx|+d}{d},
% \end{equation}
% corresponding respectively to the fully non-commutative and fully commutative
% hierarchies. Partial commutation therefore interpolates between the exponential
% word growth of the free monoid and the polynomial growth of the commutative
% monoid, with the actual saving governed by the clique structure of $\graph$.

  \section{Tracial partially commutative polynomial optimization}
  \label{sec:tracial_pcpo}

  One can consider a tracial version of PCPO, in which the objective is
  evaluated with respect to a trace rather than a vector state \cite{burgdorf2013tracial, klep2016constrained}.
  This is useful in several quantum information applications; see, for instance,
  \cite{navascues2015bounding,gribling2018bounds, Tavakoli2022informationally,roch2025prepare} for examples. % Tracial moment matrix implicitly appear in NV15 Equation 9.

  We write $\operatorname{tr}$ for a normalized trace, so that
  $\operatorname{tr}(\mathds 1)=1$. Note that in general, a trace is not well-defined for any bounded operators in $\mathcal B(H)$ when $H$ is infinite dimensional. So in the tracial version of PCPO, we restrict the optimization to operator representations $A_x$ in tracial von Neumann algebras. The tracial PCPO problem is then the following. 
  \begin{tcolorbox}[
  title=TPCPO problem,
  colframe=gray!60,
  colback=white,
  coltitle=black,
  colbacktitle=gray!60]
  \begin{align}
  p_{\operatorname{tr}}^{\mathrm{opt}}
  =
  \inf_A\quad
  & \operatorname{tr}(p(A))\nonumber \\
  \text{s.t.}\quad
  & A_xA_y=A_yA_x
  \qquad \forall \{x,y\}\in E, \\
  & q(A)\ge 0
  \qquad \forall q\in Q,\nonumber\\
  & r(A)=0
  \qquad \forall r\in R,\nonumber
  \end{align}
  where the optimization is over all tracial von Neumann algebra operator representations $A$ of the variables in $\xx$.
  \end{tcolorbox}

  Any feasible solution of the tracial PCPO problem defines a linear functional $L:\ncalgebra\to\kk$ by $L(f)=\operatorname{tr}(f(A))$. This functional satisfies the same normalization, positivity, and linear constraints as in the non-tracial case, but it also satisfies the additional cyclicity constraint
\begin{equation}
  L(fg)=L(gf)
  \qquad
  \forall f,g\in\ncalgebra,
\end{equation}
which originates from the cyclic property of the trace. 
Similarly to the non-tracial case, one can instead optimize over such linear
functionals.

By linearity, it is enough to impose the cyclicity constraint on monomials:
\begin{equation}\label{eq:cyclic_constraints}
  L(vw)=L(wv)
  \qquad
  \forall v,w\in\ncmonoid.
\end{equation}
One can then build SDP relaxations for TPCPO using the hybrid approach of subsection~\ref{sec:hybrid_approaches}, where the partial commutation relations are handled by working in the quotient algebra $\kk[\ncmonoid_\graph]$ and where both the additional constraints $R$ and the cyclicity constraints \eqref{eq:cyclic_constraints} are imposed as linear constraints on the moment variables.

Note that in this hybrid approach, it is sufficient to impose the cyclicity constraints only on the pc-monomials, i.e., to replace the constraints \eqref{eq:cyclic_constraints} by
\begin{equation}
  L([v][w])=L([w][v])
  \qquad
  \forall [v],[w]\in\ncmonoid_\graph.
\end{equation}
Indeed, when quotienting by the partial commutation relations, one projects the linear functional $L$ to a linear functional on the quotient algebra $\kk[\ncmonoid_\graph]$, i.e, $L(w)=L([w])$ for all $w\in\ncmonoid$. Hence, $L(vw) = L([vw]) = L([v][w])$ since the projection is a monoid morphism. Similarly, $L(wv) = L([wv]) = L([w][v])$. The constraint $L(vw)=L(wv)$ is therefore equivalent to $L([v][w])=L([w][v])$.

One can go further and quotient not only by the partial commutation relations, but also by the cyclicity constraints.
As before, let $I_{R_\graph}=\idealgen{R_\graph}$ be the ideal generated by the partial commutation relations, which defines a vector subspace of $\ncalgebra$. Let $I_{\mathrm{cyc}}$ be the vector subspace generated by the cyclicity constraints:
\begin{equation}
  I_{\mathrm{cyc}} = \operatorname{span}_{\kk}\{\,vw-wv:\ v,w\in\ncmonoid\,\}.
\end{equation}
Then any feasible linear functional $L$ vanishes on the vector space $I_{R_\graph} + I_{\mathrm{cyc}}$ generated by the union of the two sets of constraints. It is therefore sufficient to consider linear functionals
\begin{equation}
  L:\ncalgebra/(I_{R_\graph}+I_{\mathrm{cyc}})\to\kk.
\end{equation}
restricted to the quotient space $\ncalgebra/(I_{R_\graph}+I_{\mathrm{cyc}})$. Denote $\llbracket p \rrbracket$ the elements of this quotient space, corresponding to equivalent classes of polynomials modulo the partial commutation relations and the cyclicity constraints. Since these constraints are binomial, these equivalent classes are linear combinations $\llbracket p \rrbracket = \sum_{\llbracket w \rrbracket} p_{\llbracket w \rrbracket} \llbracket w \rrbracket$ of equivalent classes $\llbracket w \rrbracket$ of monomials.
That is, 
\begin{equation}
  \ncalgebra/(I_{R_\graph}+I_{\mathrm{cyc}})
  \simeq
  \kk\left[\ncmonoid_\graph^{\mathrm{cyc}}\right]\,
\end{equation}
where $\ncmonoid_\graph^{\mathrm{cyc}}$ is the set of \emph{cyclic pc-monomials}, defined as follows.

To define equivalent classes of monomials modulo the partial commutation relations and the cyclicity constraints, we can proceed in two steps: first quotient by the partial commutation relations, leading to the pc-monomials $[w]\in \ncmonoid_\graph$, and then quotient by the cyclicity constraints. Indeed the quotient map $\pi_\graph:\ncmonoid\to\ncmonoid_\graph$ is a monoid morphism, so a cyclic move $vw\leftrightarrow wv$ in $\ncmonoid$ descends to the cyclic move $[v] [w] \leftrightarrow [w] [v]$ in $\ncmonoid_\graph$. Conversely, any cyclic move in $\ncmonoid_\graph$ can
be lifted to such a move after choosing representatives in $\ncmonoid$.

Define therefore in $\ncmonoid_\graph$, the equivalence relation $\sim_{\mathrm{cyc}}$ that is generated by the cyclicity constraints. That is, two pc-monomials $[v],[w]\in\ncmonoid_\graph$ are cyclically equivalent $[v]\sim_{\mathrm{cyc}}[w]$ if one can be obtained from the other by a finite sequence of cyclic moves: $[v]=[u_1][u_2] \leftrightarrow [u_2][u_1] = [u_3][u_4] \leftrightarrow [u_4][u_3] = \cdots = [u_{n-1}][u_n] \leftrightarrow [u_n][u_{n-1}] = [w]$ for some $[u_1],\ldots,[u_n]\in\ncmonoid_\graph$.

For $[w]\in\ncmonoid_\graph$, we denote its cyclic class by
\begin{equation}
  \llbracket w \rrbracket := \{\, [v]\in\ncmonoid_\graph:\ [v]\sim_{\mathrm{cyc}}[w]\,\}
\end{equation}
and let
\begin{equation}
  \ncmonoid_\graph^{\mathrm{cyc}} := \ncmonoid_\graph/\!\sim_{\mathrm{cyc}}
\end{equation}
be the set of cyclic pc-monomials $\llbracket w \rrbracket$.
\begin{example} Let $X=\{a,b,c\}$ and assume only $a$ and $c$ commute. Then in $\ncmonoid_\graph$ we have
  \begin{equation}
  [abc]=[ab][c]
  \sim_{\mathrm{cyc}}
  [c][ab]=[cab]=[acb],
\end{equation}
and then
\begin{equation}
  [acb]=[a][cb]
  \sim_{\mathrm{cyc}}
  [cb][a]=[cba].
\end{equation}
Hence $\llbracket abc \rrbracket=\llbracket cba \rrbracket$.

Note that as words in $\ncmonoid$, $abc$ and $cba$ are not cyclically equivalent. However, they are indeed equivalent when taking into account both the partial commutation relations and the cyclicity constraints:  $abc \sim_{\mathrm{cyc}} cab \sim_{\mathrm{pc}} acb \sim_{\mathrm{cyc}} cba$.
\end{example}

After imposing both the partial commutation relations and the cyclicity
constraint, a feasible tracial functional may thus be viewed as a linear functional
\begin{equation}
  L:
  \kk[\ncmonoid_\graph^{\mathrm{cyc}}]
  \longrightarrow
  \kk .
\end{equation}
Such a linear functional is fully specified by the moment variables
\begin{equation}
  y_{\llbracket w \rrbracket} = L(\llbracket w \rrbracket)
\end{equation}
indexed by the cyclic pc-monomials $\llbracket w \rrbracket\in\ncmonoid_\graph^{\mathrm{cyc}}$.
As before, we can thus build SDP relaxations for TPCPO by optimizing over such moment variables, subject to the positivity and linear constraints inherited from the original problem, while the partial commutation and cyclicity constraints are automatically taken into account by the restriction to the quotient space $\kk[\ncmonoid_\graph^{\mathrm{cyc}}]$.

However, there is an important caveat. 
The identification
\begin{equation}
  \ncalgebra/(I_{R_\graph}+I_{\mathrm{cyc}})
  \simeq
  \kk\left[\ncmonoid_\graph^{\mathrm{cyc}}\right]
\end{equation}
is an identification of vector spaces. Here
$\kk\left[\ncmonoid_\graph^{\mathrm{cyc}}\right]$ denotes the vector space
freely spanned by the set $\ncmonoid_\graph^{\mathrm{cyc}}$ of cyclic
pc-monomials. It should not be interpreted as a monoid algebra, because
$\ncmonoid_\graph^{\mathrm{cyc}}$ is not a monoid in general.

Indeed, the quotient $\ncmonoid_\graph$ by the partial commutation relations is a monoid quotient:
the product of two pc-monomials $[v],[w]\in\ncmonoid_\graph$ is well-defined as
the pc-monomial $[v][w]=[vw]$. By contrast, the cyclic quotient $\ncmonoid_\graph^{\mathrm{cyc}}$ is not compatible
with multiplication. In general, from
\begin{equation}
  [v]\sim_{\mathrm{cyc}}[w]
\end{equation}
one cannot conclude that
\begin{equation}
  [s][v][t]\sim_{\mathrm{cyc}}[s][w][t]
\end{equation}
for arbitrary pc-monomials $[s],[t]\in\ncmonoid_\graph$. Thus products of
cyclic classes are not well-defined.

This point is important for the SDP construction since the moment and localizing matrices are defined in terms of products of monomials. 
For instance, in the original formulation before taking any quotient, the moment matrix has entries $M(\mathsf{L})(v,w) = \mathsf{L}(v^*w)$ indexed by free monomials in $\ncmonoid$. 
After quotienting by the partial commutation relations, these entries become $\mathsf{L}(v^*w) \mapsto \mathrm{L}([v^*w])$, where $[v^*w]$ is the pc-monomial obtained by first computing the product $v^*w$ in the free monoid $\ncmonoid$ and then projecting to the quotient $\ncmonoid_\graph$. Since this projection is a $\inv$-monoid morphism, we have $[v^*w]=[v]^*[w]$. Thus the moment matrix entries can be written as $M(\mathrm{L})([v],[w]) = \mathrm{L}([v]^*[w])$, where the product is computed in the PC monoid $\ncmonoid_\graph$. This gives a reduced moment matrix $M(\mathrm{L})$ indexed by pc-monomials in $\ncmonoid_\graph$.

In the tracial case, one performs one further projection $\mathrm{L}([v]^*[w])\mapsto L(\llbracket [v]^*[w] \rrbracket)$. However, the product $[v]^*[w]$ in $\ncmonoid_\graph$ cannot be replaced by a product of cyclic classes. 

Thus the row and column indices of the moment and localizing matrices are still indexed by ordinary pc-monomials $[u],[v]\in\ncmonoid_\graph$, while their entries are defined as
\begin{equation}
  M(q,L)_{[v],[w]}
  =
  L\left(\llbracket [v]^\inv[w] \rrbracket\right).
\end{equation}
The equality constraints in $R$ are imposed in the same way as
\begin{equation}
    L\left(\llbracket [s][r][t] \rrbracket\right) = 0
\end{equation}
for $[s],[t]\in\ncmonoid_\graph$.

In summary, the SDP relaxations of TPCPO are defined exactly in the same way as in the non-tracial case, with the same indexing for the moment, localizing matrices and linear constraints in term of pc-monomials in $\ncmonoid_\graph$. Except that once products in $\ncmonoid_\graph$ are computed, one further applies the cyclic equivalence relation to the resulting pc-monomial before associating it with a corresponding moment variable $L(\llbracket w \rrbracket) = y_{\llbracket w \rrbracket}$.

Finally, note that the involution does descend to the cyclic quotient. Indeed, define
\begin{equation}
  \llbracket w \rrbracket^\inv := \llbracket [w]^\inv \rrbracket.
\end{equation}
Then 
\begin{equation}
\llbracket v \rrbracket = \llbracket w \rrbracket \iff  \llbracket v \rrbracket^\inv = \llbracket w \rrbracket^\inv.
\end{equation}
Indeed, if $[v]$ and $[w]$ are related by a sequence of cyclic moves of the form $[u_i][u_{i+1}]\sim_{\mathrm{cyc}}[u_{i+1}][u_i]$, then $[v]^\inv$ and $[w]^\inv$ are related by the sequence of cyclic moves
$\left([u_i][u_{i+1}]\right)^\inv = [u_{i+1}]^\inv[u_i]^\inv \sim_{\mathrm{cyc}} [u_i]^\inv[u_{i+1}]^\inv = \left([u_{i+1}][u_i]\right)^\inv$. Hence cyclic equivalence is stable under involution.

As in the non-tracial case, one then needs to store only one moment variable $y_{\llbracket w \rrbracket}$ per pair of cyclic classes $\{\llbracket w \rrbracket, \llbracket w \rrbracket^\inv\}$, since by the Hermitian symmetry of the moments
\begin{equation}
  y_{{\llbracket w \rrbracket}^\inv} = y_{\llbracket [w]^\inv \rrbracket}
  =
  \overline{y_{\llbracket w \rrbracket}} .
\end{equation}

\subsection{Cyclic equivalence of pc-monomials}
To implement the TPCPO SDP relaxation, one needs
to project a pc-monomial $[w]\in\ncmonoid_\graph$ to its cyclic class
$\llbracket w \rrbracket$, or at least to decide whether two pc-monomials
$[v],[w]\in\ncmonoid_\graph$ are cyclically equivalent. The latter is sufficient,
since one can then build the cyclic pc-monomials incrementally during the SDP
construction by merging cyclically equivalent ones. We now discuss how to do this.

\subsubsection{Cyclic equivalence of circuits}
As in the previous section, consider a circuit representation $C=(K_1,\ldots,K_k)$ of a monomial $w$. To define a cyclic equivalence of circuits, we imagine closing the circuit by connecting the output of the circuit back to its input. As before, gates are allowed to slide along wires, as long as they do not cross any gate with a common wire. However,  now they can also slide across the boundaries of the circuit. This corresponds to the cyclic move $xv \leftrightarrow vx$ where $x\in\xx$ is a gate at the front or back of the circuit and $v\in\ncmonoid$ is the rest of the circuit.

Two circuits are said to be cyclically equivalent if one can be obtained from the other by a finite sequence of such slide-gate moves, inside the circuit (as in the non-cyclic case) or across the boundaries of the circuit.

\begin{example}[name=Running example,continues=example:abcd_circuit]
  Consider the following circuits associated to the monomials $cbabdb$ and $ab^2dbc$, where the output to input wires are dashed. 
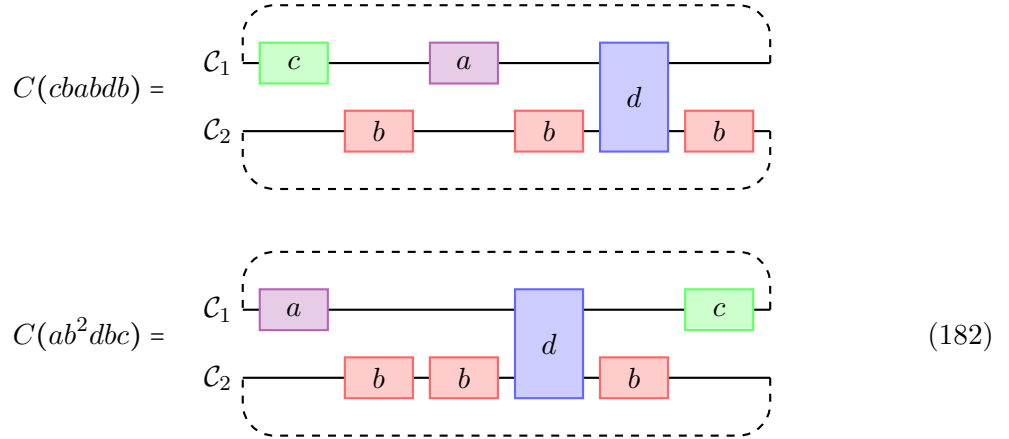
\begin{align}\label{eq:cyclic_circuit_example}
  C(cbabdb) = \quad&{\begin{tikzpicture}[baseline=(current bounding box.center),scale=0.9]
    % Wires/Subsystemss
  \draw[thick] (0,2) node[left] {$\clique_1$} -- (7.75,2);
  \draw[thick, dashed, rounded corners=12pt] (7.75,2) -- (7.75,2.85) -- (0,2.85) -- (0,2);
  \draw[thick] (0,1) node[left] {$\clique_2$} -- (7.75,1);
  \draw[thick, dashed, rounded corners=12pt] (7.75,1) -- (7.75,0.15) -- (0,0.15) -- (0,1);
  \filldraw[fill=green!20, draw=green!60, thick] (0.25, 1.7) rectangle (1.25, 2.3) node[midway] {$c$};
   \filldraw[fill=red!20, draw=red!60, thick] (1.5, 0.7) rectangle (2.5, 1.3) node[midway] {$b$};
  \filldraw[fill=violet!20, draw=violet!60, thick] (2.75, 1.7) rectangle (3.75, 2.3) node[midway] {$a$};
  \filldraw[fill=red!20, draw=red!60, thick] (4, 0.7) rectangle (5, 1.3) node[midway] {$b$};
  \filldraw[fill=blue!20, draw=blue!60, thick] (5.25, 0.7) rectangle (6.25, 2.3) node[midway] {$d$};
  \filldraw[fill=red!20, draw=red!60, thick] (6.5, 0.7) rectangle (7.5, 1.3) node[midway] {$b$};
  \end{tikzpicture}}\nonumber \\[4ex]
  C(ab^2dbc)= \quad &
  {\begin{tikzpicture}[baseline=(current bounding box.center),scale=0.9]
    % Wires/Subsystems
  \draw[thick] (0,2) node[left] {$\clique_1$} -- (7.75,2);
  \draw[thick, dashed, rounded corners=12pt] (7.75,2) -- (7.75,2.85) -- (0,2.85) -- (0,2);
  \draw[thick] (0,1) node[left] {$\clique_2$} -- (7.75,1);
  \draw[thick, dashed, rounded corners=12pt] (7.75,1) -- (7.75,0.15) -- (0,0.15) -- (0,1);
  \filldraw[fill=violet!20, draw=violet!60, thick] (0.25, 1.7) rectangle (1.25, 2.3) node[midway] {$a$};
   \filldraw[fill=red!20, draw=red!60, thick] (1.5, 0.7) rectangle (2.5, 1.3) node[midway] {$b$};
  \filldraw[fill=red!20, draw=red!60, thick] (2.75, 0.7) rectangle (3.75, 1.3) node[midway] {$b$};
   \filldraw[fill=blue!20, draw=blue!60, thick] (4, 0.7) rectangle (5, 2.3) node[midway] {$d$};
  \filldraw[fill=red!20, draw=red!60, thick] (5.25, 0.7) rectangle (6.25, 1.3) node[midway] {$b$};
  \filldraw[fill=green!20, draw=green!60, thick] (6.5, 1.7) rectangle (7.5, 2.3) node[midway] {$c$};
  \end{tikzpicture}}
\end{align}
These two circuits are cyclically equivalent, since one can obtain the second one by sliding the gates of the first one across the boundary. This implies that $\llbracket cbabdb \rrbracket = \llbracket ab^2dbc \rrbracket$
\end{example}

If two circuits are cyclically equivalent, then their wire words are cyclically
equivalent. More precisely, let $W(C)=(w_1,\ldots,w_m)$, and $W(C')=(w'_1,\ldots,w'_m)$ be the wire words of two cyclically equivalent circuits. Then, for each wire
\(i\), the word \(w'_i\) is a cyclic rotation of \(w_i\). Indeed, writing $w_i=y_{i,1}\cdots y_{i,k_i}$, there exists an integer \(t_i\) such that $w'_i  =  y_{i,1+t_i}\cdots y_{i,k_i+t_i}$  with indices taken modulo \(k_i\).

The converse is false. Wirewise cyclic equivalence allows the different wire
words to be rotated independently, whereas an actual cyclic circuit admits only
rotations that are compatible with the common gate occurrences shared by
several wires.

\begin{example} Let $X=\{a,b,c\}$ and assume only $a$ and $c$ commute, i.e., the wires in a circuit representation are the cliques $\clique_1=\{a,b\}$ and $\clique_2=\{b,c\}$. Consider the two circuits for the words $bacb$ and $babc$:
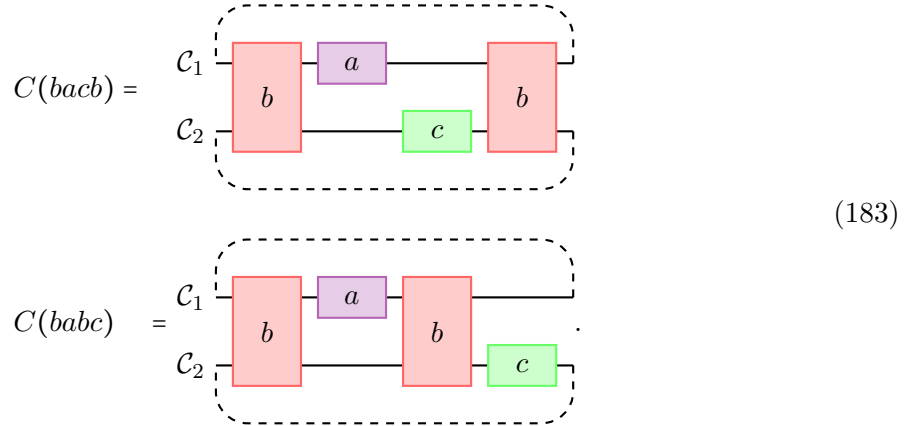
\begin{equation}
\begin{aligned}
C(bacb) =\quad
&
{\begin{tikzpicture}[baseline=(current bounding box.center),scale=0.9]
  \draw[thick] (0,2) node[left] {$\clique_1$} -- (5.25,2);
  \draw[thick, dashed, rounded corners=12pt] (5.25,2) -- (5.25,2.85) -- (0,2.85) -- (0,2);
  \draw[thick] (0,1) node[left] {$\clique_2$} -- (5.25,1);
  \draw[thick, dashed, rounded corners=12pt] (5.25,1) -- (5.25,0.15) -- (0,0.15) -- (0,1);
  \filldraw[fill=red!20, draw=red!60, thick] (0.25,0.7) rectangle (1.25,2.3) node[midway] {$b$};
  \filldraw[fill=violet!20, draw=violet!60, thick] (1.5,1.7) rectangle (2.5,2.3) node[midway] {$a$};
  \filldraw[fill=green!20, draw=green!60, thick] (2.75,0.7) rectangle (3.75,1.3) node[midway] {$c$};
  \filldraw[fill=red!20, draw=red!60, thick] (4,0.7) rectangle (5,2.3) node[midway] {$b$};
\end{tikzpicture}}
\\[3ex]
 C(babc)\quad =
&
{\begin{tikzpicture}[baseline=(current bounding box.center),scale=0.9]
  \draw[thick] (0,2) node[left] {$\clique_1$} -- (5.25,2);
  \draw[thick, dashed, rounded corners=12pt] (5.25,2) -- (5.25,2.85) -- (0,2.85) -- (0,2);
  \draw[thick] (0,1) node[left] {$\clique_2$} -- (5.25,1);
  \draw[thick, dashed, rounded corners=12pt] (5.25,1) -- (5.25,0.15) -- (0,0.15) -- (0,1);
  \filldraw[fill=red!20, draw=red!60, thick] (0.25,0.7) rectangle (1.25,2.3) node[midway] {$b$};
  \filldraw[fill=violet!20, draw=violet!60, thick] (1.5,1.7) rectangle (2.5,2.3) node[midway] {$a$};
  \filldraw[fill=red!20, draw=red!60, thick] (2.75,0.7) rectangle (3.75,2.3) node[midway] {$b$};
  \filldraw[fill=green!20, draw=green!60, thick] (4,0.7) rectangle (5,1.3) node[midway] {$c$};
\end{tikzpicture}}.
\end{aligned}
\end{equation}
Then the wire representation of $bacb$ is $W=(bab,bcb)$ and the one of $babc$ is $W'=(bab,bbc)$. So $w_1=w_1'$ and $w_2$ and $w_2'$ are equivalent by rotation of their letters. However, the two circuits $C(bacb)$ and $ C(babc)$ are not cyclically equivalent, since there is no way to slide the gates in one circuit to obtain the other. Therefore $\llbracket bacb \rrbracket \neq \llbracket babc \rrbracket$.
\end{example}

\subsubsection{Cyclic occurrence graphs}

We have seen in subsection~\ref{sec:other_representations} that the occurrence graph $\Gamma$ and the dependency graph $\bar \Gamma$ of a circuit representation $C$ gives a canonical representative of the pc-monomial $[w]$ associated to $C$. In particular, two circuits are equivalent if and only if they have the same occurrence graph and the same dependency graph.
This is so because these graphs encode all the information about the relative order of the gate occurrences in the circuit, up to partial commutations. It therefore eliminates the arbitrariness in the circuit representation coming from the freedom of sliding gates along wires. 

Two circuits $C$ and $C'$ that are cyclically equivalent do not have, however, the same occurrence graphs $\Gamma$ and $\Gamma'$, nor the same dependency graphs $\bar{\Gamma}$ and $\bar{\Gamma}'$, since the relative order of the gates can change when sliding gates across the boundaries of the circuit.

Their dependency graphs $\bar{\Gamma}$ and $\bar{\Gamma}'$, though, are related in a simple way. We first remind that the dependency graph $\bar{\Gamma}$ of a circuit $C$ is obtained by viewing each gate occurrence $(x,k)$ has a vertex and adding a directed edge from $(x,k)$ to $(y,\ell)$ whenever the two gates share a wire $i$ and $(x,k)$ occurs before $(y,\ell)$, regardless of whether they are immediate successors along that wire. If a gate occurrence $(x,k)$ is at the front of the circuit $C$, then it is a source vertex in the dependency graph $\bar{\Gamma}$, with an edge $(x,k)\rightarrow (y,l)$  for every occurrence $(y,l)$ where $y$ does not commute with $x$.
If $(x,k)$ is instead at the back of the circuit $C$, then it is a sink vertex, with an edge $(y,l)\rightarrow (x,k)$ for every occurrence $(y,l)$ where $y$ does not commute with $x$. Thus moving a gate occurrence from the front to the back of the circuit changes it from a source to a sink in the dependency graph and simply amounts to flipping the direction of all edges leaving that vertex in the dependency graph. Conversely, moving a gate occurrence from the back to the front of the circuit changes it from a sink to a source in the dependency graph and simply amounts to flipping the direction of all edges incident to that vertex in the dependency graph. 
Two circuits that are cyclically equivalent have thus full dependency graphs that are related by a sequence of such source-to-sink flips.

\begin{example} Consider the two cyclically equivalent circuits $C(cbabdb)$ and $C(ab^2dbc)$ in \eqref{eq:cyclic_circuit_example}. Their occurrence graphs are
\begin{align} 
  \Gamma(cbabdb) =\quad &
    {\begin{tikzpicture}[baseline=(current bounding box.center)]
      \node (c) at (0,2) {};
      \node (b1) at (0,0) {};
      \node (a) at (2,2) {};
      \node (b2) at (2,0) {};
      \node (d) at (3,1) {};
      \node (b3) at (4.7,1) {};
      \fill (c) circle (1.2pt);
      \fill (b1) circle (1.2pt);
      \fill (a) circle (1.2pt);
      \fill (b2) circle (1.2pt);
      \fill (d) circle (1.2pt);
      \fill (b3) circle (1.2pt);
      \node[above] at (c) {$(c,1)$};
      \node[below] at (b1) {$(b,1)$};
      \node[above] at (a) {$(a,1)$};
      \node[below] at (b2) {$(b,2)$};
      \node[above right] at (d) {$(d,1)$};
      \node[right] at (b3) {$(b,3)$};
      \draw[-{Stealth[length=2mm]}] (c) -- (a);
      \draw[-{Stealth[length=2mm]}] (b1) -- (b2);
      \draw[-{Stealth[length=2mm]}] (a) -- (d);
      \draw[-{Stealth[length=2mm]}] (b2) -- (d);
      \draw[-{Stealth[length=2mm]}] (d) -- (b3);
    \end{tikzpicture}}\nonumber \\
    & \nonumber \\
  \Gamma(ab^2dbc) =\quad &
    {\begin{tikzpicture}[baseline=(current bounding box.center)]
      \node (c) at (0,2) {};
      \node (b1) at (0,0) {};
      \node (a) at (2,2) {};
      \node (b2) at (2,0) {};
      \node (d) at (3,1) {};
      \node (b3) at (4.7,1) {};
      \fill (c) circle (1.2pt);
      \fill (b1) circle (1.2pt);
      \fill (a) circle (1.2pt);
      \fill (b2) circle (1.2pt);
      \fill (d) circle (1.2pt);
      \fill (b3) circle (1.2pt);
      \node[above] at (c) {$(c,1)$};
      \node[below] at (b1) {$(b,1)$};
      \node[above] at (a) {$(a,1)$};
      \node[below] at (b2) {$(b,2)$};
      \node[above right] at (d) {$(d,1)$};
      \node[right] at (b3) {$(b,3)$};
      \draw[-{Stealth[length=2mm]}] (d) -- (c);
      \draw[-{Stealth[length=2mm]}] (b1) -- (b2);
      \draw[-{Stealth[length=2mm]}] (a) -- (d);
      \draw[-{Stealth[length=2mm]}] (b2) -- (d);
      \draw[-{Stealth[length=2mm]}] (d) -- (b3);
    \end{tikzpicture}}
  \end{align}
Their dependency graphs are 
\begin{align}
  \bar\Gamma(cbabdb) =\quad &{
  \begin{tikzpicture}[baseline=(current bounding box.center)]
    \node (c) at (0,2) {};
    \node (b1) at (0,0) {};
    \node (a) at (2,2) {};
    \node (b2) at (2,0) {};
    \node (d) at (3,1) {};
    \node (b3) at (4.7  ,1) {};
    \fill (c) circle (1.2pt);
    \fill (b1) circle (1.2pt);
    \fill (a) circle (1.2pt);
    \fill (b2) circle (1.2pt);
    \fill (d) circle (1.2pt);
    \fill (b3) circle (1.2pt);
    \node[above] at (c) {$(c,1)$};
    \node[below] at (b1) {$(b,1)$};
    \node[above] at (a) {$(a,1)$};
    \node[below] at (b2) {$(b,2)$};
    \node[above right] at (d) {$(d,1)$};
    \node[right] at (b3) {$(b,3)$};
    \draw[-{Stealth[length=2mm]}] (c) -- (a);
    \draw[-{Stealth[length=2mm]}] (b1) -- (b2);
    \draw[-{Stealth[length=2mm]}] (a) -- (d);
    \draw[-{Stealth[length=2mm]}] (b2) -- (d);
    \draw[-{Stealth[length=2mm]}] (d) -- (b3);
    % Additional edges for the dependency graph
    \draw[-{Stealth[length=2mm]}] (c) -- (d);
    \draw[-{Stealth[length=2mm]}] (b1) -- (d);
    \draw[-{Stealth[length=2mm]}] (b1) -- (b3);
    \draw[-{Stealth[length=2mm]}] (b2) -- (b3);
  \end{tikzpicture}}\nonumber \\
  & \nonumber \\
  \bar\Gamma(ab^2dbc) =\quad &{
  \begin{tikzpicture}[baseline=(current bounding box.center)]
    \node (c) at (0,2) {};
    \node (b1) at (0,0) {};
    \node (a) at (2,2) {};
    \node (b2) at (2,0) {};       
    \node (d) at (3,1) {};
    \node (b3) at (4.7  ,1) {};
    \fill (c) circle (1.2pt);
    \fill (b1) circle (1.2pt);
    \fill (a) circle (1.2pt);
    \fill (b2) circle (1.2pt);
    \fill (d) circle (1.2pt);
    \fill (b3) circle (1.2pt);
    \node[above] at (c) {$(c,1)$};  
    \node[below] at (b1) {$(b,1)$};
    \node[above] at (a) {$(a,1)$};
    \node[below] at (b2) {$(b,2)$};
    \node[above right] at (d) {$(d,1)$};
    \node[right] at (b3) {$(b,3)$};
    \draw[-{Stealth[length=2mm]}] (d) -- (c);
    \draw[-{Stealth[length=2mm]}] (b1) -- (b2);
    \draw[-{Stealth[length=2mm]}] (a) -- (d);
    \draw[-{Stealth[length=2mm]}] (b2) -- (d);
    \draw[-{Stealth[length=2mm]}] (d) -- (b3);
    % Additional edges for the dependency graph
    \draw[-{Stealth[length=2mm]},] (b1) -- (d);
    \draw[-{Stealth[length=2mm]},] (b1) -- (b3);
    \draw[-{Stealth[length=2mm]},] (b2) -- (b3);
    \draw[-{Stealth[length=2mm]},] (a) -- (c);
  \end{tikzpicture}}
\end{align}
They just differ by a source-to-sink flip at the vertex $(c,1)$, which is a source in $\bar\Gamma(cbabdb)$ and a sink in $\bar\Gamma(ab^2dbc)$.
\end{example}

These source-to-sink flips, and the equivalence relation they generate on
acyclic orientations of a graph, were studied in\cite{Pretzel}. Their
relation with cyclic versions of partial orders was developed in \cite{develin} under the name of
\emph{toric posets}. Related ideas for partially commutative monoids and
Coxeter groups appear in the work \cite{chao} on toric
heaps, which may be viewed as cyclic analogues of Viennot heaps.

We now give an alternative, but closely related, description of cyclic equivalence at the level of occurrence graphs, where instead of keeping all dependencies, we keep only immediate successor relations on the closed wires. Given a circuit $C$, we introduce its \emph{cyclic occurrence graph} $\mathring \Gamma$ as the graph obtained from the usual occurrence graph $\Gamma$ by adding boundary edges that connect the last occurrence of a gate on a wire to the first occurrence of a gate on that same wire.

Specifically, we define the \emph{cyclic occurrence graph} $\mathring \Gamma$ associated to a circuit $C$ as follows. 
\begin{enumerate}
  \item The vertices of $\mathring \Gamma$ are the gate occurrences $\Omega = \{(x,k): x\in\xx, 1\le k \le |x|_C\}$, where $|x|_C$ is the number of occurrences of $x$ in the circuit $C$.
  \item There is a directed, labelled edge 
  \begin{equation}
    (x,k) \xrightarrow{0} (y,l) 
  \end{equation}
  if the $k$-th occurrence of $x$ is immediately followed by the $l$-th occurrence of $y$ on some wire inside the circuit.
  \item There is a directed, labelled edge 
  \begin{equation}
    (x,k) \xrightarrow{1} (y,l) 
  \end{equation}
  if $(x,k)$ is the last occurrence of $x$ on some wire and $(y,l)$ is the first occurrence of $y$ on that same wire.
\end{enumerate}
In the following, we call the labels $\varepsilon\in\{0,1\}$ of the directed
edges the \emph{gains}. A gain $\varepsilon=0$ means that the edge is a regular
edge, while a gain $\varepsilon=1$ means that the edge is
a wrap edge crossing the cyclic boundary.  The subgraph of $\mathring\Gamma$
consisting of the edges of gain $0$ is the ordinary occurrence graph $\Gamma$ of
the circuit $C$; the additional edges of gain $1$ record the wrap edges.

Ordinary gate-slides inside the circuit $C$ do not change the occurrence graph $\Gamma$, hence they do not change
$\mathring\Gamma$. 
The only elementary moves that can change $\mathring\Gamma$ are the moves through the boundary. The following proposition characterizes how two circuit that are cyclically equivalent are related in terms of their cyclic occurrence graphs.
\begin{proposition}
  Two circuits $C$ and $D$ are cyclically equivalent,
  and hence represent the same cyclic pc-monomial, if and only if their cyclic
  occurrence graphs $\mathring\Gamma_C$ and $\mathring\Gamma_D$ are related as
  follows. There exists a bijection
  \begin{equation}\label{eq:phi}
    \phi:\Omega_C\to\Omega_D
  \end{equation}
  preserving gate labels, $\phi(x,k)=(x,\phi_x(k))$, and a function
  \begin{equation}\label{eq:h}
    h:\Omega_C\to\mathbb Z
  \end{equation}
  such that, for every edge
  \begin{equation}
    p\xrightarrow{\varepsilon}q
  \end{equation}
  in $\mathring\Gamma_C$, there is an edge
  \begin{equation}
    \phi(p)\xrightarrow{\varepsilon+h(p)-h(q)}\phi(q)
  \end{equation}
  in $\mathring\Gamma_D$, and conversely.
\end{proposition}

\begin{proof}
We first prove the forward implication. By the assumptions, there is a sequence of gate slides that transforms $C$ into $D$. Ordinary gate-slides inside the
 circuit do not change the cyclic occurrence graph, so it is enough to
track what happens when a gate is slid through the cyclic boundary.

Suppose that an occurrence \(p=(x,1)\in\Omega_C\) is at the front of the
current circuit. In the cyclic occurrence graph this means that every
edge leaving \(p\) has gain \(0\), while every edge entering \(p\) has gain
\(1\). Sliding \(p\) through the boundary to the back changes only the gains of
the edges incident to \(p\): outgoing edges become wrap edges, and incoming
edges become regular edges. Thus, for every edge \(p\to q\),
\begin{equation}
  \varepsilon'(p,q)=\varepsilon(p,q)+1,
\end{equation}
whereas, for every edge \(q\to p\),
\begin{equation}
  \varepsilon'(q,p)=\varepsilon(q,p)-1.
\end{equation}
All other gains are unchanged. The new gains still belong to \(\{0,1\}\)
precisely because \(p\) was at the front of the circuit.

There is also a relabelling of repeated occurrences. If the occurrence moved
through the boundary is \(p=(x,1)\), then after the move it becomes the last
occurrence of \(x\). Thus the occurrences of \(x\) are cyclically renumbered,
\begin{equation}
  (x,k)\longmapsto (x,k-1),
\end{equation}
with indices taken modulo \(|x|_C\). Occurrences of all other letters keep their
indices.

Iterating such boundary slides to go from $C$ to $D$ produces a relabelling \(\phi\) as in
\eqref{eq:phi}, preserving gate labels, together with a function \(h\) as in
\eqref{eq:h}. Here \(h(p)\) records the net number of times the occurrence
\(p\) has been moved through the boundary from front to back. Hence an edge
\begin{equation}
  p\xrightarrow{\varepsilon}q
\end{equation}
is transformed into
\begin{equation}
  \phi(p)\xrightarrow{\varepsilon+h(p)-h(q)}\phi(q).
\end{equation}

Conversely, suppose that such \(\phi\) and \(h\) are given. Pulling
\(\mathring\Gamma_D\) back along \(\phi\), we identify \(\Omega_D\) with
\(\Omega_C\). After this identification, the two cyclic occurrence graphs have
the same directed edge set. They differ only in their gain assignments.

Let \(\varepsilon\) denote the gain assignment of \(\mathring\Gamma_C\). For any
function \(\eta:\Omega_C\to\mathbb Z\), define
\begin{equation}\label{eq:eta_u_def}
  \varepsilon_\eta(p,q)
  :=
  \varepsilon(p,q)+\eta(p)-\eta(q)
\end{equation}
on each directed edge \(p\to q\). Thus the pulled-back graph
\(\mathring\Gamma_D\) has gain assignment \(\varepsilon_h\).

Adding a constant to \(h\) on each connected component of the underlying
undirected occurrence graph does not change the gains \(\varepsilon_h\). Hence, after
doing this component by component, we may assume that
\begin{equation}
  \min h=0.
\end{equation}
In particular, \(h\ge 0\).

For \(\eta:\Omega_C\to\mathbb Z_{\ge 0}\), write
\begin{equation}
  \mathring\Gamma(\eta):=(\Omega_C,E,\varepsilon_\eta),
\end{equation}
where \(E\) is the common directed edge set. We call \(\eta\)
\emph{admissible} if \(\mathring\Gamma(\eta)\) is the cyclic occurrence graph of
some circuit \(C_\eta\).
With this terminology, \(0\) is admissible, with \(C_0=C\), and \(h\) is
admissible, with \(C_h\) equal to the pulled-back copy of \(D\), up to ordinary
gate slides. We shall prove that \(C_h\) is cyclically equivalent to \(C\).

We argue by descending induction on
\begin{equation}
  |\eta|:=\sum_{p\in\Omega_C}\eta(p).
\end{equation}
More precisely, we prove that if \(\eta\) is admissible, then either \(\eta=0\), or
there exists an admissible function \(\eta^-\) with
\begin{equation}
  |\eta^-|=|\eta|-1
\end{equation}
such that \(C_\eta\) is obtained from \(C_{\eta^-}\) by one boundary slide. Iterating
this reduction from \(\eta=h\) to \(\eta=0\) then proves that \(C_h\) is cyclically
equivalent to \(C_0=C\).

Let \(\eta\) be admissible with \(|\eta|>0\), and set
\begin{equation}
  P_\eta:=\{p\in\Omega_C:\eta(p)>0\}.
\end{equation}
Let \(\Gamma_\eta\) be the directed graph formed by the zero-gain edges of
\(\mathring\Gamma(\eta)\), i.e., $\Gamma_\eta$ is the ordinary occurrence graph of \(C_\eta\). Hence \(\Gamma_\eta\) is acyclic.

Choose \(p\in P_\eta\) maximal in \(P_\eta\) with respect to the partial order
generated by directed paths in \(\Gamma_\eta\). Thus there is no directed path in
\(\Gamma_\eta\) from \(p\) to any other vertex of \(P_\eta\).

We claim that \(p\) is at the back of the circuit \(C_\eta\). Equivalently, every
outgoing edge of \(p\) has \(\varepsilon_\eta\)-gain \(1\), and every incoming edge of
\(p\) has \(\varepsilon_\eta\)-gain \(0\).

First suppose that an outgoing edge \(p\to q\) had \(\varepsilon_\eta\)-gain \(0\). Then
\begin{equation}
  0=\varepsilon(p,q)+\eta(p)-\eta(q),
\end{equation}
hence
\begin{equation}
  \eta(q)=\eta(p)+\varepsilon(p,q)\ge \eta(p)>0.
\end{equation}
Thus \(q\in P_\eta\). But \(p\to q\) is an edge of \(Z_\eta\), contradicting the
maximality of \(p\). Hence every outgoing edge of \(p\) has \(\varepsilon_\eta\)-gain
\(1\).

Next suppose that an incoming edge \(r\to p\) had \(\varepsilon_\eta\)-gain \(1\). Then
\begin{equation}
  1=\varepsilon(r,p)+\eta(r)-\eta(p),
\end{equation}
hence
\begin{equation}
  \eta(r)=\eta(p)+1-\varepsilon(r,p)\ge \eta(p)>0.
\end{equation}
Thus \(r\in P_\eta\).

Choose a closed wire whose immediate-successor relation contributes the edge
\(r\to p\). Since \(\mathring\Gamma(\eta)\) is the cyclic occurrence graph of the
 circuit \(C_\eta\), each closed wire has exactly one wrap edge, that is,
one edge of gain \(1\). Therefore \(r\to p\) is the unique wrap edge on this
wire. All other successor edges on the same closed wire, going from \(p\) back
to \(r\), have gain \(0\). Hence there is a directed path in \(Z_\eta\) from \(p\)
to \(r\). Since \(r\in P_\eta\), this again contradicts the maximality of \(p\).

Therefore every incoming edge of \(p\) has \(\varepsilon_\eta\)-gain \(0\), and the claim
is proved: \(p\) is at the back of \(C_\eta\).

Define a new function $\eta^-$ as follows
\begin{equation}
  \eta^-(p)=\eta(p)-1,
  \qquad
  \eta^-(q)=\eta(q)\quad(q\neq p).
\end{equation}
Since \(\eta(p)>0\), we still have \(\eta^-\ge 0\), and
\begin{equation}
  |\eta^-|=|\eta|-1.
\end{equation}

We compare \(\varepsilon_{\eta^-}\) with \(\varepsilon_\eta\). Only the gains of edges incident to
\(p\) change. Since \(p\) is at the back for \(\varepsilon_\eta\), every outgoing edge
\(p\to q\) has \(\varepsilon_\eta(p,q)=1\), and every incoming edge \(r\to p\) has
\(\varepsilon_\eta(r,p)=0\). Hence
\begin{equation}
  \varepsilon_{\eta^-}(p,q)=\varepsilon_\eta(p,q)-1=0
\end{equation}
for every outgoing edge \(p\to q\), while
\begin{equation}
  \varepsilon_{\eta^-}(r,p)=\varepsilon_\eta(r,p)+1=1
\end{equation}
for every incoming edge \(r\to p\). All other gains are unchanged.

Thus \(\mathring\Gamma(\eta^-)\) is exactly the cyclic occurrence graph obtained
from \(\mathring\Gamma(\eta)\) by sliding the occurrence \(p\) through the cyclic
boundary from the back to the front. Hence \(\eta^-\) is admissible, and the
corresponding circuit \(C_{\eta^-}\) is obtained from \(C_\eta\) by this single
backward boundary slide. Equivalently, \(C_\eta\) is obtained from \(C_{\eta^-}\) by
sliding \(p\) from the front to the back.

Repeating this construction finitely many times, starting from \(\eta_0=h\), gives
a sequence
\begin{equation}
  \eta_0=h,\ \eta_1,\ldots,\eta_N=0
\end{equation}
such that, for each \(i=0,\ldots,N-1\),
\(C_{\eta_{i+1}}\) is obtained from \(C_{\eta_i}\)
by sliding some gate through the cyclic boundary from the back to the front. Thus \(C_{\eta_0}=C_h\) is cyclically equivalent to \(C_{\eta_N}=C\).

Finally, by construction,
\(\mathring\Gamma(h)\) is the pullback of \(\mathring\Gamma_D\) along \(\phi\).
Thus the pullback circuit \(C_h\) and circuit \(D\) have the same cyclic
occurrence graph. Hence they represent the same  circuit up to ordinary
gate slides. Consequently \(D\) is cyclically equivalent to \(C_h\), and since
\(C_h\) is cyclically equivalent to \(C\), we conclude that \(D\) is cyclically
equivalent to \(C\).
\end{proof}

\begin{example}[name=Running example,continues=example:abcd_circuit]
  The cyclic occurrence graphs of the circuits \(C(cbabdb)\) and \(C(ab^2dbc)\) are the following (where edges of gain $0$ are solid and edges of gain $1$ are dashed):
  \begin{align}
    \mathring\Gamma(cbabdb) =\quad &
    
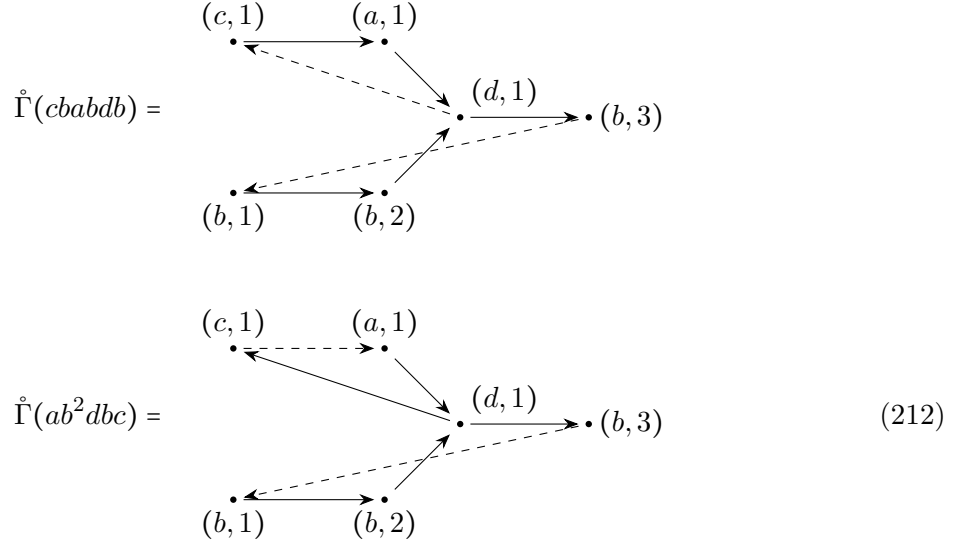
{\begin{tikzpicture}[baseline=(current bounding box.center)]
      \node (c) at (0,2) {};
      \node (b1) at (0,0) {};
      \node (a) at (2,2) {};
      \node (b2) at (2,0) {};
      \node (d) at (3,1) {};
      \node (b3) at (4.7,1) {};
      \fill (c) circle (1.2pt);
      \fill (b1) circle (1.2pt);
      \fill (a) circle (1.2pt);
      \fill (b2) circle (1.2pt);
      \fill (d) circle (1.2pt);
      \fill (b3) circle (1.2pt);
      \node[above] at (c) {$(c,1)$};
      \node[below] at (b1) {$(b,1)$};
      \node[above] at (a) {$(a,1)$};
      \node[below] at (b2) {$(b,2)$};
      \node[above right] at (d) {$(d,1)$};
      \node[right] at (b3) {$(b,3)$};
      \draw[-{Stealth[length=2mm]}] (c) -- (a);
      \draw[-{Stealth[length=2mm]}] (b1) -- (b2);
      \draw[-{Stealth[length=2mm]}] (a) -- (d);
      \draw[-{Stealth[length=2mm]}] (b2)  -- (d);
      \draw[-{Stealth[length=2mm]}] (d) -- (b3);
      \draw[-{Stealth[length=2mm]},dashed] (d) -- (c);
      \draw[-{Stealth[length=2mm]},dashed] (b3) -- (b1);
    \end{tikzpicture}}\nonumber \\
    & \nonumber \\
    \mathring\Gamma(ab^2dbc) =\quad &
    {\begin{tikzpicture}[baseline=(current bounding box.center)]
      \node (c) at (0,2) {};
      \node (b1) at (0,0) {};
      \node (a) at (2,2) {};  
    \node (b2) at (2,0) {};
      \node (d) at (3,1) {};
      \node (b3) at (4.7  ,1) {};
      \fill (c) circle (1.2pt);
      \fill (b1) circle (1.2pt);
      \fill (a) circle (1.2pt);
      \fill (b2) circle (1.2pt);
      \fill (d) circle (1.2pt);
      \fill (b3) circle (1.2pt);
      \node[above] at (c) {$(c,1)$};
      \node[below] at (b1) {$(b,1)$};
      \node[above] at (a) {$(a,1)$};
      \node[below] at (b2) {$(b,2)$};
      \node[above right] at (d) {$(d,1)$};
      \node[right] at (b3) {$(b,3)$};
      \draw[-{Stealth[length=2mm]}] (d) -- (c);
      \draw[-{Stealth[length=2mm]}] (b1) -- (b2);
      \draw[-{Stealth[length=2mm]}] (a) -- (d);
      \draw[-{Stealth[length=2mm]}] (b2) -- (d);
      \draw[-{Stealth[length=2mm]}] (d) -- (b3);
      \draw[-{Stealth[length=2mm]},dashed] (b3) -- (b1);
      \draw[-{Stealth[length=2mm]},dashed] (c) -- (a);
    \end{tikzpicture}}
  \end{align}
  The bijection \(\phi\) is the identity, and the height function \(h\) is given by
  \begin{equation}
    h(c,1)=1,\quad h(b,1)=0,\quad h(a,1)=0,\quad h(b,2)=0,\quad h(d,1)=0,\quad h(b,3)=0.
  \end{equation}
\end{example}

\subsection{Algorithmic implications.}
The preceding proposition gives a finite graph-theoretic criterion for cyclic
equivalence. In particular, after identifying a label-preserving bijection $\phi$ between the two cyclic occurrence graphs, one may test whether the gain assignments differ by a height transformation $h$. This is just a potential-existence problem on a directed graph, which is linear in the number of edges.
Alternatively, at the level of dependency graphs, cyclic equivalence can be
viewed as membership in the same source-to-sink flip class of acyclic
orientations \cite{Pretzel, develin}.

For implementations of the PCPO or TPCPO hierarchy, however, it is preferable to use the wire representation of pc-monomial as tuple of wire words. A direct algorithmic treatment is provided by
Liu--Wrathall--Zeger~\cite{liu1990efficient} who give a linear-time
procedure for deciding whether two wire representations are cyclically equivalent\footnote{Note that there is a small difference in terminology. Liu--Wrathall--Zeger say that two words are cyclically equal if they are related by a single elementary cyclic move. Here, we use the term \emph{cyclic equivalence} for the equivalence relation generated by such moves. Thus our cyclic equivalence is the transitive closure of their cyclic equality relation. This is referred to as conjugacy in \cite{duboc1986some,liu1990efficient}.}. Namely, two pc-monomials $[u]$ and $[v]$ in the alphabet $X$ are cyclically equivalent if $(i)$ each letter occurs the same number of times in $[u]$ and $[v]$ and $(ii)$ there are pc-monomials $[l]$ and $[r]$ such that $[l][u][r] = [v]^n$ for some $n = 1, \ldots, |X|$. Both of these conditions can be checked with the wire representations $\Pi(u)$ and $\Pi(v)$.

Consequently, cyclic pc-monomials can be built incrementally during the SDP
construction. As pc-monomials are generated, one stores one representative for
each cyclic class and compares each new pc-monomial with the existing
representatives using the Liu--Wrathall--Zeger test. If the test succeeds, the
new monomial is merged with the corresponding existing class; otherwise, it
starts a new cyclic class.

\section{Discussion}\label{sec:discussion}
In this paper, we detailed the quotient-based approach for NCPO, where polynomial equalities are handled directly at the algebraic level. Quotient reductions are already used in practice, especially in simple settings such as Bell scenarios. We pointed out that the ad-hoc substitution rules used in these cases, however, are not necessarily complete in more complex instances, even in simple scenarios where partial commutations are the only equality relations.

We then introduced PCPO, and its tracial counterpart TPCPO, which interpolate between the fully non-commutative and fully commutative regimes. We showed how existing representations of the underlying partially-commutative monoid can be used to construct quotient-based SDP relaxations for these problems. We introduced a circuit representation of pc-monomials providing a convenient and intuitive way to visualize the underlying algebraic structure.

An implementation of the PCPO and TPCPO hierarchies based on the approach developed here is available as a Julia package \cite{pcpo_implementation_package}. A detailed description of the implementation, together with numerical benchmarks and comparisons with existing approaches, is provided in \cite{pcpo_implementation}.

Several directions remain open. A first natural extension is to incorporate in the quotient approach additional algebraic constraints that frequently arise in quantum information, such as projective, involutive, or unitary-type relations (e.g., $x^2=x$, $x^2=1$, $x^*x=1$). A second direction is to develop Gröbner-basis-like methods directly in the partially commutative setting to handle additional families of polynomial equalities on top of partial commutation. Progress on these questions would further improve the efficiency and robustness of SDP hierarchies in non-commutative optimization.

\section*{Acknowledgements}
We acknowledge funding from the European Union’s Horizon Europe research and innovation programme under the project `Quantum Security Networks Partnership" (QSNP, grant agreement No 101114043), 
from the F.R.S-FNRS through the PDR T.0171.22,
from the FWO and F.R.S.-FNRS under the Excellence of Science (EOS) programme project 40007526,
from the FWO through the BeQuNet SBO project S008323N. 
S.P. is a Research Director of the Fonds de la Recherche Scientifique – FNRS.

Funded by the European Union. Views and opinions expressed are however those of the authors only and do not necessarily reflect those of the European Union. Neither the European Union nor the granting authority can be held responsible for them.

% Bibliography
\bibliographystyle{alpha}
\bibliography{refs}

\end{document}